\documentclass[12pt,american,english]{article}
\usepackage{mathptmx}
\usepackage[T1]{fontenc}
\usepackage[utf8]{inputenc}
\usepackage{geometry}
\usepackage{color}
\usepackage{babel}
\usepackage{amsmath}
\usepackage{amsthm}
\usepackage{amssymb}
\usepackage{graphicx}
\usepackage{setspace}
\usepackage{wasysym}
\usepackage[authoryear]{natbib}
\usepackage[unicode=true,
 bookmarks=true,bookmarksnumbered=false,bookmarksopen=false,
 breaklinks=true,pdfborder={0 0 0},pdfborderstyle={},backref=false,colorlinks=true]
 {hyperref}
\hypersetup{pdftitle={Mislearning from Censored Data: The Gambler's Fallacy and Other Correlational Mistakes in Optimal-Stopping Problems},
 pdfauthor={Kevin He},
 pdfpagelayout=OneColumn, pdfnewwindow=true, pdfstartview=XYZ, plainpages=false, urlcolor=[rgb]{0.0430 ,0, 0.5}, linkcolor=[rgb]{0.0430 ,0, 0.5}, citecolor=[rgb]{0.0430 ,0, 0.5}, hypertexnames=false}

\makeatletter

\providecommand{\tabularnewline}{\\}

\theoremstyle{definition}
\newtheorem{defn}{\protect\definitionname}
\theoremstyle{plain}
\newtheorem{prop}{\protect\propositionname}
\theoremstyle{plain}
\newtheorem{lem}{\protect\lemmaname}
\theoremstyle{definition}
 \newtheorem{example}{\protect\examplename}

\setcitestyle{round}

\usepackage{lmodern}
\usepackage{bibunits}
\renewcommand{\@bibunitname}{\jobname.\the\@bibunitauxcnt}

\makeatother

\addto\captionsamerican{\renewcommand{\definitionname}{Definition}}
\addto\captionsamerican{\renewcommand{\examplename}{Example}}
\addto\captionsamerican{\renewcommand{\lemmaname}{Lemma}}
\addto\captionsamerican{\renewcommand{\propositionname}{Proposition}}
\addto\captionsenglish{\renewcommand{\definitionname}{Definition}}
\addto\captionsenglish{\renewcommand{\examplename}{Example}}
\addto\captionsenglish{\renewcommand{\lemmaname}{Lemma}}
\addto\captionsenglish{\renewcommand{\propositionname}{Proposition}}
\providecommand{\definitionname}{Definition}
\providecommand{\examplename}{Example}
\providecommand{\lemmaname}{Lemma}
\providecommand{\propositionname}{Proposition}

\begin{document}
\title{\vspace*{-40bp}
{\normalsize{}
\setcounter{page}{0}
\interfootnotelinepenalty=10000}Mislearning from Censored Data: \\ The Gambler's Fallacy and Other
Correlational Mistakes in Optimal-Stopping Problems}
\author{Kevin He\thanks{University of Pennsylvania. Email: \protect\href{mailto:hesichao\%40gmail.com}{hesichao@gmail.com}.
I am indebted to Drew Fudenberg, Matthew Rabin, Tomasz Strzalecki,
and Ben Golub for their guidance and support. I thank Rani Spiegler
and anonymous referees, Isaiah Andrews, Ruiqing Cao, In-Koo Cho, Martin
Cripps, Krishna Dasaratha, Jetlir Duraj, Ben Enke, Ignacio Esponda,
Jiacheng Feng, Mira Frick, Tristan Gagnon-Bartsch, Ashvin Gandhi,
Oliver Hart, Johannes H\"{o}rner, Alice Hsiaw, Ryota Iijima, Yuhta
Ishii, Lawrence Jin, Yizhou Jin, Michihiro Kandori, Max Kasy, Shengwu
Li, Jonathan Libgober, Matthew Lilley, George Mailath, Eric Maskin,
Weicheng Min, Xiaosheng Mu, Andy Newman, Harry Pei, Joshua Schwartzstein,
Roberto Serrano, Philipp Strack, Elie Tamer, Omer Tamuz, Michael Thaler,
Linh T. Tô, Maria Voronina, Yuichi Yamamoto, and seminar participants
for their insightful comments. I thank the California Institute of
Technology for hospitality when some of the work on this paper was
completed.}}
\date{{\normalsize{}}%
\begin{tabular}{rl}
First version: & March 21, 2018\tabularnewline
This version: & August 18, 2021\tabularnewline
\end{tabular}}

\maketitle
\vspace*{-20bp}

\begin{abstract}
I study endogenous learning dynamics for people who misperceive intertemporal
correlations in random sequences. Biased agents face an optimal-stopping
problem. They are uncertain about the underlying distribution and
learn its parameters from predecessors. Agents stop when early draws
are “good enough,” so predecessors’ experiences contain negative streaks
but not positive streaks. When agents wrongly expect systematic reversals
(the ``gambler's fallacy''), they understate the likelihood of consecutive
below-average draws, converge to over-pessimistic beliefs about the
distribution’s mean, and stop too early. Agents uncertain about the
distribution’s variance overestimate it to an extent that depends
on predecessors’ stopping thresholds. I also analyze how other misperceptions
of intertemporal correlation interact with endogenous data censoring.

\vspace*{\medskipamount}

\thispagestyle{empty}
\end{abstract}

\section{\label{sec:Introduction}Introduction}

When a fair coin lands on tails three times in a row, many people
wrongly expect the same coin to have an increased chance of landing
on heads on the next toss to ``balance things out.'' This mistaken
belief stems from a widespread statistical bias called the \emph{gambler's
fallacy}, where people expect too much reversal from sequential realizations
of independent random events. Studies have documented the gambler's
fallacy in settings where it is strictly costly, such as lotteries
with pari-mutuel payouts \citep*{terrell1994test,suetens2016predicting}
and incentivized lab experiments \citep*{benjamin2017biased}. The
same bias also affects experienced decision-makers in high-stakes
environments, including immigration judges \citep*{chen2016decision}
and MBA admissions interviewers \citep{simonsohn2013daily}.

The gambler's fallacy affects people's behavior and beliefs in optimal-stopping
problems, an important class of economic environments where agents
act on sequential signal realizations. For instance, \citet*{mueller2018job}
use survey data to document beliefs consistent with the gambler's
fallacy in job search, finding that job seekers' perceived probability
of becoming employed within the next few months \emph{increases} over
the course of the unemployment spell. In settings like this, how does
the bias affect society's long-run beliefs about the economic fundamentals
(e.g., the labor market conditions) and how does it influence agents'
behavior? These questions are challenging because the biased agents
do not passively observe an exogenous data stream, but take stopping
actions that censor the observation of future signal realizations.
The stopping decisions, in turn, depend on the agents' (possibly mistaken)
beliefs about the fundamentals.

In this paper, I study novel implications of the gambler's fallacy
and other correlational mistakes in optimal-stopping problems when
a society of biased agents learn about the underlying distributions.
Agents take turns playing the same stage game: an optimal-stopping
problem with draws in different periods generated from fixed but unknown
distributions. Agents learn about the means of the distributions from
experience, but start with a dogmatic and wrong belief about the correlation
between the draws. For instance, when the draws are objectively independent
but agents expect the draws to exhibit reversals conditional on the
means, they suffer from the gambler's fallacy. I show the non-self-confirming
steady state of misspecified Bayesian learning in this environment
involves distorted beliefs about the marginal distributions and suboptimal
stopping behavior, and the directions of these errors depend on details
of the correlational mistake. I derive further results about how changes
in the stage game affect long-run learning outcomes and how additional
uncertainty about the variances of the distributions interacts with
stopping incentives.

To illustrate the main mechanism behind these results, consider as
a running example HR managers who suffer from the gambler's fallacy.
Each manager sequentially interviews candidates for a single job opening
and exaggerates how unlikely it is to get consecutive above-average
or consecutive below-average applicants (relative to the labor pool
mean). This error stems from the same psychology that leads people
to exaggerate how unlikely it is to get consecutive heads or consecutive
tails when tossing a fair coin. Evidence from MBA admissions suggests
this bias can have a sizable effect on sequential interviews: following
applicants who are one standard deviation worse than usual, interviewers
expect the next candidate to exceed average quality by the equivalent
of two years of work experience \citep{simonsohn2013daily}.

Suppose the managers are initially uncertain about the labor pool
quality and collectively learn about this fundamental over time. Every
manager is responsible for hiring in a different year. Each junior
manager consults with senior managers and adopts their beliefs about
the labor pool based on their recruiting experience for similar positions
in the past. The junior manager then implements a stopping strategy
for her own recruiting problem, updates her belief at the end of the
hiring season, and shares this new belief with her successors.\footnote{This environment where managers pass down their beliefs is equivalent
to biased managers updating their beliefs using all past managers'
hiring experience.} How does the gambler's fallacy influence the managers’ beliefs and
behavior in the long run?

In this example, agents tend to stop when early draws are deemed ``good
enough,'' causing an asymmetric truncation of experience. When a
manager discovers a sufficiently strong candidate early in the hiring
cycle, she stops her recruitment efforts and does not observe what
additional candidates would have been found for the same job opening
with a longer search. This endogenous \emph{censoring effect} on histories
interacts with the gambler's fallacy bias and generates pessimistic
inference about the labor pool. Managers continue searching only when
their early candidates are below-average. They misinterpret subsequent
above-average candidates as the expected positive reversal after bad
initial outcomes, not as strong signals about the labor pool. On the
other hand, they are surprised by subsequent below-average candidates
since their bias leads them to understate the likelihood of bad streaks,
misreading consecutive bad draws as very strong negative signals about
the pool. That is, after bad early draws, managers under-infer from
subsequent good draws but over-infer from subsequent bad draws. On
average, they communicate an over-pessimistic impression of the labor
pool to future junior managers. This pessimism informs the junior
managers' stopping strategy and affects the kind of censored history
they observe and the new beliefs they pass down to their own successors.

The key mechanism behind my results is the \emph{interaction} between
psychological bias and data censoring in stopping problems. Neither
is dispensable. Agents who do not suffer from correlational mistakes
learn the fundamentals correctly even from censored histories. Conversely,
in an environment without censoring where agents observe ex-post what
would have been drawn in each period of the optimal-stopping problem,
even biased agents learn the fundamentals correctly. In particular,
the gambler's fallacy is a ``symmetric'' bias; the ``asymmetric''
learning outcome of over-pessimism only obtains when the bias interacts
with an (endogenous) asymmetric censoring mechanism that tends to
produce data containing negative streaks but not positive streaks.
More broadly, the selective censoring of sequential signals represents
a natural source of data endogeneity whose impact on different biases
remains understudied.

The misinference mechanism central to this paper implies novel comparative
statics predictions about how the economic environment affects learning
outcomes under the gambler's fallacy. Returning to \citet*{mueller2018job}'s
context of job seekers, my results suggest that government policies
subsidizing longer search, such as extended unemployment insurance,
help mitigate belief distortions for job seekers who commit the gambler's
fallacy. This is because such policies lead agents to use higher acceptance
thresholds and generate less censored histories, which in turn induce
less pessimistic beliefs for their successors. Comparative statics
of this sort are unique to a setting where biased agents learn from
endogenously censored histories — changing the stage game has no effect
on the long-run learning outcomes if data is exogenous or if agents
are correctly specified.

Finally, I extend the analysis for the case of the gambler's fallacy
by considering uncertainty about both the means and variances of the
distributions. In this joint estimation, agents misinfer means by
the same amounts as in the baseline model and exaggerate variances.
The idea is that agents attribute streaks of good or bad draws to
``noise.'' The degree of belief in this\emph{ fictitious variation}
both depends on the severity of history censoring (as the amount of
``noise'' inferred depends on the kind of data) and influences the
agents' stopping strategy (as higher variance encourages continuing
in search problems due to option value). To illustrate how this belief
in fictitious variation interacts with endogenous learning, I show
that a society where agents are uncertain about the variances end
up with a less distorted long-run belief about the means than another
society where agents know the correct variances. This is despite the
fact that agents in both societies would make the same (mis)inference
about the means when given the same data.

The rest of the paper is organized as follows. Section \ref{sec:model}
presents the model and discusses the modeling assumptions. The model
is general enough to capture various misperceptions of intertemporal
correlation, with the gambler's fallacy as a special case. Section
\ref{sec:results} analyzes the steady state of learning and contains
the main results of the paper. Section \ref{sec:Convergence} proves
the convergence of misspecified learning dynamics to the steady state.
Section \ref{sec:Related-Literature} discusses related theoretical
literature. Section \ref{sec:Concluding-Remarks} concludes.

\section{\label{sec:model}Model}

\subsection{The Objective Environment}

The stage game is a two-period optimal-stopping problem. In the first
period, the agent draws $x_{1}\in\mathbb{R}$ and decides whether
to stop. If she stops, her payoff is $u_{1}(x_{1})=x_{1}$ and the
stage game ends. If she continues, she incurs a cost $\kappa\in\mathbb{R}$,
enters the second period, then draws $x_{2}\in\mathbb{R}$. (This
$\kappa$ may also be negative, a subsidy for continuing.) There is
probability $0\le q<1$ that the first draw can be recalled in the
second period and the agent can pick the best of the two draws, but
with complementary probability the first draw is no longer available.
So the agent's expected payoff from continuing, conditional on the
draws, is $u_{2}(x_{1},x_{2})=q\cdot\max(x_{1},x_{2})+(1-q)x_{2}-\kappa$.
Both $q$ and $\kappa$ are known parameters.

This stage game fits a number of economic situations:
\begin{itemize}
\item Many industries have an annual hiring cycle. Consider a firm in such
an industry and an HR manager who must fill a job opening during this
year's cycle. In the early phase of the hiring cycle, she finds a
candidate with quality $x_{1}$. She must decide between hiring this
candidate immediately or waiting. Waiting lets her continue searching
in the late phase of the cycle, but carries the risk that the early
candidate accepts an offer from a different firm in the interim.
\item A homeowner lists his house for sale and receives an offer in each
period. The homeowner must decide whether to accept the first offer
he gets and take his house off the market, or to wait for the second
offer, incurring a waiting cost and risking the first buyer leaving
the market.
\item An unemployed worker searches for jobs. While unemployed, she receives
a job offer in each period and decides whether to continue her job
search. Once she becomes employed, she stops searching and no longer
receives further offers.
\end{itemize}
The draws $x_{1},x_{2}$ are the realizations of two possibly correlated
Gaussian random variables $X_{1},X_{2}$, with unconditional means
$\mu_{1}^{\bullet},\mu_{2}^{\bullet}\in\mathbb{R}$. We have $X_{1}=\mu_{1}^{\bullet}+\epsilon_{1}$
and $X_{2}=\mu_{2}^{\bullet}+\epsilon_{2},$ where $\epsilon_{1}\sim\mathcal{N}(0,\sigma^{2})$
and $(\epsilon_{2}\mid\epsilon_{1})\sim\mathcal{N}(-r\epsilon_{1},\sigma^{2})$
for some fixed value of $r\in\mathbb{R}$. The parameters $\mu_{1}^{\bullet},\mu_{2}^{\bullet}\in\mathbb{R}$
are the \emph{true fundamentals} that stand for the average qualities
of the two pools in the two periods. (In general we may have $\mu_{1}^{\bullet}\ne\mu_{2}^{\bullet}$.
For instance, this might happen due to dynamic adverse selection in
the labor pool over time in the example of the HR manager.) The $\epsilon_{1},\epsilon_{2}$
terms represent the idiosyncratic factors that determine how the agent's
actual draws deviate from the average qualities of the respective
pools, with $r$ the\emph{ true reversal parameter. }When $r>0,$
the idiosyncratic factors that lead to an unusually good first draw
relative to the early pool quality also portend a below average second
draw. (Such reversals may happen, for instance, if the agent is exhausting
a small pool.) Note that $X_{1},X_{2}$ are independent when $r=0,$
negatively correlated when $r>0$, and positively correlated when
$r<0.$

\subsection{Gambler's Fallacy and Other Correlational Mistakes}

I introduce a general model of misperceptions of intertemporal correlation,
with the gambler's fallacy as a special case. Section \ref{sec:results}
will both analyze how different kinds of correlational mistakes interact
with endogenous data censoring, and present more in-depth results
that focus on the gambler's fallacy. 

Agents are uncertain about both the fundamentals and the reversal
parameter. They  believe that if the average qualities of the pools
are $\mu_{1},\mu_{2}\in\mathbb{R}$, then the draws are generated
by $X_{1}=\mu_{1}+\epsilon_{1},$ $X_{2}=\mu_{2}+\epsilon_{2}$ with
$\epsilon_{1}\sim\mathcal{N}(0,\sigma^{2}),$ $(\epsilon_{2}\mid\epsilon_{1})\sim\mathcal{N}(-\gamma\epsilon_{1},\sigma^{2})$
for some unknown $\gamma\in[\gamma_{l},\gamma_{h}]$. If $0=r<\gamma_{l}$,
 then the agents suffer from the gambler's fallacy. This may represent
a superstitious belief in an environment where the two draws are objectively
independent that if someone gets lucky on the first draw, then bad
luck is ``due'' to befall them in the near future. More generally,
when $r<\gamma_{l}$ (but $r$ may not be 0), agents exaggerate the
amount of reversal in the idiosyncratic factors across the draws.
On the other hand, we may also have $r>\gamma_{h},$ in which case
agents dogmatically underestimate the amount of reversal. This might
be called a form of ``hot-hand fallacy,'' where following a ``lucky''
first draw agents systematically overestimate the chance of another
good draw (and symmetrically for bad draws).\footnote{\citet{rabin2010gambler} propose a different mechanism for the hot-hand
fallacy: agents expect reversals (not streaks) conditional on the
fundamentals, but misinfer fundamentals. This also leads agents to
predict that streaks will continue.}

Denote by $\phi(\cdot\mid\mu)$ the Gaussian density with mean $\mu$
and variance $\sigma^{2}$, and let $\Psi(\mu_{1},\mu_{2};\gamma)$
refer to the joint distribution $X_{1}=\mu_{1}+\epsilon_{1}$, $X_{2}=\mu_{2}+\epsilon_{2}$
with $\epsilon_{1}\sim\phi(\cdot\mid0),$ $(\epsilon_{2}\mid\epsilon_{1})\sim\phi(\cdot\mid-\gamma\epsilon_{1})$.
Agents believe the joint distribution of $(X_{1},X_{2})$ is described
by one of the \emph{feasible models}, $\{\Psi(\mu_{1},\mu_{2};\gamma):(\mu_{1},\mu_{2})\in\mathbb{R}^{2},\gamma\in[\gamma_{l},\gamma_{h}]\}.$
If $r\notin[\gamma_{l},\gamma_{h}],$ then the set of feasible models
excludes the true model, $\Psi^{\bullet}:=\Psi(\mu_{1}^{\bullet},\mu_{2}^{\bullet};r)$,
so Bayesian updating within the class of feasible models amounts to
misspecified learning. I use misspecification as a tool to represent
and study the gambler's fallacy and other correlational mistakes.

Throughout, I maintain the assumption that $r,\gamma_{l},\gamma_{h}\ne-1.$
It turns out that for the model $\Psi(\mu_{1},\mu_{2};-1)$ with any
$\mu_{1},\mu_{2},$ all stopping strategies are optimal. So I rule
out this knife-edge case by assuming that neither the true reversal
parameter nor one of the end points of $[\gamma_{l},\gamma_{h}]$
is exactly equal to $-1.$ I still allow the case that the interval
of subjectively feasible reversal parameters contains $-1$ in its
interior. Finally, denote $\gamma_{n}:=\arg\min_{\gamma\in[\gamma_{l},\gamma_{h}]}|\gamma-r|$
as the nearest point in the interval $[\gamma_{l},\gamma_{h}]$ to
$r.$ Note that if $r\in[\gamma_{l},\gamma_{h}],$ then the nearest
point is $\gamma_{n}=r$ itself. Otherwise, $\gamma_{n}$ is one of
the end points, $\gamma_{l}$ or $\gamma_{h}$.

\subsection{The Steady State}

Suppose a sequence of agents arrive one per round $(t=1,2,3,...)$
and take turns playing the stage game. All agents have the same set
of reversal parameters $[\gamma_{l},\gamma_{h}]$ that they find plausible.
They face the same but unknown objective pool qualities $(\mu_{1}^{\bullet},\mu_{2}^{\bullet})$
and true reversal parameter $r.$ At the end of each round $t$, the
$t$-th agent updates her belief about qualities and about the reversal
parameter using her experience, then communicates her updated belief
to her successor. The successor acts based on the inherited belief,
then passes down an updated belief at the end of the round to his
own successor, and so forth. I now define the steady state of this
learning system.

Roughly speaking, a\emph{ steady state }of the system consists of
a strategy $S^{\infty}:\mathbb{R}\to\{\text{Stop, Continue}\}$ that
maps the realization of the first draw $X_{1}=x_{1}$ into a stopping
decision, and point-mass beliefs about the pool qualities and the
reversal parameter, $(\mu_{1}^{\infty},\mu_{2}^{\infty},\gamma^{\infty})\in\mathbb{R}^{2}\times[\gamma_{l},\gamma_{h}]$,
so that: (i) agents find it optimal to follow strategy $S^{\infty}$
given beliefs $(\mu_{1}^{\infty},\mu_{2}^{\infty},\gamma^{\infty})$;
(ii) $(\mu_{1}^{\infty},\mu_{2}^{\infty},\gamma^{\infty})$ are the
``best-fitting'' beliefs about the pool qualities and the reversal
parameter given data generated from the strategy $S^{\infty}$. The
steady state corresponds to \citet{esponda2016berk}'s Berk-Nash equilibrium
adapted to the current setting.

To make precise the meaning of ``best-fitting'' beliefs for misspecified
learners, the \emph{history} of the stage game is an element $h\in\mathbb{H}:=\mathbb{R}\times(\mathbb{R}\cup\{\varnothing\})$.
If an agent decides to stop after $X_{1}=x_{1}$, her history is $(x_{1},\varnothing)$.
If an agent continues after $X_{1}=x_{1}$ and gets a second draw
$X_{2}=x_{2}$, her history is $(x_{1},x_{2})$. The symbol $\varnothing$
is a\emph{ censoring indicator}, emphasizing if the agent stops, then
the counterfactual second draw that she would have found had she continued
remains unobserved.

Consider the strategy $S$ and the parameters $(\mu_{1},\mu_{2},\gamma)$.
The agent's subjective likelihood of the history $h=(x_{1},x_{2})$
with $S(x_{1})=\text{Continue}$ is $\phi(x_{1}\mid\mu_{1})\cdot\phi(x_{2}\mid\mu_{2}-\gamma(x_{1}-\mu_{1}))$,
while that of the history $h=(x_{1},\varnothing)$ with $S(x_{1})=\text{Stop}$
is $\phi(x_{1}\mid\mu_{1})$. Let $(\mu_{1}^{*}(S),\mu_{2}^{*}(S),\gamma^{*}(S))\in\mathbb{R}^{2}\times[\gamma_{l},\gamma_{h}]$
be the \emph{pseudo-true parameters }with respect to $S$ that maximize
the expected log-likelihood of the agent's history, with the expectation
taken over the true distribution of histories generated by $S$. Intuitively
speaking, these correspond to the long-run inferences about the fundamentals
and the reversal parameter when a large sample of histories is generated
using the stopping strategy $S$.

Equivalently, the pseudo-true parameters minimize the \emph{KL divergence}
between the expected and the objective distributions over histories.
Let $\mathcal{H}(\Psi(\mu_{1},\mu_{2};\gamma);S)$ refer to the distribution
of histories when the draws have the joint distribution $\Psi(\mu_{1},\mu_{2};\gamma)$
and histories are censored according to the strategy $S$. The true
distribution of histories given strategy $S$ is $\mathcal{H}(\Psi(\mu_{1}^{\bullet},\mu_{2}^{\bullet};r);S)$,
which I abbreviate as $\mathcal{H}^{\bullet}(S)$. To avoid trivialities,
I will focus on steady states where agents continue with positive
probability (otherwise their beliefs are not disciplined by the observation
of any second-period draws), that is to say strategies $S$ where
$S(x_{1})=\text{Continue}$ for a positive Lebesgue measure of $x_{1}\in\mathbb{R}.$
For such an $S$, the \emph{Kullback-Leibler (KL) divergence} from
$\mathcal{H}^{\bullet}(S)$ to $\mathcal{H}(\Psi(\mu_{1},\mu_{2};\gamma);S)$,
denoted by $D_{KL}(\mathcal{H}^{\bullet}(S)\ ||\ \mathcal{H}(\Psi(\mu_{1},\mu_{2};\gamma);S)\ )$,
is {\small{}
\begin{align}
 & \int_{x_{1}\in S^{-1}(\text{Stop})}\phi(x_{1}\mid\mu_{1}^{\bullet})\cdot\ln\left(\frac{\phi(x_{1}\mid\mu_{1}^{\bullet})}{\phi(x_{1}\mid\mu_{1})}\right)dx_{1}\nonumber \\
 & +\int_{x_{1}\in S^{-1}(\text{Cont.})}\left\{ \int_{-\infty}^{\infty}\begin{array}{c}
\phi(x_{1}\mid\mu_{1}^{\bullet})\cdot\phi(x_{2}\mid\mu_{2}^{\bullet}-r(x_{1}-\mu_{1}))\\
\cdot\ln\left[\frac{\phi(x_{1}\mid\mu_{1}^{\bullet})\cdot\phi(x_{2}\mid\mu_{2}^{\bullet}-r(x_{1}-\mu_{1}))}{\phi(x_{1}\mid\mu_{1})\cdot\phi(x_{2}\mid\mu_{2}-\gamma(x_{1}-\mu_{1}))}\right]
\end{array}dx_{2}\right\} dx_{1}.\label{eq:KL}
\end{align}
}{\small\par}

So the KL divergence in Equation (\ref{eq:KL}) is the expected log-likelihood
ratio of the history under the true process versus under the model
$\Psi(\mu_{1},\mu_{2};\gamma)$, where expectation over histories
is taken under the true process. In general, this optimization objective
depends on the stopping strategy $S$. It is simple to see that the
minimizers of KL divergence are the same as the maximizers of expected
log-likelihood of the history.

I formalize the definition of a steady state:
\begin{defn}
\label{def:A-steady-state}A \emph{steady state}\textbf{ }consists
of $\mu_{1}^{\infty},\mu_{2}^{\infty}\in\mathbb{R}$, $\gamma^{\infty}\in[\gamma_{l},\gamma_{h}],$
and a strategy $S^{\infty}$ such that: (i) $S^{\infty}$ continues
with positive probability and is optimal among all stopping strategies
for the model $\Psi(\mu_{1}^{\infty},\mu_{2}^{\infty};\gamma^{\infty})$;
(ii) $\mu_{1}^{\infty}=\mu_{1}^{*}(S^{\infty})$, $\mu_{2}^{\infty}=\mu_{2}^{*}(S^{\infty})$,
$\gamma^{\infty}=\gamma^{*}(S^{\infty})$.
\end{defn}
The steady state is not a self-confirming equilibrium. There is positive
KL divergence between the true data distribution in the steady state
and the data distribution under $\Psi(\mu_{1}^{\infty},\mu_{2}^{\infty};\gamma^{\infty})$,
so even the best-fitting beliefs do not perfectly explain the data.
To see this, consider the special case of $r=0,$ $\gamma_{n}>0.$
Objectively, the conditional distribution $X_{2}\mid(X_{1}=x_{1})$
has a mean of $\mu_{2}^{\bullet}$ for every $x_{1}\in\mathbb{R}.$
In the steady state, the biased agents believe the same conditional
distribution has a mean of $\mu_{2}^{\infty}-\gamma_{n}(x_{1}-\mu_{1}^{\infty})$,
which only equals $\mu_{2}^{\bullet}$ for one value of $x_{1}.$
The histories cannot be fully explained by $\Psi(\mu_{1}^{\infty},\mu_{2}^{\infty};\gamma^{\infty})$,
as the predicted conditional distribution $X_{2}\mid(X_{1}=x_{1})$
does not match what is in the data for almost all $x_{1}$ values
where the steady-state strategy chooses to continue. 

We may view the steady state as a stand-alone equilibrium concept
that captures the optimality of behavior given beliefs and the constrained-optimality
of inferences given behavior, in the sense of minimizing KL divergence.
Alternatively, Section \ref{sec:Convergence} provides a Bayesian-learning
foundation for the steady state, in an environment where agents are
not actually solving the KL divergence minimization problem given
in Equation (\ref{eq:KL}), and do not observe any history of the
stage game other than the history they personally experience. In that
setting, Equation (\ref{eq:KL}) is involved in characterizing the
steady state when a sequence of agents each play the stage game once
and pass down their updated Bayesian beliefs to their successors.

\subsection{\label{subsec:Discussion-of-Behavioral}Discussion of Behavioral
Assumptions}

In this paper, the agents' correlational mistake stems from their
dogmatic belief in the interval $[\gamma_{l},\gamma_{h}],$ which
may exclude the true reversal parameter $r$. One story about how
the agents erroneously think $\gamma_{l},\gamma_{h}>0$ in an environment
with $r=0$ (that is, suffer from the gambler's fallacy) relates to
\citet{kahneman1972subjective}'s representativeness heuristic in
judging the likelihoods of random sequences. Objectively, the idiosyncratic
factors $\epsilon_{i}$ (e.g., luck) that govern how draws in different
periods deviate from their respective pool averages are sampled i.i.d.
from a mean-zero distribution. The representativeness heuristic states
that people know certain ``essential characteristics'' of the parent
population generating these idiosyncratic factors (perhaps by observing
their luck in other settings where the fundamentals are known), but
exaggerate the extent to which small samples typically represent these
characteristics. Agents who expect a sample of size two $(\epsilon_{1},\epsilon_{2})$
to approximate the mean-zero property of the parent population of
idiosyncratic factors should believe in a reversal of luck, that is
$\gamma_{l},\gamma_{h}>0$.

This is not a fully detailed and satisfactory microfoundation for
the gambler's fallacy bias, and unfortunately there is limited work
on the origin and persistence of biases in learning contexts. This
literature typically studies the implications of a dogmatically wrong
belief about one parameter on the Bayesian inference about a different
parameter (e.g., \citet*{heidhues2018unrealistic,heidhues2019overconfidence}).
Better understanding why mistakes persist is an important next step.

My setup corresponds to the model of the gambler's fallacy introduced
in \citet{rabin2010gambler}, but applied to a different fundamental
process. \citet{rabin2010gambler} study a setting where a signal
$s_{t}=\theta_{t}+\epsilon_{t}$ is generated each period $t$ around
the fundamental $\theta_{t}$. Objectively $\epsilon_{t}\stackrel{\text{i.i.d.}}{\sim}\mathcal{N}(0,\sigma_{\epsilon}^{2})$,
but agents believe $\epsilon_{t}=\omega_{t}-\alpha\delta\epsilon_{t-1}-\alpha\delta^{2}\epsilon_{t-2}-...$,
for $\omega_{t}\stackrel{\text{i.i.d.}}{\sim}\mathcal{N}(0,\sigma_{\omega}^{2})$
and some $\alpha>0,\delta\in(0,1).$ This specializes to my model
with $r=0$ when there are two periods $t=1,2$, the fundamental process
is $\theta_{t}=\mu_{t}$ for deterministic but unknown $\mu_{1},\mu_{2},$
agents know the variance $\sigma_{\omega}^{2}=\sigma_{\epsilon}^{2},$
and $\gamma_{l}=\gamma_{h}=\alpha\delta$. For \citet{rabin2010gambler},
the fundamentals $(\theta_{t})$ follow an AR(1) process instead of
being deterministic, and they study agents who exogenously observe
all signals and estimate the long-run mean and persistence of the
fundamental process. I study a different environment with endogenous
data where agents' stopping decisions censor the observation of future
signals.

\section{\label{sec:results}Steady-State Results}

\subsection{Inference about Parameters from Censored Data}

A\emph{ cutoff strategy} is a strategy $S$ whose stopping region
$S^{-1}(\text{Stop})$ is either $[c,\infty)$ for some $c\in\mathbb{R}\cup\{\infty\}$
or $(-\infty,c]$ for some $c\in\mathbb{R}\cup\{-\infty\}$. The next
proposition provides a closed-form expression for the pseudo-true
parameters as a function of the cutoff threshold $c$ in a cutoff
strategy $S$. This result can be thought of as a one-sided benchmark
of how biased learners misinfer the fundamentals and the reversal
parameter using data censored at an exogenously given threshold. The
subsequent steady-state analysis considers stopping strategies that
best respond to the beliefs they induce. All proofs appear in the
Appendix.
\begin{prop}
\label{prop:pseudo_true_normal}For any strategy $S$ that continues
with positive probability, $\mu_{1}^{*}(S)=\mu_{1}^{\bullet}$, $\gamma^{*}(S)=\gamma_{n}$.
If $S$ is a cutoff strategy that stops when $x_{1}\ge c$ for some
$c\in\mathbb{R}\cup\{\infty\},$ then $\mu_{2}^{*}(c)=\mu_{2}^{\bullet}+(r-\gamma_{n})\cdot(\mu_{1}^{\bullet}-\mathbb{E}[X_{1}\mid X_{1}\le c])$.
If $S$ is a cutoff strategy that stops when $x_{1}\le c$ for some
$c\in\mathbb{R}\cup\{-\infty\},$ then $\mu_{2}^{*}(c)=\mu_{2}^{\bullet}+(r-\gamma_{n})\cdot(\mu_{1}^{\bullet}-\mathbb{E}[X_{1}\mid X_{1}\ge c]).$
\end{prop}
Proposition \ref{prop:pseudo_true_normal} shows that the misinference
phenomenon requires both data censoring and the correlational mistake.
Even biased agents with $r\notin[\gamma_{l},\gamma_{h}]$ correctly
estimate the fundamentals in the absence of censoring (i.e., under
the strategy $S$ that never stops). Conversely, agents whose prior
belief does not contain a dogmatic correlational mistake (i.e., when
$r\in[\gamma_{l},\gamma_{h}]$) end up with correct beliefs about
the fundamentals for any level of censoring.

Whether biased agents with $r\notin[\gamma_{l},\gamma_{h}]$ will
hold over-pessimistic or over-optimistic beliefs about the fundamentals
depends on the direction of their correlational mistake and the direction
of data censoring. When $r-\gamma_{n}<0$ and the strategy stops for
high values of $X_{1},$ and when $r-\gamma_{n}>0$ and the strategy
stops for low values of $X_{1},$ agents have over-pessimistic beliefs.
When $r-\gamma_{n}<0$ and the strategy stops for low values of $X_{1},$
and when $r-\gamma_{n}>0$ and the strategy stops for high values
of $X_{1},$ agents have over-optimistic beliefs. In all cases, more
severe censoring (i.e., a cutoff strategy that stops for more realizations
of $X_{1}$) exacerbates the belief distortion. Details of the intertemporal
correlation misperception interact with the region of selective censoring
to determine agents' long-run beliefs.

Turning to our main application, when agents exaggerate reversals
$r-\gamma_{n}<0$ and observe data generated from a cutoff rule that
stops for high $X_{1}$ (e.g., stop searching if and only if the early
candidate's quality is higher than some $c$), they have over-pessimistic
beliefs about $\mu_{2}$ and their beliefs decrease without bound
as the stopping threshold $c$ decreases. I will use this application
to explain why directional data censoring leads to belief distortions
for biased learners.

Suppose $r=0$ and $\gamma_{n}>0$. Under the gambler's fallacy, the
expected realization of $X_{2}$ depends on two factors: the second-period
pool quality $\mu_{2},$ and a reversal effect based on the realization
of $X_{1}$. The society of biased agents who stop for low values
of $X_{1}$ cannot end up with a correct or over-optimistic belief
about $\mu_{2}$, else they would be systematically disappointed by
the realizations of $X_{2}$ in their own histories in an environment
where $X_{1},X_{2}$ are objectively independent. This is because
the second draw is only observed when the first draw's quality is
low enough, a contingency that leads biased agents to expect positive
reversal on average. The long-run beliefs of the agents thus feature
two mistakes partially canceling each other out to better fit the
data, as their pessimism about the quality of the late-phase pool
counteracts their false expectation of positive reversals when the
first draw is bad enough to be rejected.

The severity of the biased agents' pessimism increases with the severity
of censoring. The intuition is that the bias leads agents to infer
a lower $\mu_{2}^{*}$ to better match $X_{2}$'s in histories that
start with bad $X_{1}$'s, but doing so carries the cost of a worse
model fit for histories that start with intermediate $X_{1}$'s. More
severe censoring — generated by a strategy that stops not only after
the very good early early draws but also after the intermediate ones
— alleviates this cost, as histories that start with intermediate
$X_{1}$'s no longer contain their associated $X_{2}$'s. The extra
censoring thus decreases the optimal inference $\mu_{2}^{*}$.

The agents jointly estimate the reversal parameter $r$ and the fundamentals
$\mu_{1}^{\bullet},\mu_{2}^{\bullet}$. Proposition \ref{prop:pseudo_true_normal}
says that agents always end up believing the nearest feasible parameter
$\gamma_{n}$ to the true reversal parameter $r$. To gain some geometric
intuition for this result, view the agents' inference problem as using
a scatter plot of $(x_{1},x_{2})$ data points to estimate a conditional
expectation, $\mathbb{E}[X_{2}\mid X_{1}=x_{1}]$. This conditional
expectation is a linear function in $x_{1}$ with a slope of $-\gamma$
and an intercept determined by $\mu_{2}$. The conditional expectation
in the true data-generating process has the slope $-r$. The agent
is free to infer any intercept, but must pick a slope such that $\gamma\in[\gamma_{l},\gamma_{h}].$
Geometrically speaking, the best-fitting regression line will have
the slope $-\gamma_{n}$. A line with a slope as close as possible
to the data-generating slope and the best-fitting intercept given
this slope will better describe the data points than a line with any
other feasible slope and any other intercept.

Proposition \ref{prop:pseudo_true_normal} also tells us that the
quality of the early pool is always correctly estimated with any stopping
strategy. This is because the first draw's quality $X_{1}$ is always
observed, and $\mu_{1}^{*}=\mu_{1}^{\bullet}$ provides the best fit
for the first-period data. The agents cannot improve the fit of second-period
data by distorting their inference about the early pool: for any reversal
parameter $\gamma$, fundamentals $(\mu_{1}^{'},\mu_{2})$ and $(\mu_{1}^{\bullet},\mu_{2}-\gamma(\mu_{1}^{\bullet}-\mu_{1}^{'}))$
generate the same conditional distributions of $X_{2}\mid(X_{1}=x_{1})$
for any realization $x_{1}$. Any distortion of the inference about
early pool from $\mu_{1}^{\bullet}$ to $\mu_{1}^{'}$ to better explain
$X_{2}$ data can be equivalently done by keeping $\mu_{1}^{*}=\mu_{1}^{\bullet}$
and shifting $\mu_{2}^{*}$ by $-\gamma(\mu_{1}^{\bullet}-\mu_{1}^{'})$.
There is no trade-off between fitting $X_{1}$ and fitting $X_{2}$,
so the agents correctly infer $\mu_{1}^{\bullet}$ to provide the
best fit for the early-pool mean.

\citet*{mueller2018job} report in their Figure 3 that very recently
unemployed workers underestimate their probability of finding a job
in the next three months. This is consistent with Proposition \ref{prop:pseudo_true_normal}'s
prediction of ex-ante pessimistic beliefs at the start of the search,
in a world where people suffer from the gambler's fallacy and accept
early draws (i.e., job offers) that are sufficiently good.

\subsection{Steady-State Stopping Behavior}

In this section, I turn to behavior in the steady state. In the main
application of the gambler's fallacy $(r=0,$ $\gamma_{n}>0)$, we
know from Proposition \ref{prop:pseudo_true_normal} that agents end
up with over-pessimistic beliefs about $\mu_{2}$ if they infer from
histories that are censored when $X_{1}\ge c$ for any threshold $c\in\mathbb{R}.$
But this pessimistic belief does not by itself imply that the misspecified
agents must stop too often compared to a rational agent who knows
the true fundamentals and $r.$ Outside of the steady state, there
is an intuition that an agent with the gambler's fallacy may stop
less often than a rational one, even if the biased agent is over-pessimistic
about $\mu_{2}$. Consider an environment with $r=0,$ $\gamma_{n}>0,$
and suppose the stopping problem satisfies $\kappa=0,q=0$, so there
is no cost of continuing but also no probability of recall. Suppose
the true fundamentals are $\mu_{1}^{\bullet}\gg\mu_{2}^{\bullet}$.
If a biased agent has the correct beliefs about the fundamentals,
she perceives a greater continuation value after $X_{1}=\mu_{2}^{\bullet}$
than a rational agent with the same correct beliefs, since the former
holds a false expectation of positive reversals after a bad (relative
to $\mu_{1}^{\bullet})$ early draw. The rational stopping cutoff
is $c^{\bullet}=\mu_{2}^{\bullet}$ and the rational agent is willing
to stop after $X_{1}=\mu_{2}^{\bullet}$, but the biased agent strictly
prefers to continue after such an early draw and has an indifference
threshold strictly above $c^{\bullet}$. By continuity, the biased
agent's cutoff threshold remains strictly above $c^{\bullet}$ even
under slightly pessimistic beliefs about $\mu_{2}.$

Such ambiguity about behavior disappears in the steady state. The
main result of this section, Proposition \ref{prop:steady_state_example},
compares the \emph{steady-state} stopping behavior of the biased learners
to the objectively optimal thresholds. Towards this result, I begin
with a lemma that characterizes the optimal behavior for an agent
that believes in the model $\Psi(\mu_{1},\mu_{2};\gamma)$, and a
sufficient condition about the existence and uniqueness of the steady
state.
\begin{lem}
\label{lem:behavior_simplified} Consider the model $\Psi(\mu_{1},\mu_{2};\gamma)$
for any $\mu_{1},\mu_{2},\gamma\in\mathbb{R}.$ When $\gamma\ne-1,$
there is a unique cutoff $C(\mu_{1},\mu_{2};\gamma)$ so that the
agent is indifferent between continuing and stopping after $X_{1}=C(\mu_{1},\mu_{2};\gamma)$.
When $\gamma>-1,$ the optimal strategy is to stop when $X_{1}\ge C(\mu_{1},\mu_{2};\gamma)$,
and $\mu_{2}\mapsto C(\mu_{1},\mu_{2};\gamma)$ is strictly increasing.
When $\gamma<-1,$ the optimal strategy is to stop when $X_{1}\le C(\mu_{1},\mu_{2};\gamma)$,
and $\mu_{2}\mapsto C(\mu_{1},\mu_{2};\gamma)$ is strictly decreasing.
\end{lem}
Lemma \ref{lem:behavior_simplified} says the optimal behavior under
the model $\Psi(\mu_{1},\mu_{2};\gamma)$ is a cutoff strategy, and
whether the agent stops after high enough or low enough values of
$X_{1}$ depends on if $\gamma>-1$ or $\gamma<-1.$ To understand
why, note that if the agent thinks $X_{1},X_{2}$ are independent
($\gamma=0$), then she will choose to stop when the realization of
$X_{1}$ is so large that the known payoff from stopping exceeds the
expectation of the uncertain payoff from continuing and drawing an
independent $X_{2}$. But if the agent thinks $X_{1},X_{2}$ are sufficiently
positively correlated ($\gamma<-1)$, then larger realizations of
$X_{1}$ make it even more attractive to continue. In this case, it
is bad realizations of $X_{1}$ that cause the agent to stop, for
the positive correlation makes the agent pessimistic about $X_{2}$
after a bad $X_{1}$.

Suppose $\gamma_{n}>-1$, and consider a simplified setting where
the agents know $\mu_{1}=\mu_{1}^{\bullet}$ and always believe in
$\gamma=\gamma_{n}$. For agents who exaggerate reversals ($r-\gamma_{n}<0$),
there is a positive feedback loop between distorted beliefs and distorted
strategies: a more pessimistic belief about the second-period pool
leads to a lower stopping cutoff by Lemma \ref{lem:behavior_simplified},
and a lower stopping cutoff leads to more pessimistic beliefs by Proposition
\ref{prop:pseudo_true_normal}. On the other hand, for agents who
suffer from the opposite correlational mistake ($r-\gamma_{n}>0$),
there is instead a negative feedback loop: a more pessimistic belief
about $\mu_{2}$ still leads to a lower stopping cutoff, but a lower
stopping cutoff leads to more \emph{optimistic} beliefs by Proposition
\ref{prop:pseudo_true_normal}. \citet*{heidhues2018unrealistic}
show that overconfidence and underconfidence biases in a static effort-choice
problem also lead to positive and negative feedback loops, respectively.
In both environments, reversing the direction of the bias changes
the nature of the feedback cycle between distorted actions and distorted
beliefs.

The next result gives a sufficient condition for the existence and
uniqueness of the steady state.
\begin{prop}
\label{prop:existence_and_uniqueness} There exists a unique steady
state if $|\frac{r-\gamma_{n}}{1+\gamma_{n}}|<1$.
\end{prop}
When $r=0,$ so the draws are objectively independent, Proposition
\ref{prop:existence_and_uniqueness} says a unique steady state exists
under any amount of the gambler's fallacy ($\gamma_{n}>0$), and also
under a moderate amount of the opposite correlational mistake ($-1/2<\gamma_{n}<0$).
In general, a steady state may fail to exist when Proposition \ref{prop:existence_and_uniqueness}'s
condition is violated, as the next example shows.
\begin{example}
\label{exa:non-existence}Suppose $\kappa=0$ and $q=0$ (no cost
of continuing and no probability of recall), and let $\gamma_{l}=\gamma_{h}=0$,
$r=-2,$ $\mu_{1}^{\bullet}=\mu_{2}^{\bullet}=0$. No steady state
exists in this setting. This is because by Lemma \ref{lem:behavior_simplified},
steady-state behavior must involve stopping for $X_{1}\ge c$ for
some $c\in\mathbb{R}.$ In fact, since the agent believes $X_{1},X_{2}$
are independent, she is indifferent between continuing and stopping
if the early draw equals $\mu_{2},$ her belief about the mean of
the second-period draw. Proposition \ref{prop:pseudo_true_normal}
implies her belief $\mu_{2}$ is related to $c$ by $\mu_{2}^{*}(c)=2\cdot\mathbb{E}[X_{1}\mid X_{1}\le c]<0$.
We need to find a $c<0$ such that $c=2\cdot\mathbb{E}[X_{1}\mid X_{1}\le c]$,
which is impossible. Intuitively, the feedback cycle between more
pessimistic beliefs and lower cutoff thresholds is expansionary and
tends to $-\infty.$
\end{example}
As Example \ref{exa:non-existence} hints at, the condition $|\frac{r-\gamma_{n}}{1+\gamma_{n}}|<1$
in Proposition \ref{prop:existence_and_uniqueness} ensures that the
feedback between beliefs and behavior is a contraction map.

Under the condition $|\frac{r-\gamma_{n}}{1+\gamma_{n}}|<1$, the
next result compares the (unique) steady-state cutoff threshold $c^{\infty}$
with the objectively optimal one, $c^{\bullet}$. Of course, by Lemma
\ref{lem:behavior_simplified}, if $r$ and $\gamma_{n}$ are on the
opposite sides of $-1,$ then the comparison of thresholds is meaningless
as the steady-state behavior will have the ``opposite'' kind of
stopping region relative to the optimal behavior. When they are on
the same side of $-1,$ Proposition \ref{prop:steady_state_example}
shows that whether $c^{\infty}<c^{\bullet}$ or $c^{\infty}>c^{\bullet}$
depends on the direction of the correlational mistake.
\begin{prop}
\label{prop:steady_state_example}Suppose $|\frac{r-\gamma_{n}}{1+\gamma_{n}}|<1$,
and suppose either both $r,\gamma_{n}>-1$ or both $r,\gamma_{n}<-1.$
Let $c^{\infty}$ be the cutoff where the steady-state strategy switches
between continuing and stopping, and let $c^{\bullet}$ be switching
cutoff of the objectively optimal strategy. If $r-\gamma_{n}<0,$
then $c^{\infty}<c^{\bullet}.$ If $r-\gamma_{n}>0,$ then $c^{\infty}>c^{\bullet}.$
\end{prop}
Combined with Proposition \ref{prop:pseudo_true_normal} and Lemma
\ref{lem:behavior_simplified}, Proposition \ref{prop:steady_state_example}
tells us the following when $|\frac{r-\gamma_{n}}{1+\gamma_{n}}|<1$:
if $r,\gamma_{n}>-1$ (so that steady-state and optimal strategies
stop after good first-period draws), then $\gamma_{n}>r$ implies
that the agent stops too often and underestimate $\mu_{2}$, while
$\gamma_{n}<r$ implies that the agent stops too rarely and overestimates
$\mu_{2}$. By contrast, if $r,\gamma_{n}<-1$ (so that steady-state
and optimal strategies stop after bad first-period draws), then the
implications of these two biases are reversed.

In particular, when $r=0$ and $\gamma_{n}>0$,  Proposition \ref{prop:steady_state_example}'s
early-stopping conclusion strengthens Proposition \ref{prop:pseudo_true_normal}'s
over-pessimism result. In the steady state, agents must be \emph{sufficiently}
pessimistic as to overcome the opposite intuition about late stopping
under the gambler's fallacy discussed earlier.  To understand the
intuition,  note biased agents believe in different conditional distributions
of $X_{2}$ following different realizations of $X_{1},$ with more
pessimistic beliefs after higher realizations. In a steady state $((\mu_{1}^{\infty},\mu_{2}^{\infty},\gamma_{n}),c^{\infty}),$
the agents' subjective distribution of $X_{2}$ following $X_{1}=c^{\infty}$
must be a leftward shift of the true distribution $\phi(\cdot\mid\mu_{2}^{\bullet})$.
Else, their subjective distributions of $X_{2}$ would stochastically
dominate the true distribution following \emph{all} $x_{1}$ values
in the continuation region, so heuristically they could improve the
fit of their model by lowering their belief about $\mu_{2}$. The
biased agents' indifference at $c^{\infty}$ is thus based on an overly
pessimistic belief about the continuation value, so we must have $c^{\infty}<c^{\bullet}$.

\subsection{Gambler's Fallacy with Independent Draws}

In this section, I derive additional steady-state results for the
main application of agents who suffer from the gambler's fallacy in
an environment with independent $X_{1}$ and $X_{2}$: that is, $r=0$
and $\gamma_{n}>0$.

\subsubsection{Comparative Statics in the Stage Game's Parameters}

How do steady-state beliefs react to changes in the stage game's parameters,
$q$ and $\kappa$? In general, when learners infer from exogenous
data, their decision problem does not influence learning outcomes.
This observation holds independently of whether learners are misspecified.
On the other hand, correctly specified learners in my setting always
end up with correct beliefs in the long run, so the game parameters
are again irrelevant. With misspecified learners in an endogenous-data
setting, however, changes in the stage game carry long-run consequences
on society's beliefs about the fundamentals.
\begin{prop}
\label{prop:comparative_statics_example} Suppose $r=0$ and $\gamma_{n}>0.$
Let $((\mu_{1}^{(q,\kappa)},\mu_{2}^{(q,\kappa)},\gamma_{n}),c^{(q,\kappa)})$
denote the unique steady-state beliefs and cutoff under parameters
$q\in[0,1),\kappa\in\mathbb{R}$. The steady-state belief $\mu_{2}^{(q,\kappa)}$
is strictly increasing in $q$ and strictly decreasing in $\kappa$,
but always satisfies $\mu_{2}^{(q,\kappa)}<\mu_{2}^{\bullet}$. The
steady-state cutoff threshold $c^{(q,\kappa)}$ is strictly increasing
in $q$ and strictly decreasing in $\kappa.$
\end{prop}
Proposition \ref{prop:comparative_statics_example} provides novel
predictions about how the economic environment affects biased inference
under the gambler's fallacy. It says when agents are more patient
(i.e., suffer a lower waiting cost, or receive a higher subsidy for
continuing) or when they have a higher chance of recalling previous
draws, then they will end up with less distorted beliefs about the
pool in the long run. These changes in environmental parameters partially
correct society's long-run beliefs by incentivizing longer search
and mitigating the censoring effect.

\subsubsection{Fictitious Variation and Censoring}

So far, I have assumed agents hold dogmatic and correct beliefs about
the variance of $X_{1}$ and the conditional variance of $X_{2}\mid(X_{1}=x_{1}).$
Now consider agents who are uncertain about these variances and jointly
estimate them together with the means of the pools. I show that agents
end up exaggerating the variances, in a way that depends on the severity
of data censoring.

For $\mu_{1},\mu_{2}\in\mathbb{R},$ $\sigma_{1}^{2},\sigma_{2}^{2}\ge0,$
and $\gamma\in\mathbb{R},$ let $\Psi(\mu_{1},\mu_{2},\sigma_{1}^{2},\sigma_{2}^{2};\gamma)$
refer to the joint distribution $X_{1}=\mu_{1}+\epsilon_{1},$ $X_{2}=\mu_{2}+\epsilon_{2}$
with $\epsilon_{1}\sim\mathcal{N}(0,\sigma_{1}^{2}),$ $(\epsilon_{2}\mid\epsilon_{1})\sim\mathcal{N}(-\gamma\epsilon_{1},\sigma_{2}^{2})$.
In this section, ``fundamentals'' refer to the four parameters $\mu_{1},\mu_{2},\sigma_{1}^{2},\sigma_{2}^{2}$,
and I assume for simplicity $\gamma_{l}=\gamma_{h}=\gamma>0.$ Objectively,
$X_{1},X_{2}$ are independent Gaussian random variables each with
a variance of $(\sigma^{\bullet})^{2}>0$, so the true joint distribution
of $(X_{1},X_{2})$ is $\Psi^{\bullet}:=\Psi(\mu_{1}^{\bullet},\mu_{2}^{\bullet},(\sigma^{\bullet})^{2},(\sigma^{\bullet})^{2};0)$.

Following Equation (\ref{eq:KL}), write $D_{KL}(\mathcal{H}^{\bullet}(c)\parallel\mathcal{H}(\Psi(\mu_{1},\mu_{2},\sigma_{1}^{2},\sigma_{2}^{2};\gamma);c))\ )$
to denote the KL divergence between the true distribution of histories
with $X_{2}$ censored whenever $X_{1}>c$ and the implied history
distribution under the fundamentals $\mu_{1},\mu_{2},\sigma_{1}^{2},\sigma_{2}^{2}$.
This divergence is given by {\footnotesize{}
\begin{align}
 & \int_{c}^{\infty}\phi(x_{1}\mid\mu_{1}^{\bullet},(\sigma^{\bullet})^{2})\cdot\ln\left(\frac{\phi(x_{1}\mid\mu_{1}^{\bullet},(\sigma^{\bullet})^{2})}{\phi(x_{1}\mid\mu_{1},\sigma_{1}^{2})}\right)dx_{1}\label{eq:KL_mean_var}\\
 & +\int_{-\infty}^{c}\left\{ \int_{-\infty}^{\infty}\phi(x_{1}\mid\mu_{1}^{\bullet},(\sigma^{\bullet})^{2})\cdot\phi(x_{2}\mid\mu_{2}^{\bullet},(\sigma^{\bullet})^{2})\cdot\ln\left[\frac{\phi(x_{1}\mid\mu_{1}^{\bullet},(\sigma^{\bullet})^{2})\cdot\phi(x_{2}\mid\mu_{2}^{\bullet},(\sigma^{\bullet})^{2})}{\phi(x_{1}\mid\mu_{1},\sigma_{2}^{2})\cdot\phi(x_{2}\mid\mu_{2}-\gamma(x_{1}-\mu_{1}),\sigma_{2}^{2})}\right]dx_{2}\right\} dx_{1},\nonumber 
\end{align}
}where $\phi(x\mid\mu,\sigma^{2})$ is the Gaussian density with mean
$\mu$ and variance $\sigma^{2},$ evaluated at $x.$

The next proposition gives closed-form expressions for the pseudo-true
fundamentals $\mu_{1}^{*},\mu_{2}^{*},(\sigma_{1}^{*})^{2},(\sigma_{2}^{*})^{2}$
that minimize Equation (\ref{eq:KL_mean_var}).
\selectlanguage{american}%
\begin{prop}
\label{prop:pseudo_true_mean_var} Suppose $r=0.$ The solutions of
\[
\min_{\mu_{1,}\mu_{2}\in\mathbb{R},\sigma_{1}^{2},\sigma_{2}^{2}\ge0}D_{KL}(\mathcal{H}^{\bullet}(c)\parallel\mathcal{H}(\Psi(\mu_{1},\mu_{2},\sigma_{1}^{2},\sigma_{2}^{2};\gamma);c))\ )
\]
 are $\mu_{1}^{*}(c)=\mu_{1}^{\bullet},$ $\mu_{2}^{*}(c)=\mu_{2}^{\bullet}-\gamma\left(\mu_{1}^{\bullet}-\mathbb{E}\left[X_{1}\mid X_{1}\le c\right]\right),$
$(\sigma_{1}^{*})^{2}(c)=(\sigma^{\bullet})^{2},$ and $(\sigma_{2}^{*})^{2}(c)=(\sigma^{\bullet})^{2}+\gamma^{2}\textnormal{Var}[X_{1}\mid X_{1}\le c].$
So, $(\sigma_{2}^{*})^{2}(c)$ strictly increases in $c.$
\end{prop}
Comparing Proposition \ref{prop:pseudo_true_mean_var} and with the
expressions for $\mu_{1}^{*}(c),\mu_{2}^{*}(c)$ in Proposition \ref{prop:pseudo_true_normal}
(for the special case of $r=0$, $\gamma_{n}=\gamma>0$, and a strategy
that stops when  $X_{1}\ge c$) shows that \foreignlanguage{english}{agents}
misinfer the means in the same way regardless of whether they know
the variances. Biased \foreignlanguage{english}{agents} correctly
estimate the first-period variance, $(\sigma_{1}^{*})^{2}=(\sigma^{\bullet})^{2},$
but over-estimate second-period variance. They exaggerate the variation
in quality among the late-phase draws. This phenomenon relates to
findings in \citet{rabin2002inference} and \citet{rabin2010gambler},
who refer to exaggeration of variance under the gambler's fallacy
as \emph{fictitious variation}. The key innovation of Proposition
\foreignlanguage{english}{\ref{prop:pseudo_true_mean_var} is to show,
in an endogenous-data setting, how the degree of fictitious variation
depends on the severity of censoring.}

The magnitude of this distortion increases in the severity of the
gambler's fallacy but decreases with the severity of the censoring,
as $\text{Var}[X_{1}\mid X_{1}\le c]$ increases in $c$ for $X_{1}$
Gaussian.\foreignlanguage{english}{ Here is the intuition. Whereas
the objective conditional distribution of $X_{2}\mid(X_{1}=x_{1})$
is independent of $x_{1},$ the biased agents entertain different
beliefs about this distribution for different $x_{1}$'s. The agents'
best-fitting inference about $\mu_{2}$ ensures their belief about
$X_{2}\mid(X_{1}=x_{1})$ fits the data well following ``typical''
realizations of $x_{1}$ in the continuation region $(-\infty,c]$.
But they are still surprised when they experience a streak of bad
draws in their own stage game. Agents who observe such surprising
streaks attribute the unexpectedly low realizations of $X_{2}$ to
``noise,'' and thus pass down beliefs that estimate a higher conditional
variance of $X_{2}\mid(X_{1}=x_{1})$. A larger fraction of the agents
attribute their data to ``noise'' when $\text{Var}[X_{1}\mid X_{1}\le c]$
is larger, for the frequency of the surprising streaks depends on
how much $X_{1}$ tends to deviate from its typical value of $\mathbb{E}[X_{1}\mid X_{1}\le c]$
conditional on the event $\{X_{1}\le c\}$.}

\selectlanguage{english}%
The next result demonstrates the interplay between fictitious variation
and endogenous censoring in the steady state. Consider two societies
of agents, who have the same bias, play the same stage game, and face
the same true fundamentals. Agents in Society A know the true variances
and only infer about $(\mu_{1},\mu_{2}),$ while those in Society
B do not know the variances and infer about $(\mu_{1},\mu_{2},\sigma_{1}^{2},\sigma_{2}^{2})$.
\begin{prop}
\label{prop:AB_societies}Suppose $r=0,$ $\gamma_{l}=\gamma_{h}=\gamma$,
and the probability of recall is interior, $0<q<1.$ Let $(\mu_{1}^{A},\mu_{2}^{A},c^{A})$
and $(\mu_{1}^{B},\mu_{2}^{B},(\sigma_{1}^{B})^{2},(\sigma_{2}^{B})^{2},c^{B})$
be the steady-state beliefs about the fundamentals and the steady-state
cutoffs in the two societies. Then $\mu_{2}^{B}>\mu_{2}^{A}$ and
$c^{B}>c^{A}.$ Also, $\sigma_{2}^{B}>\sigma_{2}^{*}(c^{A})$.
\end{prop}
The endogenous-data setting leads to two novel implications of fictitious
variation relative to \foreignlanguage{american}{\citet{rabin2010gambler}'s
exogenous-data world.} First, even though \foreignlanguage{american}{Proposition
\ref{prop:pseudo_true_mean_var} implies that the two societies would
make the same inferences about the pool means if they were given the
same data, in steady state Society B holds more optimistic (i.e.,
more correct) beliefs about $\mu_{2}$ and uses a higher cutoff than
Society A. Allowing uncertainty on one dimension (variance) ends up
affecting society's long-run inference in another dimension (mean),
because a belief in fictitious variation increases the }agents'\foreignlanguage{american}{
perceived option value of continuing and thus changes their behavior
and the kind of data they observe in the steady state. Second, fictitious
variation has a ``multiplier effect,'' as formalized by the final
statement of Proposition \ref{prop:AB_societies}. Society B's steady-state
belief about $\sigma_{2}$ is higher than what it would have been
had they simply inferred using data generated from Society A's steady-state
cutoff $c^{A}.$ Allowing for uncertainty about the pool variances
leads to fictitious variation that increases Society B's cutoff above
$c^{A}$. This is because when the }agent\foreignlanguage{american}{
can recall the first draw with an interior probability, the option
value of waiting for the second draw is larger when the second labor
pool has a larger variance in quality. This higher cutoff further
heightens Society B's belief in fictitious variation, since Proposition
\ref{prop:pseudo_true_mean_var} implies $\sigma_{2}^{*}(c)$ is strictly
increasing, and so forth.}

\section{\label{sec:Convergence}Convergence to the Steady State}

This section shows the steady state defined and studied earlier corresponds
to the long-run learning outcome for a society of biased agents acting
one by one.

Time is discrete and partitioned into rounds $t=1,2,3,...$ One short-lived
agent arrives per round. For simplicity, in analyzing convergence
I focus on learning about the fundamentals $\mu_{1},\mu_{2}$ and
suppose agents have a degenerate belief about the reversal parameter,
$\gamma_{l}=\gamma_{h}>-1.$ Agent 1 starts with a prior belief $M_{0}$
given by a continuously differentiable prior density $m_{0}:[\underline{\mu}_{1},\bar{\mu}_{1}]\times[\underline{\mu}_{2},\bar{\mu}_{2}]\to\mathbb{R}_{>0}$,
while each agent $t\ge2$ adopts the final belief $\tilde{M}_{t-1}$
of agent $t-1$ as her prior belief. Since all agents commit the same
statistical bias, each agent's inherited belief aggregates all the
information in all predecessors' histories. The same learning dynamics
obtain in an environment where every agent starts with the common
prior belief $M_{0}$ and observes the stage-game histories of all
predecessors.

In each round $t$, agent $t$ chooses a cutoff threshold $\tilde{C}_{t}$
to maximize her expected payoff based on her prior belief.\footnote{I focus on learning across different iterations of the stage game
and assume agents do not update beliefs within the stage game.} She observes the outcome of her game and updates her belief from
$\tilde{M}_{t-1}$ to $\tilde{M}_{t}$ by applying Bayes' rule to
her stage-game history, $\tilde{H}_{t}\in\mathbb{H}$. She then passes
down $\tilde{M}_{t}$ as the prior belief of agent $t+1.$

By Proposition \ref{prop:existence_and_uniqueness},    there exists
a unique steady state $((\mu_{1}^{\infty},\mu_{2}^{\infty},\gamma_{n}),c^{\infty})$
when $|\frac{r-\gamma_{n}}{1+\gamma_{n}}|<1$. Proposition \ref{prop:one-by-one}
shows that almost surely behavior and belief converge to this steady
state for any prior density $m_{0}$, provided the support $[\underline{\mu}_{1},\bar{\mu}_{1}]\times[\underline{\mu}_{2},\bar{\mu}_{2}]$
includes the steady-state beliefs $(\mu_{1}^{\infty},\mu_{2}^{\infty})$.
To state this convergence result formally, I need to develop the probability
space underlying the learning system.

The sequences $(\tilde{M}_{t}),(\tilde{C}_{t}),(\tilde{H}_{t})$ are
stochastic processes whose randomness stem from randomness of the
stage-game draws in different rounds. The convergence result is about
the almost sure convergence of the processes $(\tilde{M}_{t})$ and
$(\tilde{C}_{t}).$ Consider the $\mathbb{R}^{2}$-valued stochastic
process $(X_{t})_{t\ge1}=(X_{1,t},X_{2,t})_{t\ge1}$, where $X_{t}$
and $X_{t^{'}}$ are independent for $t\ne t^{'}$. Within each $t,$
$X_{1,t}\sim\phi(\cdot\mid\mu_{1}^{\bullet})$ and $X_{2,t}\mid(X_{1,t}=x_{1,t})\sim\phi(\cdot\mid\mu_{2}^{\bullet}-r(x_{1,t}-\mu_{1}^{\bullet}))$
are jointly Gaussian. Interpret $X_{t}$ as the pair of potential
draws in the $t$-th round of the stage game. Clearly, there exists
a probability space $(\Omega,\mathcal{A},\mathbb{P})$, with sample
space $\Omega=(\mathbb{R}^{2})^{\infty}$ interpreted as paths of
the process just described, $\mathcal{A}$ the Borel $\sigma$-algebra
on $\Omega,$ and $\mathbb{P}$ the measure on sample paths so that
the process $X_{t}(\omega)=\omega_{t}$ has the desired distribution.
The term ``almost surely'' means ``with probability 1 with respect
to the realization of the infinite sequence of all (potential) draws'',
i.e., $\mathbb{P}$-almost surely. The processes $(\tilde{M}_{t}),(\tilde{C}_{t}),(\tilde{H}_{t})$
are defined on this probability space and adapted to the filtration
$(\mathcal{F}_{t})_{t\ge1}$, where $\mathcal{F}_{t}$ is the sub-$\sigma$-algebra
generated by draws up to round $t$, $\mathcal{F}_{t}=\sigma((X_{s})_{s=1}^{t})$.
Write $(\tilde{\mu}_{1,t},\tilde{\mu}_{2,t})$ for the random element
in $[\underline{\mu}_{1},\bar{\mu}_{1}]\times[\underline{\mu}_{2},\bar{\mu}_{2}]$
given by the belief $\tilde{M}_{t}$.
\begin{prop}
\label{prop:one-by-one} Suppose $|\frac{r-\gamma_{n}}{1+\gamma_{n}}|<1$,
$r\ne\gamma_{n},$ and $\gamma_{n}>-1$. Provided $\underline{\mu}_{1}<\mu_{1}^{\bullet}<\bar{\mu}_{1}$
and $\underline{\mu}_{2}<\mu_{2}^{\infty}<\bar{\mu}_{2}$, almost
surely $\lim_{t\to\infty}\tilde{C}_{t}=c^{\infty}$ and $(\tilde{\mu}_{1,t},\tilde{\mu}_{2,t})_{t\ge1}$
converges in $L^{1}$ to $(\mu_{1}^{\bullet},\mu_{2}^{\infty})$,
where $((\mu_{1}^{\bullet},\mu_{2}^{\infty},\gamma_{n}),c^{\infty})$
is the unique steady state.
\end{prop}

\subsection{Proof Outline for Proposition \ref{prop:one-by-one}}

The argument for Proposition \ref{prop:one-by-one} adapts techniques
from \citet*{heidhues2018unrealistic}, in particular a law of large
numbers for martingale increments. I discuss the novelties specific
to my environment below.

\subsubsection{When $\mu_{1}$ Is Known}

First consider a simpler situation where agents dogmatically know
that $\mu_{1}=\mu_{1}^{\bullet}$ and only entertain uncertainty about
$\mu_{2}$ in some bounded interval $[\underline{\mu}_{2},\bar{\mu}_{2}]$
that includes $\mu_{2}^{\infty}$. I use a statistical tool from \citet*{heidhues2018unrealistic},
a version of the law of large numbers for martingales whose quadratic
variation grows linearly.

\textbf{Proposition 10 }from \citet*{heidhues2018unrealistic}: \emph{Let
$(y_{t})_{t}$ be a martingale that satisfies a.s. $[y_{t}]\le vt$
for some constant $v\ge0.$ We have that a.s. $\lim_{t\to\infty}\frac{y_{t}}{t}=0$.}

After simplifying the problem with this result, I establish a pair
of mutual bounds on asymptotic behavior and asymptotic beliefs. If
cutoff thresholds are asymptotically bounded between $c^{l}$ and
$c^{h},$ $c^{l}<c^{h},$ then beliefs about $\mu_{2}$ must be asymptotically
supported on the interval $[\mu_{2}^{*}(c^{l}),\mu_{2}^{*}(c^{h})]$
when $r-\gamma_{n}<0,$and asymptotically supported on the interval
$[\mu_{2}^{*}(c^{h}),\mu_{2}^{*}(c^{l})]$ when $r-\gamma_{n}>0$.
Conversely, if belief is asymptotically supported on the subinterval
$[\mu_{2}^{l},\mu_{2}^{h}]\subseteq[\underline{\mu}_{2},\bar{\mu}_{2}]$,
then cutoff thresholds must be asymptotically bounded between $C(\mu_{1}^{\bullet},\mu_{2}^{l};\gamma_{n})$
and $C(\mu_{1}^{\bullet},\mu_{2}^{h};\gamma_{n})$.

Applying this pair of lemmas to $[\underline{\mu}_{2},\bar{\mu}_{2}]$,
I conclude that asymptotically $\tilde{M}_{t}$ must be supported
on the subinterval with the end points $\mathcal{I}(\underline{\mu}_{2})$
and $\mathcal{I}(\bar{\mu}_{2}),$ where $\mathcal{I}$ is the composition
$\mathcal{I}(\mu_{2}):=\mu_{2}^{*}(C(\mu_{1}^{\bullet},\mu_{2};\gamma)).$
The proof of Proposition \ref{prop:existence_and_uniqueness} implies
that $\mathcal{I}$ is a contraction map whose iterates converge to
$\mu_{2}^{\infty}.$ Therefore by repeatedly applying the pair of
lemmas, the bound on asymptotic beliefs gets refined down to the singleton
$\{\mu_{2}^{\infty}\}$, showing the almost-sure convergence of beliefs
and behavior.

\subsubsection{Uncertainty About $\mu_{1}$}

In the hypothesis of Proposition \ref{prop:one-by-one}, both $\mu_{1}$
and $\mu_{2}$ are unknown, so there is two-dimensional uncertainty
about the fundamentals. This complication prevents a direct application
of \citet*{heidhues2018unrealistic}'s statistical tools, as their
tools are only designed to work with a one-dimensional fundamental.
But the structure of the inference problem is such that I can separately
bound the agents' asymptotic beliefs in two ``directions,'' thus
reducing the task of proving a two-dimensional belief bound into a
pair of tasks involving one-dimensional belief bounds.

Consider a pair of fundamentals, $(\mu_{1},\mu_{2})$ and $(\mu_{1}^{'},\mu_{2}^{'})=(\mu_{1}+d,\mu_{2}-\gamma d)$
for some $d>0$, satisfying $\mu_{1},\mu_{1}^{'}\le\mu_{1}^{\bullet}$.
That is, $(\mu_{1},\mu_{2})$ and $(\mu_{1}^{'},\mu_{2}^{'})$ lie
on the same line with slope $-\gamma$. For any uncensored history
$(x_{1},x_{2})\in\mathbb{R}^{2}$, the likelihood of second-period
draw $x_{2}$ is the same under both pairs of fundamentals, $\phi(x_{2}\mid\mu_{2}-\gamma(x_{1}-\mu_{1}))=\phi(x_{2}\mid\mu_{2}^{'}-\gamma(x_{1}-\mu_{1}^{'})).$
So both pairs of fundamentals $(\mu_{1},\mu_{2})$ and $(\mu_{1}^{'},\mu_{2}^{'})$
explain $X_{2}$ data equally well in \emph{all} uncensored histories.
At the same time, $(\mu_{1}^{'},\mu_{2}^{'})$ provides a strictly
better fit for $X_{1}$ data on average than $(\mu_{1},\mu_{2}),$
since $\mu_{1}<\mu_{1}^{'}\le\mu_{1}^{\bullet}.$ This means in the
long run, fundamentals $(\mu_{1},\mu_{2})$ should receive much less
posterior probability than $(\mu_{1}^{'},\mu_{2}^{'})$, as the latter
better rationalize the data overall.

To formalize this, I compute the directional derivative for data log-likelihood
along the vector $(1,-\gamma)$ in the space of fundamentals. I establish
an (almost-sure) positive lowerbound on this directional derivative
at all points at least $2\epsilon^{'}$ to the left of $\mu_{1}^{\bullet},$
and an analogous negative upperbound to the right of $\mu_{1}^{\bullet}.$
(The picture below is an illustration for the case of $\gamma>0$.)
This allows me to show the region colored in red receives 0 posterior
probability asymptotically, by comparing each point in red with a
corresponding point in blue along a line of slope $-\gamma$.
\begin{center}
\includegraphics[scale=0.55]{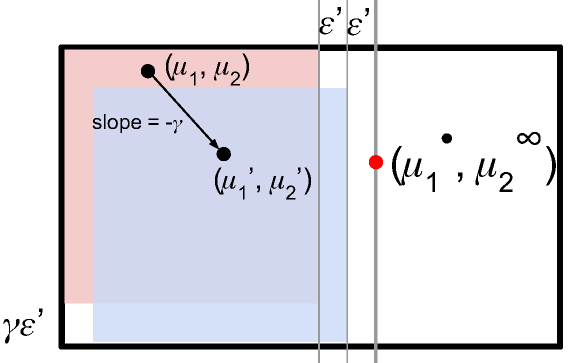}
\par\end{center}

By repeating this argument for small values of $\epsilon^{'}$ (and
applying the symmetric bound to the right of $\mu_{1}^{\bullet})$,
I show that belief is asymptotically concentrated either along a small
vertical strip containing the steady state beliefs, $(\mu_{1}^{\bullet},\mu_{2}^{\infty})$,
or along an edge of belief's support, colored in green. The latter
possibility requires belief in an extreme value of $\mu_{2}\in\{\underline{\mu}_{2},\bar{\mu}_{2}\}$
in the support of the prior $m_{0}$ and can be ruled out by showing
that, within these regions, slightly increasing or decreasing belief
in $\mu_{2}$ leads to better fit.
\begin{center}
\includegraphics[scale=0.55]{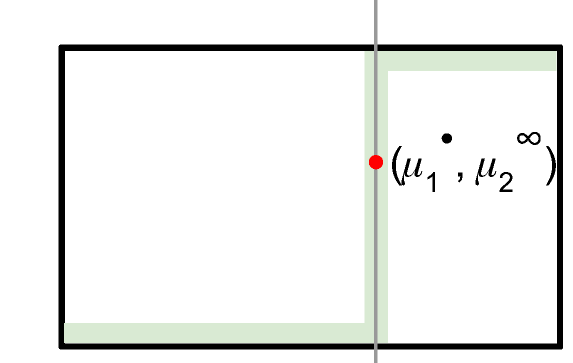}
\par\end{center}

Having restricted the long-run belief to a thin vertical strip, the
first ``direction'' of the belief bounds is complete and the dimensionality
of uncertainty is effectively reduced back to one. The rest of the
argument proceeds similarly to the case where agents know $\mu_{1}^{\bullet}$
discussed above.

\section{\label{sec:Related-Literature}Related Theoretical Literature}

A strand of behavioral economics literature has focused on a different
cognitive error when agents learn from partial data: selection neglect.
Theory papers in this area have studied agents who observe a selective
sample in different settings: good's quality in a bilateral trade
game \citep{esponda2008behavioral}, investment outcomes by past entrepreneurs
\citep{jehiel2018investment}, government policy effectiveness \citep{esponda2017conditional,esponda2019retrospective},
and outcomes of recent experiments \citep{chen_censor}. In all of
these settings, the sample selection depends on some unobserved private
information of other players. Biased agents fail to account for the
informational content of selection,\footnote{Some recent experiments have demonstrated selection neglect in laboratory
subjects: \citet{enke2017you}, \citet*{barron2019everyday}, \citet*{Araujo2018}.} thus make wrong inferences. While I also consider a setting where
agents learn from partial data, I focus on the implications of a different
bias in such environments: the gambler's fallacy. Selection neglect
and the gambler's fallacy can be conceptually unified under the broader
category of correlational mistakes. As \citet{spiegler2016bayesian}
and \citet{spiegler2017data} point out, many examples of selection
neglect can be viewed as biases stemming from incorrect conditional-independence
assumptions. I emphasize that the biased agents in my world do not
additionally suffer from selection neglect. Agents derive different
inferences from histories censored at different thresholds purely
as a result of misperceiving the reversal parameter that relates different
draws; this conclusion does not come from the combination of multiple
behavioral biases.

\citet{rabin2002inference} and \citet{rabin2010gambler} are the
first to study the inferential mistakes implied by the gambler's fallacy.
Like these papers, I consider agents who believe in reversals conditional\emph{
}on the underlying fundamentals and mislearn some parameters of the
world as a result. Except for an example in \citet{rabin2002inference},
all such investigations focus on passive inference, whereby learners
observe an exogenous signal process. By contrast, this paper examines
an endogenous learning setting where actions affect observables. Section
7 of \citet{rabin2002inference} discusses an example of endogenous
learning with a finite-urn model of the gambler's fallacy. The nature
of \citet{rabin2002inference}'s endogenous data, however, is unrelated
to the censoring effect central to the current paper.\footnote{In \citet{rabin2002inference}'s example, biased agents (correctly)
believe that the part of the data which is always observable is independent
of the part of the data which is sometimes missing. However, what
I term the ``censoring effect'' is about misinference resulting
from agents wrongly believing in negative correlation between the
early draws that are always observed and the later draws that may
be censored, depending on the realizations of the early draws. I discuss
this further in the Online Appendix of an earlier version of this
paper: \href{https://arxiv.org/pdf/1803.08170v5.pdf}{https://arxiv.org/pdf/1803.08170v5.pdf}}

This work joins a strand of literature on the implications of misspecified
Bayesian learning when the learner's actions affect the data she observes.
The earliest example is \citet{nyarko1991learning}. \citet{esponda2016berk}
propose an equilibrium concept for such settings — the Berk-Nash equilibrium.
Subsequently, a number of papers have studied the properties of Berk-Nash
equilibria in different applied contexts \citep*{fudenberg2017,heidhues2018unrealistic,FII2018}
and the persistence and stability of misspecifications \citep*{FII_welfare_based,FL_mutation,he2020evolutionarily}.
In addition to using this framework to explore the gambler's fallacy,
I also highlight a new source of data endogeneity relative to the
existing papers — the censoring effect in an optimal-stopping problem.
Agents' stopping decisions determine how many signals they observe
about the fundamentals. Other recent papers \citep*{esponda2021asymptotic,fudenberg2021limit,frick2021belief,heidhues2018convergence}
prove general theorems about the convergence of misspecified learning
in different settings. Though not the primary contribution of this
work, the convergence result in Proposition \ref{prop:one-by-one}
deals with a setting that is not covered by these papers: a multi-dimensional
inference problem with a continuum of states, signals, and actions.

Although Section \ref{sec:Convergence} considers a learning system
with a sequence of short-lived agents, the ``social learning'' aspect
of the framework is not central to the results. In fact, the environment
where a sequence of short-lived agents act one at a time is equivalent
to an environment where a single long-lived agent plays the stage
game repeatedly, myopically maximizing her expected payoff in each
iteration of the stage game. In the growing literature on social learning
with misspecified Bayesians (e.g., \citet*{eyster2010naive,guarino13,bohren2016informational,bohren2017bounded,dasaratha2020network,FII2018,bushong2018misattributeReference}),
agents observe their predecessors' actions but make errors when inverting
these actions to deduce said predecessors' information. This kind
of action inversion does not take place here: later agents inherit
all the information that their predecessors have seen by adopting
their beliefs, so predecessors' actions are uninformative.

The econometrics literature has also studied data-generating processes
with censoring — for example, the Tobit model and models of competing
risks.\footnote{References can be found in \citet{amemiya1985advanced} and \citet{crowder2001classical}.}
This literature has primarily focused on the issue of model identification
from censored data \citep*{cox1962,tsiatis1975nonidentifiability,heckman1989identifiability}.
In my setting, there is no identification problem for correctly specified
agents. Instead, I study how agents make wrong parameter estimates
from censored data when they infer using a family of misspecified
models. Another contrast is that the econometrics literature has focused
on exogenous data-censoring mechanisms, but censoring is endogenous
in this paper and depends on the beliefs of previous agents.

\section{\label{sec:Concluding-Remarks}Concluding Remarks}

This paper studies endogenous learning dynamics of misspecified agents.
The general framework allows different correlational mistakes, and
shows the interaction between the statistical bias and data censoring
in optimal-stopping problems distorts beliefs and behavior. When agents
suffer from the gambler's fallacy, they hold overly pessimistic beliefs
about the fundamentals and stop too frequently in the steady state.
Lower continuation costs, as well as initial uncertainty about the
distribution's variance, partially correct asymptotic beliefs about
the distribution's mean.

An earlier version of this paper\footnote{Available at \href{https://arxiv.org/pdf/1803.08170v5.pdf}{https://arxiv.org/pdf/1803.08170v5.pdf}}
shows that the steady-state results (about over-pessimistic inference
and early stopping) and the convergence result continue to hold for
a larger class of stage games and any symmetric, log-concave distributions.
That earlier version also contains an extension with any finite number
$L\ge2$ of periods instead of two periods.

In line with previous work on the gambler's fallacy, I take the behavioral
error as given and do not try to explain the origin of the bias. Endogenizing
the gambler's fallacy and other common statistical errors is an interesting
open question.

I have studied a particular environment where censoring happens (histories
in optimal-stopping problems). The key mechanism I highlight, the
interaction between data censoring and bias, applies more broadly
and delivers different predictions in different contexts. Environments
that feature different censoring patterns would produce different
predictions, but again through the same basic mechanism— interaction
between censoring and bias. More broadly, other kinds of ``symmetric''
behavioral biases may lead to ``asymmetric'' predictions in environments
that feature directional data censoring. I am leaving open the interaction
of other kinds of behavioral learning with other censoring mechanisms
to future work. 

\begin{spacing}{1.1}
\bibliographystyle{ecta}
\bibliography{gambler}

\end{spacing}
\begin{center}
\textbf{\Large{}Appendix}{\Large\par}
\par\end{center}

\renewcommand{\thesection}{A\arabic{section}} 
\renewcommand{\thedefn}{A.\arabic{defn}} 
\renewcommand{\thelem}{A.\arabic{lem}} 
\renewcommand{\theprop}{A.\arabic{prop}} 
\renewcommand{\thetable}{A.\arabic{table}}
\setcounter{section}{0}
\setcounter{prop}{0}
\setcounter{lem}{0}
\setcounter{defn}{0}
\setcounter{table}{0}

\section{\label{sec:Omitted-Proofs-from-Main}Proofs}

\subsection{Proof of Proposition \ref{prop:pseudo_true_normal}}
\begin{proof}
In the true model, $X_{2}|(X_{1}=x_{1})\sim\mathcal{N}(\mu_{2}^{\bullet}-r(x_{1}-\mu_{1}^{\bullet}),\sigma^{2})$,
while the agents' feasible model $\Psi(\mu_{1},\mu_{2};\gamma)$ has
$X_{2}|(X_{1}=x_{1})\sim\mathcal{N}(\mu_{2}-\gamma(x_{1}-\mu_{1}),\sigma^{2})$.
Suppose histories are generated with a stopping rule that continues
in the positive Lebesgue measure set $K\subseteq\mathbb{R}.$ The
objective in Equation (\ref{eq:KL}) is:{\small{}
\begin{align*}
 & \int_{x_{1}\notin K}\phi(x_{1}\mid\mu_{1}^{\bullet})\cdot\ln\left(\frac{\phi(x_{1}\mid\mu_{1}^{\bullet})}{\phi(x_{1}\mid\mu_{1})}\right)dx_{1}\\
 & +\int_{x_{1}\in K}\phi(x_{1}\mid\mu_{1}^{\bullet})\cdot\left\{ \int_{-\infty}^{\infty}\phi(x_{2}\mid\mu_{2}^{\bullet}-r(x_{1}-\mu_{1}^{\bullet}))\cdot\ln\left[\frac{\phi(x_{1}\mid\mu_{1}^{\bullet})\cdot\phi(x_{2}\mid\mu_{2}^{\bullet}-r(x_{1}-\mu_{1}^{\bullet}))}{\phi(x_{1}\mid\mu_{1})\cdot\phi(x_{2}\mid\mu_{2}-\gamma(x_{1}-\mu_{1}))}\right]dx_{2}\right\} dx_{1}.
\end{align*}
}This can be rewritten as {\footnotesize{}
\begin{align*}
 & \int_{x_{1}\notin K}\phi(x_{1}\mid\mu_{1}^{\bullet})\ln\left(\frac{\phi(x_{1}\mid\mu_{1}^{\bullet})}{\phi(x_{1}\mid\mu_{1})}\right)dx_{1}+\int_{x_{1}\in K}\phi(x_{1}\mid\mu_{1}^{\bullet})\left\{ \int_{-\infty}^{\infty}\phi(x_{2}\mid\mu_{2}^{\bullet}-r(x_{1}-\mu_{1}^{\bullet}))\ln\left[\frac{\phi(x_{1}\mid\mu_{1}^{\bullet})\cdot}{\phi(x_{1}\mid\mu_{1})}\right]dx_{2}\right\} dx_{1}\\
 & +\int_{x_{1}\in K}\phi(x_{1}\mid\mu_{1}^{\bullet})\cdot\left\{ \int_{-\infty}^{\infty}\phi(x_{2}\mid\mu_{2}^{\bullet}-r(x_{1}-\mu_{1}^{\bullet}))\cdot\ln\left[\frac{\phi(x_{2}\mid\mu_{2}^{\bullet}-r(x_{1}-\mu_{1}^{\bullet}))}{\phi(x_{2}\mid\mu_{2}-\gamma(x_{1}-\mu_{1}))}\right]dx_{2}\right\} dx_{1},
\end{align*}
}which is: 
\begin{align*}
 & \int_{-\infty}^{\infty}\phi(x_{1}\mid\mu_{1}^{\bullet})\cdot\ln\left(\frac{\phi(x_{1}\mid\mu_{1}^{\bullet})}{\phi(x_{1}\mid\mu_{1})}\right)dx_{1}+\\
 & +\int_{x_{1}\in K}\phi(x_{1}\mid\mu_{1}^{\bullet})\cdot\int_{-\infty}^{\infty}\phi(x_{2}\mid\mu_{2}^{\bullet}-r(x_{1}-\mu_{1}^{\bullet}))\cdot\ln\left[\frac{\phi(x_{2}\mid\mu_{2}^{\bullet}-r(x_{1}-\mu_{1}^{\bullet}))}{\phi(x_{2}\mid\mu_{2}-\gamma(x_{1}-\mu_{1}))}\right]dx_{2}dx_{1}.
\end{align*}
The KL divergence between $\mathcal{N}(\mu_{\text{true}},\sigma_{\text{true}}^{2})$
and $\mathcal{N}(\mu_{\text{model}},\sigma_{\text{model}}^{2})$ is
{\small{}$\ln\frac{\sigma_{\text{model}}}{\sigma_{\text{true}}}+\frac{\sigma_{\text{true}}^{2}+(\mu_{\text{true}}-\mu_{\text{model}})^{2}}{2\sigma_{\text{model}}^{2}}-\frac{1}{2},$}
so we may simplify the first term and the inner integral of the second
term:
\[
\frac{(\mu_{1}-\mu_{1}^{\bullet})^{2}}{2\sigma^{2}}+\int_{x_{1}\in K}\phi(x_{1}\mid\mu_{1}^{\bullet})\cdot\frac{(\mu_{2}-\gamma(x_{1}-\mu_{1})-\mu_{2}^{\bullet}+r(x_{1}-\mu_{1}^{\bullet}))^{2}}{2\sigma^{2}}dx_{1}.
\]
Multiplying through by $\sigma^{2}$, we get a simplified objective
with the same minimizers: 
\[
\xi(\mu_{1},\mu_{2},\gamma)=\frac{(\mu_{1}-\mu_{1}^{\bullet})^{2}}{2}+\int_{x_{1}\in K}\phi(x_{1}\mid\mu_{1}^{\bullet})\cdot\frac{1}{2}\cdot[\mu_{2}-\gamma(x_{1}-\mu_{1})-\mu_{2}^{\bullet}+r(x_{1}-\mu_{1}^{\bullet})]^{2}dx_{1}.
\]
We have the partial derivatives by differentiating under the integral
sign, 
\[
\frac{\partial\xi}{\partial\mu_{2}}=\int_{x_{1}\in K}\phi(x_{1}\mid\mu_{1}^{\bullet})\cdot[\mu_{2}-\gamma(x_{1}-\mu_{1})-\mu_{2}^{\bullet}+r(x_{1}-\mu_{1}^{\bullet})]dx_{1},
\]
\begin{align*}
\frac{\partial\xi}{\partial\mu_{1}} & =(\mu_{1}-\mu_{1}^{\bullet})+\gamma\int_{x_{1}\in K}\phi(x_{1}\mid\mu_{1}^{\bullet})\cdot[\mu_{2}-\gamma(x_{1}-\mu_{1})-\mu_{2}^{\bullet}+r(x_{1}-\mu_{1}^{\bullet})]dx_{1}\\
 & =(\mu_{1}-\mu_{1}^{\bullet})+\gamma\frac{\partial\xi}{\partial\mu_{2}},
\end{align*}
\[
\frac{\partial\xi}{\partial\gamma}=-\int_{x_{1}\in K}\phi(x_{1}\mid\mu_{1}^{\bullet})\cdot[x_{1}-\mu_{1}]\cdot[\mu_{2}-\gamma(x_{1}-\mu_{1})-\mu_{2}^{\bullet}+r(x_{1}-\mu_{1}^{\bullet})]dx_{1}.
\]
Suppose $(\mu_{1}^{*},\mu_{2}^{*},\gamma^{*})$ is the minimum. By
the first-order conditions for $\mu_{1}$ and $\mu_{2}$, we have:
\[
\frac{\partial\xi}{\partial\mu_{1}}(\mu_{1}^{*},\mu_{2}^{*},\gamma^{*})=\frac{\partial\xi}{\partial\mu_{2}}(\mu_{1}^{*},\mu_{2}^{*},\gamma^{*})=0\Rightarrow\mu_{1}^{*}=\mu_{1}^{\bullet}.
\]
Substituting this into the first-order condition for $\mu_{2},$ 
\[
\frac{\partial\xi}{\partial\mu_{2}}(\mu_{1}^{\bullet},\mu_{2}^{*},\gamma^{*})=0\Rightarrow\mu_{2}^{*}=\mu_{2}^{\bullet}+(r-\gamma^{*})\cdot\left(\mu_{1}^{\bullet}-\mathbb{E}[X_{1}|X_{1}\in K]\right).
\]
It remains to find $\gamma^{*}.$ We have 
\[
\frac{\partial\xi}{\partial\gamma}(\mu_{1}^{*},\mu_{2}^{*},\gamma^{*})=-\mathbb{P}[X_{1}\in K]\cdot\mathbb{E}[(X_{1}-\mu_{1}^{*})\cdot(\mu_{2}^{*}-\gamma^{*}(X_{1}-\mu_{1}^{*})-\mu_{2}^{\bullet}+r(X_{1}-\mu_{1}^{\bullet}))|X_{1}\in K].
\]
We rearrange the expectation term as: 
\begin{align*}
 & \mathbb{E}[(X_{1}-\mu_{1}^{*})\cdot(\mu_{2}^{*}-\gamma^{*}(X_{1}-\mu_{1}^{*})-\mu_{2}^{\bullet}+r(X_{1}-\mu_{1}^{\bullet}))|X_{1}\in K]\\
= & \mathbb{E}[(X_{1}-\mu_{1}^{*})|X_{1}\in K]\cdot\mathbb{E}[(\mu_{2}^{*}-\gamma^{*}(X_{1}-\mu_{1}^{*})-\mu_{2}^{\bullet}+r(X_{1}-\mu_{1}^{\bullet}))|X_{1}\in K]\\
 & +\text{Cov}[X_{1}-\mu_{1}^{*},\mu_{2}^{*}-\gamma^{*}(X_{1}-\mu_{1}^{*})-\mu_{2}^{\bullet}+r(X_{1}-\mu_{1}^{\bullet})|X_{1}\in K].
\end{align*}
The first-order condition for $\mu_{2}$ implies $\mathbb{E}[(\mu_{2}^{*}-\gamma^{*}(X_{1}-\mu_{1}^{*})-\mu_{2}^{\bullet}+r(X_{1}-\mu_{1}^{\bullet}))|X_{1}\in K]=0$
at the optimum $(\mu_{1}^{*},\mu_{2}^{*},\gamma^{*})$. Also, we may
drop terms without $X_{1}$ in the conditional covariance operator,
and we get: 
\[
\frac{\partial\xi}{\partial\gamma}(\mu_{1}^{*},\mu_{2}^{*},\gamma^{*})=\mathbb{P}[X_{1}\in K]\cdot(\gamma^{*}-r)\cdot\text{Cov}(X_{1},X_{1}|X_{1}\in K).
\]
We have $\mathbb{P}[X_{1}\in K]>0$ and $\text{Cov}(X_{1},X_{1}|X_{1}\in K)>0,$
hence we conclude 
\[
\frac{\partial\xi}{\partial\gamma}(\mu_{1}^{*},\mu_{2}^{*},\gamma^{*})\begin{cases}
>0 & \text{for }\gamma^{*}>r\\
=0 & \text{for }\gamma^{*}=r\\
<0 & \text{for }\gamma^{*}<r
\end{cases}.
\]
When $r\in[\gamma_{l},\gamma_{h}],$ $(\mu_{1}^{*},\mu_{2}^{*},\gamma^{*})$
cannot minimize $\xi$ if $\gamma^{*}\ne r$: at either end point
where FOC in $\gamma$ does not hold, $\xi$ can be strictly reduced
by changing $\gamma$ slightly. In case that $\gamma_{l}>r,$ at the
optimum we must have $\frac{\partial\xi}{\partial\gamma}(\mu_{1}^{*},\mu_{2}^{*},\gamma^{*})>0$.
By Karush-Kuhn-Tucker condition, this means the minimizer is $\gamma^{*}=\gamma_{l}.$
Conversely, when $\gamma_{h}<r,$ at the optimum we must have $\frac{\partial\xi}{\partial\gamma}(\mu_{1}^{*},\mu_{2}^{*},\gamma^{*})<0$.
In that case, the minimizer is $\gamma^{*}=\gamma_{h}$. So in both
cases, $\gamma^{*}=\gamma_{n}$ as desired.

Finally, by using $\mu_{2}^{*}=\mu_{2}^{\bullet}+(r-\gamma_{n})\cdot\left(\mu_{1}^{\bullet}-\mathbb{E}[X_{1}|X_{1}\in K]\right)$
and specializing to the case where the continuation region $K$ is
either $(-\infty,c]$ or $[c,\infty)$, we get the closed-form expression
of $\mu_{2}^{*}(c).$
\end{proof}

\subsection{Proof of Lemma \ref{lem:behavior_simplified}}

I state and prove a stronger result, which will be used in some of
the later proofs.
\begin{lem}
\label{lem:cutoff_properties} Consider the model $\Psi(\mu_{1},\mu_{2};\gamma)$
for any $\mu_{1},\mu_{2},\gamma\in\mathbb{R}.$ Let $D(x_{1})$ be
the difference between the expected payoff in stopping and continuing
after $X_{1}=x_{1}$ in the model. If $\gamma=-1$, then $D(x_{1})$
is constant in $x_{1}.$ If $\gamma>-1,$ then $D(x_{1})$ is continuous
and strictly increasing in $x_{1}$ with $\lim_{x_{1}\to\pm\infty}D(x_{1})=\pm\infty$.
If $\gamma<-1,$ then $D(x_{1})$ is continuous and strictly decreasing
in $x_{1}$ with $\lim_{x_{1}\to\pm\infty}D(x_{1})=\mp\infty$. When
$\gamma\ne-1,$ there is a unique $C(\mu_{1},\mu_{2};\gamma)$ so
that the agent is indifferent between continuing and stopping after
$x_{1}=C(\mu_{1},\mu_{2};\gamma)$. For fixed $\mu_{1}\in\mathbb{R},$
$\gamma\ne-1,$ the function $\mu_{2}\mapsto C(\mu_{1},\mu_{2};\gamma)$
is linear with a slope of $\frac{1}{\gamma+1}$.
\end{lem}
Using Lemma \ref{lem:cutoff_properties}, agents stop after high values
of $X_{1}$ when $\gamma>-1$ and stop after low values of $X_{1}$
when $\gamma<-1$, because $D$ is strictly increasing when $\gamma>-1$
and strictly decreasing when $\gamma<-1.$ Also, since $\mu_{2}\mapsto C(\mu_{1},\mu_{2};\gamma)$
has a slope of $\frac{1}{\gamma+1}$, it is strictly increasing if
$\gamma>-1$ and strictly decreasing if $\gamma<-1.$ I now prove
Lemma \ref{lem:cutoff_properties}.
\begin{proof}
In the model $\Psi(\mu_{1},\mu_{2};\gamma)$, the expected difference
between stopping and continuing after $X_{1}=x_{1}$ is:
\[
D(x_{1})=x_{1}-q\mathbb{E}[\max(x_{1},[X_{2}\mid x_{1}])]-(1-q)\mathbb{E}[X_{2}\mid x_{1}]+\kappa
\]
where $[X_{2}\mid x_{1}]\sim\mathcal{N}(\mu_{2}-\gamma(x_{1}-\mu_{1}),\sigma^{2})$.
This is clearly continuous in $x_{1}.$ When $\gamma=-1,$ $D$ is
constant because we have for every $a\in\mathbb{R},$ $\delta>0$
\[
D(a+\delta)-D(a)=\delta-q\left\{ \mathbb{E}[\max(a+\delta,[X_{2}\mid a+\delta])]-\mathbb{E}[\max(a,[X_{2}\mid a])]\right\} -\delta(1-q).
\]
In comparing $\max(a+\delta,[X_{2}\mid a+\delta])$ and $\max(a,[X_{2}\mid a])$,
note the distribution $[X_{2}\mid a+\delta]$ is $[X_{2}\mid a]$
shifted to the right by $\delta$, so the distribution $\max(a+\delta,[X_{2}\mid a+\delta])$
is just $\max(a,[X_{2}\mid a])$ shifted to the right by $\delta$.
Thus, $\mathbb{E}[\max(a+\delta,[X_{2}\mid a+\delta])]-\mathbb{E}[\max(a,[X_{2}\mid a])]=\delta.$
So overall, $D(a+\delta)-D(a)=0.$

When $\gamma>-1,$ $[X_{2}\mid a+\delta]$ is strictly stochastically
dominated by $\delta+[X_{2}\mid a]$, therefore $\mathbb{E}[\max(a+\delta,[X_{2}\mid a+\delta])<\delta+\mathbb{E}[\max(a,[X_{2}\mid a])]$.
Also, we have $\mathbb{E}[X_{2}\mid a+\delta]-\mathbb{E}[X_{2}\mid a]=-\delta\gamma.$
So, we get $D(a+\delta)-D(a)>(1-q)(1+\gamma)\delta>0.$ This shows
$D$ is strictly increasing at a rate of at least $(1-q)(1+\gamma)$
at every point in the domain, therefore $\lim_{x_{1}\to\pm\infty}D(x_{1})=\pm\infty$.

When $\gamma<-1,$ $[X_{2}\mid a+\delta]$ strictly stochastically
dominates $\delta+[X_{2}\mid a]$, therefore $\mathbb{E}[\max(a+\delta,[X_{2}\mid a+\delta])>\delta+\mathbb{E}[\max(a,[X_{2}\mid a])]$.
Also, we have $\mathbb{E}[X_{2}\mid a+\delta]-\mathbb{E}[X_{2}\mid a]=-\gamma\delta.$
So, we get $D(a+\delta)-D(a)<(1-q)(1+\gamma)\delta<0.$ This shows
$D$ is strictly decreasing at a rate of at least $(1-q)(1+\gamma)$
at every point in the domain, therefore $\lim_{x_{1}\to\pm\infty}D(x_{1})=\mp\infty$.

When $\gamma\ne-1,$ the existence and uniqueness of $C(\mu_{1},\mu_{2};\gamma)$
come from the fact that $D$ is strictly monotonic and takes on both
positive and negative values, so it must cross 0 at a unique point.

In fact, $C(\mu_{1},\mu_{2};\gamma)$ is linear in $\mu_{2}$ with
a coefficient of $\frac{1}{\gamma+1}$. To see this, fix $\mu_{1}$
and $\gamma$ and consider the difference $x_{1}-q\mathbb{E}[\max(x_{1},[X_{2}\mid x_{1}])]-(1-q)\mathbb{E}[X_{2}\mid x_{1}]+\kappa$
as a function $G(x_{1},\mu_{2})$ of $x_{1}$ and $\mu_{2}.$ For
every $\delta>0,$ we have $G(x_{1}+\frac{\delta}{\gamma+1},\mu_{2}+\delta)=G(x_{1},\mu_{2}).$
This is because 
\begin{align*}
\mathcal{N}((\mu_{2}+\delta)-\gamma((x_{1}+\frac{\delta}{\gamma+1})-\mu_{1}),\sigma^{2}) & =\mathcal{N}(\mu_{2}-\gamma(x_{1}-\mu_{1}),\sigma^{2})+\delta-\delta\frac{\gamma}{\gamma+1}\\
 & =\mathcal{N}(\mu_{2}-\gamma(x_{1}-\mu_{1}),\sigma^{2})+\delta\frac{1}{\gamma+1},
\end{align*}
 therefore $q\mathbb{E}_{\mu_{2}+\delta}[\max(x_{1}+\frac{\delta}{\gamma+1},[X_{2}\mid x_{1}])]=q(\mathbb{E}_{\mu_{2}}[\max(x_{1},[X_{2}\mid x_{1}])]+\delta\frac{1}{\gamma+1})$.
Also, $(1-q)\mathbb{E}_{\mu_{2}+\delta}[X_{2}\mid x_{1}]=(1-q)(\mathbb{E}_{\mu_{2}}[X_{2}\mid x_{1}]+\delta\frac{1}{\gamma+1}).$
Using these two facts, 
\[
G(x_{1}+\frac{\delta}{\gamma+1},\mu_{2}+\delta)-G(x_{1},\mu_{2})=\frac{\delta}{\gamma+1}-q(\delta\frac{1}{\gamma+1})-(1-q)(\delta\frac{1}{\gamma+1})=0.
\]
That is, increasing belief about $\mu_{2}$ by $\delta$ and also
increasing the realization of the early draw by $\delta/(\gamma+1)$
cancel each other out in terms of the difference between the expected
payoffs in stopping and continuing. Therefore, we must have $C(\mu_{1},\mu_{2}+\delta;\gamma)=C(\mu_{1},\mu_{2};\gamma)+\delta\frac{1}{\gamma+1}.$
\end{proof}

\subsection{Proof of Proposition \ref{prop:existence_and_uniqueness}}
\begin{proof}
Consider the map $\mathcal{I}:\mathbb{R}\to\mathbb{R}$ defined by
$\mathcal{I}(\mu_{2}):=\mu_{2}^{*}(C(\mu_{1}^{\bullet},\mu_{2};\gamma_{n}))$,
where we define $\mu_{2}^{*}(c)=\mu_{2}^{\bullet}+(r-\gamma_{n})\cdot(\mu_{1}^{\bullet}-\mathbb{E}[X_{1}\mid X_{1}\le c])$
if $\gamma_{n}>-1$ and $\mu_{2}^{*}(c)=\mu_{2}^{\bullet}+(r-\gamma_{n})\cdot(\mu_{1}^{\bullet}-\mathbb{E}[X_{1}\mid X_{1}\ge c])$
if $\gamma_{n}<-1.$ Lemma \ref{lem:cutoff_properties} shows $\mu_{2}\mapsto C(\mu_{1}^{\bullet},\mu_{2};\gamma_{n})$
is linear with a slope of $\frac{1}{\gamma_{n}+1}$. Also, by property
of the Gaussian distribution, both $c\mapsto\mathbb{E}[X_{1}\mid X_{1}\le c]$
and $c\mapsto\mathbb{E}[X_{1}\mid X_{1}\ge c]$ are Lipschitz continuous
with a Lipschitz constant of 1. Therefore, the composition $\mathcal{I}$
is Lipschitz continuous with a Lipschitz constant of $|\frac{r-\gamma_{n}}{1+\gamma_{n}}|<1$,
hence a contraction map. By property of contraction maps, $\mathcal{I}$
has a unique fixed point, which we denote $\mu_{2}^{\infty}.$ When
$\gamma_{n}>-1,$ the beliefs $(\mu_{1}^{\infty},\mu_{2}^{\infty},\gamma_{n})$
together with the cutoff strategy that stops when $c\ge C(\mu_{1}^{\infty},\mu_{2}^{\infty};\gamma_{n})$
make up a steady state by Proposition \ref{prop:pseudo_true_normal}
and Lemma \ref{lem:behavior_simplified}. When $\gamma_{n}<-1,$ the
beliefs $(\mu_{1}^{\infty},\mu_{2}^{\infty},\gamma_{n})$ together
with the cutoff strategy that stops when $c\le C(\mu_{1}^{\infty},\mu_{2}^{\infty};\gamma_{n})$
make up a steady state for the same reason. Also, this steady state
is unique. By Proposition \ref{prop:pseudo_true_normal}, in any steady-state
beliefs $(\mu_{1}',\mu_{2}',\gamma')$ we must have $\mu_{1}'=\mu_{1}^{\bullet}$,
$\gamma'=\gamma_{n}$. This implies $\mu_{2}'$ must be a fixed point
of $\mathcal{I}$ by the optimality of behavior and the KL-divergence
minimization of beliefs, yet $\mu_{2}^{\infty}$ is the unique fixed
point of $\mathcal{I}$.
\end{proof}

\subsection{Proof of Proposition \ref{prop:steady_state_example}}
\begin{proof}
Under the condition $|\frac{r-\gamma_{n}}{1+\gamma_{n}}|<1$, by Proposition
\ref{prop:existence_and_uniqueness} there exists a unique steady
state where $\gamma^{\infty}=\gamma_{n}$, and the agent uses a cutoff
strategy with some threshold $c^{\infty}.$ The agent stops when $X_{1}\ge c^{\infty}$
if $\gamma_{n}>-1,$ and stops when $X_{1}\le c^{\infty}$ if $\gamma_{n}<-1$.

Suppose $r,\gamma_{n}>-1$. Then by Proposition \ref{prop:pseudo_true_normal},
$\mu_{2}^{\infty}=\mu_{2}^{\bullet}+(r-\gamma_{n})\cdot(\mu_{1}^{\bullet}-\mathbb{E}[X_{1}\mid X_{1}\le c^{\infty}])$.
Since $\mathbb{E}[X_{1}\mid X_{1}\le c^{\infty}]<c^{\infty},$ we
get $\mu_{2}^{\infty}<\mu_{2}^{\bullet}+(r-\gamma_{n})\cdot(\mu_{1}^{\bullet}-c^{\infty})\iff\mu_{2}^{\infty}-\gamma_{n}(c^{\infty}-\mu_{1}^{\bullet})<\mu_{2}^{\bullet}-r(c^{\infty}-\mu_{1}^{\bullet})$
if $r-\gamma_{n}<0$, and symmetrically $\mu_{2}^{\infty}-\gamma_{n}(c^{\infty}-\mu_{1}^{\bullet})>\mu_{2}^{\bullet}-r(c^{\infty}-\mu_{1}^{\bullet})$
if $r-\gamma_{n}>0$. In the $r-\gamma_{n}<0$ case, it shows the
agent's belief about the second-period mean of $X_{2}$ conditional
on $X_{1}=c^{\infty}$ is strictly lower than the truth. As the agent
who believes in the model $\Psi(\mu_{1}^{\bullet},\mu_{2}^{\infty};\gamma_{n})$
is indifferent between continuing and stopping after $X_{1}=c^{\infty},$
an agent who believes in the model $\Psi(\mu_{1}^{\bullet},\mu_{2}^{\bullet};r)$
finds it strictly better to continue after $X_{1}=c^{\infty}$. Under
the model $\Psi(\mu_{1}^{\bullet},\mu_{2}^{\bullet};r)$ with $r>-1,$
by Lemma \ref{lem:cutoff_properties} the agent strictly prefers continuing
only at those $c$ with $c<C(\mu_{1}^{\bullet},\mu_{2}^{\bullet};r)=c^{\bullet}$,
which shows $c^{\infty}<c^{\bullet}$. The $r-\gamma_{n}>0$ case
symmetrically leads to the conclusion that $c^{\infty}>c^{\bullet}$.

Suppose both $r,\gamma_{n}<-1$. Then by Proposition \ref{prop:pseudo_true_normal},
$\mu_{2}^{\infty}=\mu_{2}^{\bullet}+(r-\gamma_{n})\cdot(\mu_{1}^{\bullet}-\mathbb{E}[X_{1}\mid X_{1}\ge c^{\infty}])$.
Since $\mathbb{E}[X_{1}\mid X_{1}\ge c^{\infty}]>c^{\infty},$ we
get $\mu_{2}^{\infty}>\mu_{2}^{\bullet}+(r-\gamma_{n})\cdot(\mu_{1}^{\bullet}-c^{\infty})\iff\mu_{2}^{\infty}-\gamma_{n}(c^{\infty}-\mu_{1}^{\bullet})>\mu_{2}^{\bullet}-r(c^{\infty}-\mu_{1}^{\bullet})$
if $r-\gamma_{n}<0$, and symmetrically $\mu_{2}^{\infty}-\gamma_{n}(c^{\infty}-\mu_{1}^{\bullet})<\mu_{2}^{\bullet}-r(c^{\infty}-\mu_{1}^{\bullet})$
if $r-\gamma_{n}>0$. In the $r-\gamma_{n}<0$ case, it shows the
agent's belief about the second-period mean of $X_{2}$ conditional
on $X_{1}=c^{\infty}$ is strictly higher than the truth. As the agent
who believes in the model $\Psi(\mu_{1}^{\bullet},\mu_{2}^{\infty};\gamma_{n})$
is indifferent between continuing and stopping after $X_{1}=c^{\infty},$
an agent who believes in the model $\Psi(\mu_{1}^{\bullet},\mu_{2}^{\bullet};r)$
finds it strictly better to stop after $X_{1}=c^{\infty}$. Under
the model $\Psi(\mu_{1}^{\bullet},\mu_{2}^{\bullet};r)$ with $r<-1,$
by Lemma \ref{lem:cutoff_properties} the agent strictly prefers stopping
only only at those $c$ with $c<C(\mu_{1}^{\bullet},\mu_{2}^{\bullet};r)=c^{\bullet}$,
which shows $c^{\infty}<c^{\bullet}$. The $r-\gamma_{n}>0$ case
symmetrically leads to the conclusion that $c^{\infty}>c^{\bullet}$.
\end{proof}

\subsection{Proof of Proposition \ref{prop:comparative_statics_example}}

I will show a stronger statement. Given a pair of second-period payoff
functions $u_{2}^{H},u_{2}^{L}$, say $u_{2}^{H}$ \emph{payoff dominates}
$u_{2}^{L}$ (abbreviated $u_{2}^{H}\succ u_{2}^{L})$ if for every
$x_{1}\in\mathbb{R},$ $u_{2}^{H}(x_{1},x_{2})\ge u_{2}^{L}(x_{1},x_{2})$
for every $x_{2}\in\mathbb{R},$ and also $u_{2}^{H}(x_{1},x_{2})>u_{2}^{L}(x_{1},x_{2})$
for a positive-measure set of $x_{2}$ in $\mathbb{R}$. It is clear
that increasing $q$ or decreasing $\kappa$ in the statement of Proposition
\ref{prop:comparative_statics_example} leads to a payoff dominating
game. There is a unique steady state for any $(q,\kappa)$ by Proposition
\ref{prop:existence_and_uniqueness} since $r=0$ and $\gamma_{n}>0$.
The next part of Proposition \ref{prop:comparative_statics_example}
is implied by:
\begin{prop}
\label{prop:comparative_statics} Let $r=0$ and $\gamma_{n}>0.$
Suppose both $(u_{1},u_{2}^{H})$ and $(u_{1},u_{2}^{L})$ correspond
to stage games with some $(q,\kappa)$, and that $u_{2}^{H}\succ u_{2}^{L}.$
The steady state of $(u_{1},u_{2}^{H})$ features strictly more optimistic
belief about the second-period fundamental and a strictly higher cutoff
threshold than the steady state of $(u_{1},u_{2}^{L})$.
\end{prop}
\begin{proof}
I require an auxiliary lemma.
\begin{lem}
\label{lem:C_compare_payoff_dom} Suppose both $(u_{1},u_{2}^{H})$
and $(u_{1},u_{2}^{L})$ correspond to stage games with some $(q,\kappa)$,
and that $u_{2}^{H}\succ u_{2}^{L}.$ For all $\mu_{1},\mu_{2}\in\mathbb{R},$
$\gamma>0$, $C_{u_{1},u_{2}^{H}}(\mu_{1},\mu_{2};\gamma)>C_{u_{1},u_{2}^{L}}(\mu_{1},\mu_{2};\gamma)$.
\end{lem}
\begin{proof}
Indifference $c^{L}=C_{u_{1},u_{2}^{L}}(\mu_{1},\mu_{2};\gamma)$
implies $u_{1}(c^{L})=\mathbb{E}_{\tilde{X}_{2}\sim\phi(\cdot\mid\mu_{2}-\gamma(c^{L}-\mu_{1}))}[u_{2}^{L}(c^{L},\tilde{X}_{2})].$
Since $u_{2}^{H}(c^{L},x_{2})\ge u_{2}^{L}(c^{L},x_{2})$ for all
$x_{2}\in\mathbb{R},$ with strict inequality on a positive-measure
set, this shows $u_{1}(c^{L})<\mathbb{E}_{\tilde{X}_{2}\sim\phi(\cdot\mid\mu_{2}-\gamma(c^{L}-\mu_{1}))}[u_{2}^{H}(c^{L},\tilde{X}_{2})].$
The best stopping strategy in the model $\Psi(\mu_{1},\mu_{2};\gamma)$
with the utility functions $(u_{1},u_{2}^{H})$ has a cutoff form
by Lemma \ref{lem:cutoff_properties}. This shows $C_{u_{1},u_{2}^{H}}(\mu_{1},\mu_{2};\gamma)$
is strictly above $c^{L}$.
\end{proof}
Now I return to the proof of Proposition \ref{prop:comparative_statics}.
Say the unique steady states under $(u_{1},u_{2}^{H})$ and $(u_{1},u_{2}^{L})$
are $((\mu_{1}^{\bullet},\mu_{2,H}^{\infty},\gamma_{n}),c_{H}^{\infty})$
and $((\mu_{1}^{\bullet},\mu_{2,L}^{\infty},\gamma_{n}),c_{L}^{\infty})$
respectively. Let $\mathcal{I}_{H},\mathcal{I}_{L}$ be the iteration
maps corresponding to these two stage games, that is to say 
\begin{align*}
\mathcal{I}_{H}(\mu_{2}) & :=\mu_{2}^{*}(C_{u_{1},u_{2}^{H}}(\mu_{1}^{\bullet},\mu_{2};\gamma_{n}))\\
\mathcal{I}_{L}(\mu_{2}) & :=\mu_{2}^{*}(C_{u_{1},u_{2}^{L}}(\mu_{1}^{\bullet},\mu_{2};\gamma_{n})).
\end{align*}
From the proof of Proposition \ref{prop:existence_and_uniqueness},
both $\mathcal{I}_{H}$ and $\mathcal{I}_{L}$ are contraction maps.
Consider their iterates with a starting value of $0$. That is, put
$\mu_{2,H}^{[0]}=0$, $\mu_{2,L}^{[0]}=0$ and let $\mu_{2,H}^{[t]}=\mathcal{I}_{H}(\mu_{2,H}^{[t-1]}),$
$\mu_{2,L}^{[t]}=\mathcal{I}_{L}(\mu_{2,L}^{[t-1]})$ for $t\ge1$.
By property of contraction maps and since the fixed points of the
iteration maps are the steady state beliefs, $\mu_{2,H}^{[t]}\to\mu_{2,H}^{\infty}$
and $\mu_{2,L}^{[t]}\to\mu_{2,L}^{\infty}$.

By induction, I will show $\mu_{2,L}^{[t]}\le\mu_{2,H}^{[t]}$ for
every $t\ge0.$ The base case of $t=0$ is true by definition. If
$\mu_{2,L}^{[T]}\le\mu_{2,H}^{[T]},$ then 
\[
C_{u_{1},u_{2}^{L}}(\mu_{1}^{\bullet},\mu_{2,L}^{[T]};\gamma)\le C_{u_{1},u_{2}^{L}}(\mu_{1}^{\bullet},\mu_{2,H}^{[T]};\gamma)<C_{u_{1},u_{2}^{H}}(\mu_{1}^{\bullet},\mu_{2,H}^{[T]};\gamma).
\]
 The first inequality comes from $C$ being increasing in the second
argument and the inductive hypothesis, while the second inequality
is due to Lemma \ref{lem:C_compare_payoff_dom}. Therefore, $\mathcal{I}_{L}(\mu_{2,L}^{[T]})\le\mathcal{I}_{H}(\mu_{2,H}^{[T]})$
using the fact that $\mu_{2}^{*}$ is increasing by Proposition \ref{prop:pseudo_true_normal},
so $\mu_{2,L}^{[T+1]}\le\mu_{2,H}^{[T+1]}.$

Since weak inequalities are preserved by limits, we have $\mu_{2,H}^{\infty}\ge\mu_{2,L}^{\infty}$.
It is impossible to have $\mu_{2,H}^{\infty}=\mu_{2,L}^{\infty},$
because this would lead to $c_{H}^{\infty}>c_{L}^{\infty}$ by Lemma
\ref{lem:C_compare_payoff_dom}, which in turn implies $\mu_{2,H}^{\infty}=\mu_{2}^{*}(c_{H}^{\infty})>\mu_{2}^{*}(c_{L}^{\infty})=\mu_{2,L}^{\infty}$.
This inequality contradicts $\mu_{2,H}^{\infty}=\mu_{2,L}^{\infty}$.
Therefore, we in fact have $\mu_{2,H}^{\infty}>\mu_{2,L}^{\infty}.$
The conclusion that $c_{H}^{\infty}>c_{L}^{\infty}$ follows from
Lemma \ref{lem:C_compare_payoff_dom} and the fact that $C$ is increases
in its second argument.
\end{proof}

\subsection{Proof of Proposition \ref{prop:pseudo_true_mean_var}}
\selectlanguage{american}%
\begin{proof}
Rewrite Equation (\ref{eq:KL_mean_var}) as\foreignlanguage{english}{{\small{}
\begin{align*}
 & \int_{-\infty}^{\infty}\phi(x_{1}\mid\mu_{1}^{\bullet},(\sigma^{\bullet})^{2})\cdot\ln\left(\frac{\phi(x_{1}\mid\mu_{1}^{\bullet},(\sigma^{\bullet})^{2})}{\phi(x_{1}\mid\mu_{1},\sigma_{1}^{2})}\right)dx_{1}\\
+ & \int_{-\infty}^{c}\phi(x_{1}\mid\mu_{1}^{\bullet},(\sigma^{\bullet})^{2})\cdot\int_{-\infty}^{\infty}\phi(x_{2}\mid\mu_{2}^{\bullet},(\sigma^{\bullet})^{2})\ln\left[\frac{\phi(x_{2}\mid\mu_{2}^{\bullet},(\sigma^{\bullet})^{2})}{\phi(x_{2}\mid\mu_{2}-\gamma(x_{1}-\mu_{1}),\sigma_{2}^{2})}\right]dx_{2}dx_{1}.
\end{align*}
}KL divergence between $\mathcal{N}(\mu_{\text{true}},\sigma_{\text{true}}^{2})$
and $\mathcal{N}(\mu_{\text{model}},\sigma_{\text{model}}^{2})$ is
$\ln\frac{\sigma_{\text{model}}}{\sigma_{\text{true}}}+\frac{\sigma_{\text{true}}^{2}+(\mu_{\text{true}}-\mu_{\text{model}})^{2}}{2\sigma_{\text{model}}^{2}}-\frac{1}{2}$,
so we may simplify the first term and the inner integral of the second
term. {\small{}
\[
\ln\frac{\sigma_{1}}{\sigma^{\bullet}}+\frac{(\mu_{1}-\mu_{1}^{\bullet})^{2}}{2\sigma_{1}^{2}}+\frac{(\sigma^{\bullet})^{2}}{2\sigma_{1}^{2}}-\frac{1}{2}+\int_{-\infty}^{c}\phi(x_{1}\mid\mu_{1}^{\bullet},\sigma^{\bullet})\cdot\left[\ln\frac{\sigma_{2}}{\sigma^{\bullet}}+\frac{(\sigma^{\bullet})^{2}+(\mu_{2}-\gamma(x_{1}-\mu_{1})-\mu_{2}^{\bullet})^{2}}{2\sigma_{2}^{2}}-\frac{1}{2}\right]dx_{1}.
\]
}Dropping terms not dependent on any of the four variables gives a
simplified version of the objective, 
\begin{align*}
\xi(\mu_{1},\mu_{2},\sigma_{1},\sigma_{2}):= & \ln\frac{\sigma_{1}}{\sigma^{\bullet}}+\frac{(\mu_{1}-\mu_{1}^{\bullet})^{2}}{2\sigma_{1}^{2}}+\frac{(\sigma^{\bullet})^{2}}{2\sigma_{1}^{2}}\\
 & +\int_{-\infty}^{c}\phi(x_{1}\mid\mu_{1}^{\bullet},(\sigma^{\bullet})^{2})\cdot\left[\ln\frac{\sigma_{2}}{\sigma^{\bullet}}+\frac{(\sigma^{\bullet})^{2}+(\mu_{2}-\gamma(x_{1}-\mu_{1})-\mu_{2}^{\bullet})^{2}}{2\sigma_{2}^{2}}\right]dx_{1}.
\end{align*}
Differentiating under the integral sign, 
\[
\frac{\partial\xi}{\partial\mu_{2}}=\int_{-\infty}^{c}\phi(x_{1}\mid\mu_{1}^{\bullet},(\sigma^{\bullet})^{2})\cdot\left[\frac{(\mu_{2}-\gamma(x_{1}-\mu_{1})-\mu_{2}^{\bullet})}{\sigma_{2}^{2}}\right]dx_{1}
\]
\begin{align*}
\frac{\partial\xi}{\partial\mu_{1}} & =\frac{(\mu_{1}-\mu_{1}^{\bullet})}{\sigma_{1}^{2}}+\gamma\int_{-\infty}^{c}\phi(x_{1}\mid\mu_{1}^{\bullet},(\sigma^{\bullet})^{2})\cdot\left[\frac{(\mu_{2}-\gamma(x_{1}-\mu_{1})-\mu_{2}^{\bullet})}{\sigma_{2}^{2}}\right]dx_{1}=\frac{(\mu_{1}-\mu_{1}^{\bullet})}{\sigma_{1}^{2}}+\gamma\frac{\partial\xi}{\partial\mu_{2}}.
\end{align*}
At FOC $(\mu_{1}^{*},\mu_{2}^{*},\sigma_{1}^{*},\sigma_{2}^{*}),$
we have $\frac{\partial\xi}{\partial\mu_{2}}(\mu_{1}^{*},\mu_{2}^{*},\sigma_{1}^{*},\sigma_{2}^{*})=0,$
hence $\mu_{1}^{*}=\mu_{1}^{\bullet}$. Similar arguments as before
then establish $\mu_{2}^{*}=\mu_{2}^{\bullet}-\gamma\left(\mu_{1}^{\bullet}-\mathbb{E}\left[X_{1}\mid X_{1}\le c\right]\right),$
where expectation is taken with respect to the true distribution of
$X_{1}$ (with the true variance $(\sigma^{\bullet})^{2}$). Then,
$\frac{\partial\xi}{\partial\sigma_{1}}(\mu_{1}^{*},\mu_{2}^{*},\sigma_{1}^{*},\sigma_{2}^{*})=\frac{1}{(\sigma_{1}^{*})}-\frac{(\sigma^{\bullet})^{2}}{(\sigma_{1}^{*})^{3}}=0,$
this gives $\sigma_{1}^{*}=\sigma^{\bullet}$ (since $\sigma_{1}^{*}\ge0).$}

\selectlanguage{english}%
Finally, from the FOC for $\sigma_{2},$ 
\[
\int_{-\infty}^{c}\phi(x_{1};\mu_{1}^{\bullet},(\sigma^{\bullet})^{2})\cdot\left[\frac{1}{\sigma_{2}^{*}}-\frac{(\sigma^{\bullet})^{2}+(\mu_{2}^{*}-\gamma(x_{1}-\mu_{1}^{*})-\mu_{2}^{\bullet})^{2}}{(\sigma_{2}^{*})^{3}}\right]dx_{1}=0.
\]
Substituting in values of $\mu_{1}^{*},\mu_{2}^{*}$ already solved
for, 
\begin{align*}
(\sigma_{2}^{*})^{2} & =(\sigma^{\bullet})^{2}+\mathbb{E}[(\mu_{2}^{*}-\gamma(X_{1}-\mu_{1}^{\bullet})-\mu_{2}^{\bullet})^{2}|X_{1}\le c]\\
 & =(\sigma^{\bullet})^{2}+\mathbb{E}[(\mu_{2}^{\bullet}-\gamma\left(\mu_{1}^{\bullet}-\mathbb{E}\left[X_{1}\mid X_{1}\le c\right]\right)-\gamma(X_{1}-\mu_{1}^{\bullet})-\mu_{2}^{\bullet})^{2}|X_{1}\le c]\\
 & =(\sigma^{\bullet})^{2}+\gamma^{2}\mathbb{E}\left[\left[(X_{1}-\mu_{1}^{\bullet})-\left(\mathbb{E}\left[X_{1}\mid X_{1}\le c\right]-\mu_{1}^{\bullet}\right)\right]^{2}|X_{1}\le c\right]=(\sigma^{\bullet})^{2}+\gamma^{2}\text{Var}[X_{1}|X_{1}\le c]
\end{align*}
as desired. Finally, $\sigma_{2}^{*}(c)$ is an increasing function
of $c$ because $\text{Var}[X_{1}\mid X_{1}\le c]$ increases in $c$
for $X_{1}$ Gaussian \citep{mailhot1985propriete}.
\end{proof}
\selectlanguage{english}%

\subsection{Proof of Proposition \ref{prop:AB_societies}}

\selectlanguage{american}%
I start with a lemma that says if the decision problem is convex,
a stronger belief in fictitious variation increases the subjectively
optimal cutoff threshold.
\begin{lem}
\label{lem:indifference_different_var} Suppose that under the feasible
model $\Psi(\mu_{1},\mu_{2},\sigma_{1}^{2},\sigma_{2}^{2};\gamma)$,
the agent is indifferent between stopping at $c$ and continuing.
Suppose $\hat{\sigma}_{2}^{2}>\sigma_{2}^{2}.$ Then if $x_{2}\mapsto u_{2}(c,x_{2})$
is convex with strict convexity for $x_{2}$ in a positive-measure
set, then under the feasible model $\Psi(\mu_{1},\mu_{2},\sigma_{1}^{2},\hat{\sigma}_{2}^{2};\gamma)$
the agent strictly prefers continuing at $c$.
\end{lem}
\begin{proof}
Indifference at $x_{1}=c$ under $\Psi(\mu_{1},\mu_{2},\sigma_{1}^{2},\sigma_{2}^{2};\gamma)$
implies \\ $u_{1}(c)=\mathbb{E}_{X_{2}\sim\mathcal{N}(\mu_{2}-\gamma(x_{1}-\mu_{1}),\sigma_{2}^{2})}[u_{2}(c,X_{2})].$
When hypothesis is satisfied, \\ $\mathbb{E}_{X_{2}\sim\mathcal{N}(\mu_{2}-\gamma(x_{1}-\mu_{1}),\sigma_{2}^{2})}[u_{2}(c,X_{2})]<\mathbb{E}_{X_{2}\sim\mathcal{N}(\mu_{2}-\gamma(x_{1}-\mu_{1}),\hat{\sigma}_{2}^{2})}[u_{2}(c,X_{2})]$
since $\hat{\sigma}_{2}^{2}>\sigma_{2}^{2}$ implies that $\mathcal{N}(\mu_{2}-\gamma(x_{1}-\mu_{1}),\hat{\sigma}_{2}^{2})$
is a strict mean-preserving spread of $\mathcal{N}(\mu_{2}-\gamma(x_{1}-\mu_{1}),\sigma_{2}^{2}).$
The RHS is the expected continuation payoff under model $\Psi(\mu_{1},\mu_{2},\sigma_{1}^{2},\hat{\sigma}_{2}^{2};\gamma)$,
so the agent strictly prefers continuing when $X_{1}=c.$
\end{proof}
Now I give the proof of Proposition \foreignlanguage{english}{\ref{prop:AB_societies}.}
\selectlanguage{english}%
\begin{proof}
By the proof of Proposition \ref{prop:existence_and_uniqueness},
$\mathcal{I}(\mu_{2};\gamma):=\mu_{2}^{*}(C(\mu_{1}^{\bullet},\mu_{2},\gamma))$
for society A is a contraction map in $\mu_{2}$. By way of contradiction
suppose $c^{B}\le c^{A}$. Then $\mu_{2}^{B}\le\mu_{2}^{A}$ by Proposition
\ref{prop:pseudo_true_mean_var}. In society A, $C(\mu_{1}^{\bullet},\mu_{2}^{B};\gamma)<c^{B}$
by Lemma \ref{lem:indifference_different_var}, as there is strictly
positive probability of recall. This shows $\mathcal{I}(\mu_{2}^{B};\gamma)<\mu_{2}^{B}$.
In fact, for the $t$-times iteration we have $\mathcal{I}^{(t)}(\mu_{2}^{B};\gamma)\le\mathcal{I}(\mu_{2}^{B};\gamma)<\mu_{2}^{B}$,
which means $\mathcal{I}$ has a fixed point strictly smaller than
$\mu_{2}^{A}.$ This contradicts $\mu_{2}^{A}$ being the only fixed
point of $\mathcal{I}.$ Hence we must have $\mu_{2}^{B}>\mu_{2}^{A}$
and $c^{B}>c^{A}$. We have $\sigma_{2}^{B}=\sigma_{2}^{*}(c^{B}),$
which is larger than $\sigma_{2}^{*}(c^{A})$ by combining $c^{B}>c^{A}$
with Proposition \ref{prop:pseudo_true_mean_var}.
\end{proof}

\subsection{Proof of Proposition \ref{prop:one-by-one}}

I introduce some new notation. Abbreviate $\APLbox:=[\underline{\mu}_{1},\bar{\mu}_{1}]\times[\underline{\mu}_{2},\bar{\mu}_{2}]$.
Let $\gamma=\gamma_{n}$ and let $li(\mu_{2})$ be the line in $\mathbb{R}^{2}$
with slope $-\gamma$ that passes through the point $(\mu_{1}^{\bullet},\mu_{2})$.
There are some minimal and maximal $\underline{\mu}_{2}^{\circ}$
and $\bar{\mu}_{2}^{\circ}$ so that $li(\underline{\mu}_{2}^{\circ})\cap\APLbox\ne\varnothing$
and $li(\bar{\mu}_{2}^{\circ})\cap\APLbox\ne\varnothing$. Finally,
for $\mu_{2}^{l}<\mu_{2}^{h}$, let $\lozenge[\mu_{2}^{l},\mu_{2}^{h}]:=\cup_{\mu_{2}\in[\mu_{2}^{l},\mu_{2}^{h}]}li(\mu_{2})$.
So we have $\APLbox\subseteq\lozenge[\underline{\mu}_{2}^{\circ},\bar{\mu}_{2}^{\circ}]$.
Similarly the half-open versions $\lozenge[\mu_{2}^{l},\mu_{2}^{h})$
and $\lozenge(\mu_{2}^{l},\mu_{2}^{h}]$ are defined as the unions
$\cup_{\mu_{2}\in[\mu_{2}^{l},\mu_{2}^{h})}li(\mu_{2})$ and $\cup_{\mu_{2}\in(\mu_{2}^{l},\mu_{2}^{h}]}li(\mu_{2})$.
(The picture below illustrates a case with $\gamma>0$.)
\begin{center}
\includegraphics[scale=0.35]{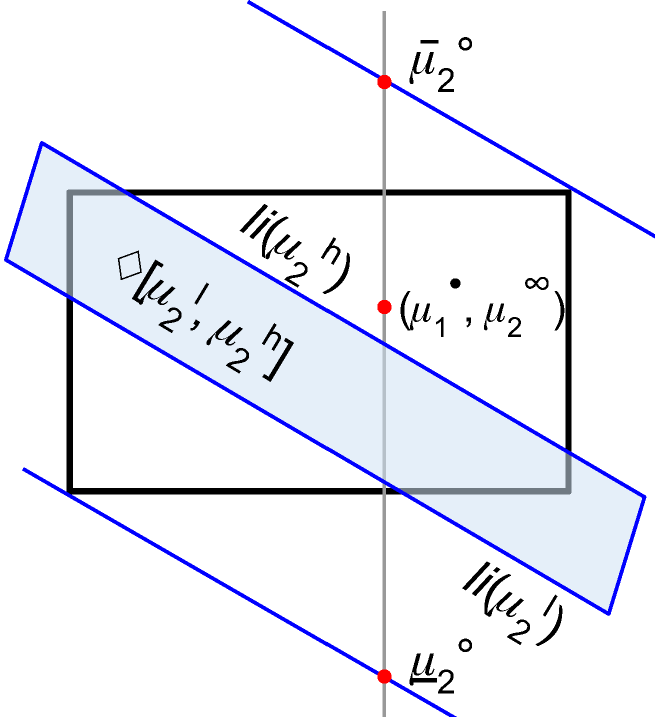}
\par\end{center}

\subsubsection{Preliminary Results}

First, I consider how the predicted second-period payoff after $X_{1}=x_{1}$
depends on the parameters of the feasible model $\Psi(\mu_{1},\mu_{2};\gamma)$.
\begin{lem}
\label{lem:C_equivalence} For every $\mu_{1},\mu_{2},x_{1}\in\mathbb{R}$,
the conditional distribution $X_{2}|X_{1}=x_{1}$ is the same under
$\Psi(\mu_{1}^{\bullet},\mu_{2}+\gamma(\mu_{1}-\mu_{1}^{\bullet});\gamma)$
and $\Psi(\mu_{1},\mu_{2};\gamma)$. So in particular, $C(\mu_{1},\mu_{2};\gamma)=C(\mu_{1}^{\bullet},\mu_{2}+\gamma(\mu_{1}-\mu_{1}^{\bullet});\gamma)$.
\end{lem}
\begin{proof}
Under the feasible model $\Psi(\mu_{1}^{\bullet},\mu_{2}+\gamma(\mu_{1}-\mu_{1}^{\bullet});\gamma)$,
the conditional density of $X_{2}$ given $X_{1}=x_{1}$ is $\phi(\cdot\mid\mu_{2}+\gamma(\mu_{1}-\mu_{1}^{\bullet})-\gamma(x_{1}-\mu_{1}^{\bullet}))$,
which simplifies to $\phi(\cdot\mid\mu_{2}-\gamma(x_{1}-\mu_{1}))$.
It is easy to see that this is also the expression for the same conditional
density under $\Psi(\mu_{1},\mu_{2};\gamma)$.

Suppose $C(\mu_{1},\mu_{2};\gamma)=c.$ This implies $u_{1}(c)=\mathbb{E}_{\Psi(\mu_{1},\mu_{2};\gamma)}[u_{2}(c,X_{2})\mid X_{1}=c].$
But by the equivalence of conditional distribution given above, 
\[
u_{1}(c)=\mathbb{E}_{\Psi(\mu_{1}^{\bullet},\mu_{2}+\gamma(\mu_{1}-\mu_{1}^{\bullet});\gamma)}[u_{2}(c,X_{2})\mid X_{1}=c].
\]
This means $c$ is also the indifference threshold for the model $\Psi(\mu_{1}^{\bullet},\mu_{2}+\gamma(\mu_{1}-\mu_{1}^{\bullet});\gamma)$.
\end{proof}
As a corollary, this lemma shows the restriction to cutoff strategies
is without loss, and that $\tilde{C}_{t}$ is well defined. That is,
for any belief given by a density on $\APLbox$, there exists a cutoff
strategy that is weakly optimal among the class of all stopping strategies,
and further this cutoff strategy is strictly optimal among the class
of cutoff strategies. This is because for any $x_{1}\in\mathbb{R}$
and any density $\tilde{m}$ on $\APLbox$, 
\begin{align*}
 & \int_{\APLbox}\mathbb{E}_{\Psi(\mu_{1},\mu_{2};\gamma)}[u_{2}(x_{1},X_{2})\mid X_{1}=x_{1}]\cdot\tilde{m}(\mu_{1},\mu_{2})d(\mu_{1},\mu_{2})\\
 & =\int_{\underline{\mu}_{2}^{\circ}}^{\bar{\mu}_{2}^{\circ}}\mathbb{E}_{\Psi(\mu_{1}^{\bullet},\mu_{2};\gamma)}[u_{2}(x_{1},X_{2})\mid X_{1}=x_{1}]\cdot\tilde{m}^{V}(\mu_{2})d\mu_{2}
\end{align*}
where $\tilde{m}^{V}(\mu_{2})$ is the integral of $\tilde{m}(\mu_{1},\mu_{2})$
over $li(\mu_{2})$. This equality holds because by Lemma \ref{lem:C_equivalence},
all fundamentals on $li(\mu_{2})$ imply the same continuation payoff
after $X_{1}=x_{1}$ as the fundamentals $(\mu_{1}^{\bullet},\mu_{2}).$
\begin{lem}
\label{lem:cutoff_optimal_single_agent}If $\gamma>-1,$ then the
function

\[
x_{1}\mapsto u_{1}(x_{1})-\int_{\underline{\mu}_{2}^{\circ}}^{\bar{\mu}_{2}^{\circ}}\mathbb{E}_{\Psi(\mu_{1}^{\bullet},\mu_{2};\gamma)}[u_{2}(x_{1},X_{2})\mid X_{1}=x_{1}]\tilde{m}^{V}(\mu_{2})d\mu_{2}
\]
 is strictly increasing, continuous, and crosses 0.
\end{lem}
\begin{proof}
Let $d\nu(\mu_{2})=\tilde{m}^{V}(\mu_{2})d\mu_{2}$. Consider the
payoff difference between accepting $x_{1}$ and continuing under
belief $\nu$, 
\[
D(x_{1};\nu):=u_{1}(x_{1})-\int\mathbb{E}_{X_{2}\sim\phi(\cdot\mid\mu_{2}-\gamma(x_{1}-\mu_{1}^{\bullet}))}[u_{2}(x_{1},X_{2})]d\nu(\mu_{2}).
\]
Note that $D(x_{1},\nu)=\int D(x_{1};\mu_{1}^{\bullet},\mu_{2},\gamma)d\nu(\mu_{2})$.
When $\gamma>-1,$ Lemma \ref{lem:cutoff_properties} shows that for
every $\mu_{2}\in\mathbb{R}$, $D(x_{1};\mu_{1}^{\bullet},\mu_{2},\gamma)$
is strictly increasing in $x_{1}$. Hence the same must hold for $D(x_{1},\nu).$

Lemma \ref{lem:cutoff_properties} shows there exists some $x_{1}'\in\mathbb{R}$
so that $D(x_{1}';\mu_{1}^{\bullet},\underline{\mu}_{2},\gamma)<0,$
and that there exists some $x_{1}''\in\mathbb{R}$ satisfying $D(x_{1}'';\mu_{1}^{\bullet},\bar{\mu}_{2},\gamma)>0$.
Since $u_{2}$ increases in its second argument, we also get $D(x_{1}';\mu_{1}^{\bullet},\mu_{2},\gamma)<0$
and $D(x_{1}'';\mu_{1}^{\bullet},\mu_{2},\gamma)>0$ for all $\mu_{2}\in[\underline{\mu}_{2},\bar{\mu}_{2}]$.
This implies $D(x_{1}';\nu)<0$ and $D(x_{1}'';\nu)>0$, as $\nu$
is supported on (a subset of) $[\underline{\mu}_{2},\bar{\mu}_{2}].$

To show $D(x_{1};\nu)$ is continuous in $x_{1}$, fix some $\bar{x}_{1}.$
Let $\pi(\mu_{2})$ represent the expectation of the absolute value
of a normal random variable with mean $\mu_{2}-\gamma((\bar{x}_{1}-1)-\mu_{1}^{\bullet})$
and variance $\sigma^{2}$. Here $\pi(\mu_{2})$ is bounded by a constant
plus a linear function of $\mu_{2}$ as we vary $\mu_{2}.$ For $|x_{1}-\bar{x}_{1}|\le1$,
\[
|\mathbb{E}_{X_{2}\sim\phi(\cdot\mid\mu_{2}-\gamma(x_{1}-\mu_{1}^{\bullet}))}[u_{2}(x_{1},X_{2})]|\le q(|\bar{x}_{1}|+1+\pi(\mu_{2}))+(1-q)\pi(\mu_{2})+|\kappa|,
\]
 and the RHS is a positive and integrable function with respect to
$d\nu(\mu_{2}).$ For a sequence $x_{1}^{(n)}\to\bar{x}_{1},$ the
integrand in $D(x_{1}^{(n)};\nu)$ is dominated by $q(|\bar{x}_{1}|+1+\pi(\mu_{2}))+(1-q)\pi(\mu_{2})+|\kappa|$
for all large enough $n$, so by dominated convergence theorem, $D(x_{1}^{(n)};\nu)\to D(\bar{x}_{1};\nu).$
So, $D(\cdot;\nu)$ is continuous.
\end{proof}
Now, the key step is to separate the two-dimensional inference problem
into a pair of one-dimensional problems.

\subsubsection{Learning $\mu_{1}^{\bullet}$}

I define the stochastic process of data log-likelihood (for a given
fundamental). For each $\mu_{1},\mu_{2}\in\text{supp}(m_{0})$, let
$\ell_{t}(\mu_{1},\mu_{2})(\omega)$ be the log likelihood that the
fundamentals are $(\mu_{1},\mu_{2})$ and histories $(\tilde{H}_{s})_{s\le t}(\omega)$
are generated by the end of round $t$. It is given by

\[
\ell_{t}(\mu_{1},\mu_{2})(\omega):=\ln(m_{0}(\mu_{1},\mu_{2}))+\sum_{s=1}^{t}\ln(\text{lik}(\tilde{H}_{s}(\omega);\mu_{1},\mu_{2}))
\]
where $\text{lik}(x_{1},\varnothing;\mu_{1},\mu_{2}):=\phi(x_{1}\mid\mu_{1})$
and $\text{lik}(x_{1},x_{2};\mu_{1},\mu_{2}):=\phi(x_{1}\mid\mu_{1})\cdot\phi(x_{2}\mid\mu_{2}-\gamma(x_{1}-\mu_{1}))$.
Let $g_{1}(\cdot)$ and $g_{2}(\cdot\mid x_{1})$ be the true densities
for the distributions of $X_{1}$ and $X_{2}|(X_{1}=x_{1})$, incorporating
the true parameters $\mu_{1}^{\bullet},\mu_{2}^{\bullet}$, and $r.$
Let $f_{2}(z)$ be the Gaussian distribution with the mean $\mu_{2}^{\bullet},$
variance $\sigma^{2}$ evaluated at $z$. By simple algebra, we may
expand
\begin{align*}
\ell_{t}(\mu_{1},\mu_{2})(\omega) & =\ln(m_{0}(\mu_{1},\mu_{2}))+\sum_{s=1}^{t}\ln[g_{1}(X_{1,s}(\omega)-\mu_{1}+\mu_{1}^{\bullet})]\\
 & +\sum_{s=1}^{t}\boldsymbol{1}\{X_{1,s}(\omega)\le\tilde{C}_{s}(\omega)\}\cdot\ln\left[f_{2}(X_{2,s}(\omega)-\mu_{2}+\mu_{2}^{\bullet}+\gamma(X_{1,s}(\omega)-\mu_{1}))\right]
\end{align*}

I first establish that, without knowing anything about the process
$(\tilde{C}_{t}),$ we can conclude agents either learn $\mu_{1}^{\bullet}$
arbitrarily well, or they believe in a boundary value of $\mu_{2}$
— that is, either $\mu_{2}=\underline{\mu}_{2}$ or $\mu_{2}=\bar{\mu}_{2}$.
(We can later rule out these boundary beliefs of $\mu_{2}$).
\begin{lem}
\label{lem:learning_mu1} Let $\epsilon>0$ be given. If $\gamma>0$,
then 
\[
\lim_{t\to\infty}\tilde{M}_{t}\{\APLbox\cap(([\mu_{1}^{\bullet}-\epsilon,\mu_{1}^{\bullet}+\epsilon]\times\mathbb{R})\cup([\underline{\mu}_{1},\mu_{1}^{\bullet}]\times[\underline{\mu}_{2},\underline{\mu}_{2}+\epsilon])\cup([\mu_{1}^{\bullet},\bar{\mu}_{1}]\times[\bar{\mu}_{2}-\epsilon,\bar{\mu}_{2}]))\}=1.
\]
 If $\gamma\le0,$ then 
\[
\lim_{t\to\infty}\tilde{M}_{t}\{\APLbox\cap(([\mu_{1}^{\bullet}-\epsilon,\mu_{1}^{\bullet}+\epsilon]\times\mathbb{R})\cup([\underline{\mu}_{1},\mu_{1}^{\bullet}]\times[\bar{\mu}_{2}-\epsilon,\bar{\mu}_{2}])\cup([\mu_{1}^{\bullet},\bar{\mu}_{1}]\times[\underline{\mu}_{2},\underline{\mu}_{2}+\epsilon]))\}=1.
\]
\end{lem}
\begin{proof}
First calculate the directional derivative $\nabla_{v}\frac{1}{t}\ell_{t}(\mu_{1},\mu_{2}),$
where \\ $v=(1/\sqrt{1+\gamma^{2}},-\gamma/\sqrt{1+\gamma^{2}})^{\prime}$
is the unit vector with slope $-\gamma$. We have 
\begin{align*}
\frac{\partial(\ell_{t}/t)}{\partial\mu_{1}}(\mu_{1},\mu_{2})= & \frac{1}{t}\frac{D_{1}m_{0}(\mu_{1},\mu_{2})}{m_{0}(\mu_{1},\mu_{2})}-\frac{1}{t}\sum_{s=1}^{t}\frac{g_{1}^{'}(X_{1,s}-\mu_{1}+\mu_{1}^{\bullet})}{g_{1}(X_{1,s}-\mu_{1}+\mu_{1}^{\bullet})}\\
 & -\frac{\gamma}{t}\sum_{s=1}^{t}\boldsymbol{1}\{X_{1,s}\le\tilde{C}_{s}\}\cdot\lambda(X_{2,s}-\mu_{2}+\mu_{2}^{\bullet}+\gamma(X_{1,s}-\mu_{1}))
\end{align*}

\[
\frac{\partial(\ell_{t}/t)}{\partial\mu_{2}}(\mu_{1},\mu_{2})=\frac{1}{t}\frac{D_{2}m_{0}(\mu_{1},\mu_{2})}{m_{0}(\mu_{1},\mu_{2})}-\frac{1}{t}\sum_{s=1}^{t}\boldsymbol{1}\{X_{1,s}\le\tilde{C}_{s}\}\cdot\lambda(X_{2,s}-\mu_{2}+\mu_{2}^{\bullet}+\gamma(X_{1,s}-\mu_{1})),
\]
 where $D_{1}m_{0}$ and $D_{2}m_{0}$ are the two partial derivatives
of $m_{0}$, and $\lambda(\cdot):=f_{2}^{'}(\cdot)/f_{2}(\cdot)$.
At every $\omega$ and every $(\mu_{1},\mu_{2}),$ note the last summand
in $\frac{\partial(\ell_{t}/t)}{\partial\mu_{1}}$ is $\gamma$ times
the last summand in $\frac{\partial(\ell_{t}/t)}{\partial\mu_{2}}$.
Therefore, 
\begin{align*}
\nabla_{v}\frac{1}{t}\ell_{t}(\mu_{1},\mu_{2})= & \frac{-1}{\sigma^{2}\sqrt{1+\gamma^{2}}}\left(\frac{1}{t}\sum_{s=1}^{t}\frac{g_{1}^{'}(X_{1,s}-\mu_{1}+\mu_{1}^{\bullet})}{g_{1}(X_{1,s}-\mu_{1}+\mu_{1}^{\bullet})}\right)+\frac{1}{t\sqrt{1+\gamma^{2}}}\frac{1}{t}\frac{D_{1}m_{0}(\mu_{1},\mu_{2})}{m_{0}(\mu_{1},\mu_{2})}\\
 & -\frac{\gamma}{t\sqrt{1+\gamma^{2}}}\frac{D_{2}m_{0}(\mu_{1},\mu_{2})}{m_{0}(\mu_{1},\mu_{2})}.
\end{align*}
Since $m_{0},D_{1}m_{0},D_{2}m_{0}$ are continuous on the compact
set $\APLbox$, there exists some $0<B<\infty$ so that $|\frac{D_{1}m_{0}(\mu_{1},\mu_{2})}{m_{0}(\mu_{1},\mu_{2})}|<B$
and $|\frac{D_{2}m_{0}(\mu_{1},\mu_{2})}{m_{0}(\mu_{1},\mu_{2})}|<B$
for all $(\mu_{1},\mu_{2})\in\APLbox$. Pick any $\epsilon^{'}>0.$
We have that for every $\omega,$

\[
\inf_{(\mu_{1},\mu_{2})\in\APLbox_{L}}\left[\left(\nabla_{v}\frac{1}{t}\ell_{t}(\mu_{1},\mu_{2})\right)+\frac{1}{\sigma^{2}\sqrt{1+\gamma^{2}}}\left(\frac{1}{t}\sum_{s=1}^{t}\frac{g_{1}^{'}(X_{1,s}-\mu_{1}+\mu_{1}^{\bullet})}{g_{1}(X_{1,s}-\mu_{1}+\mu_{1}^{\bullet})}\right)\right]\ge-\frac{1}{t}\frac{(1+\gamma)}{\sqrt{1+\gamma^{2}}}B,
\]
 where $\APLbox_{L}:=[\underline{\mu}_{1},\mu_{1}^{\bullet}-2\epsilon^{'}]\times[\underline{\mu}_{2}+\gamma\epsilon^{'},\bar{\mu}_{2}]$
when $\gamma>0$ and $\APLbox_{L}:=[\underline{\mu}_{1},\mu_{1}^{\bullet}-2\epsilon^{'}]\times[\underline{\mu}_{2},\bar{\mu}_{2}+\gamma\epsilon^{'}]$
when $\gamma\le0$ is a sub-rectangle to the left of $\mu_{1}^{\bullet}-\epsilon^{'}$.
By law of large numbers applied to the i.i.d. sequence $(\frac{g_{1}^{'}(X_{1,s}-(\mu_{1}^{\bullet}-\epsilon)+\mu_{1}^{\bullet})}{g_{1}(X_{1,s}-(\mu_{1}^{\bullet}-\epsilon)+\mu_{1}^{\bullet})})_{s\ge1},$
almost surely 
\[
\frac{1}{t}\sum_{s=1}^{t}\frac{g_{1}^{'}(X_{1,s}-(\mu_{1}^{\bullet}-\epsilon)+\mu_{1}^{\bullet})}{g_{1}(X_{1,s}-(\mu_{1}^{\bullet}-\epsilon)+\mu_{1}^{\bullet})}\to\mathbb{E}_{X\sim g_{1}}\left[\frac{g_{1}^{'}(X_{1}+\epsilon)}{g_{1}(X_{1}+\epsilon)}\right].
\]
Since $\mathbb{E}_{X\sim g_{1}}\left[\frac{g_{1}^{'}(X_{1})}{g_{1}(X_{1})}\right]=0$
and since $z\mapsto\frac{g_{1}^{'}(z)}{g_{1}(z)}=\frac{d}{dz}(\ln(g_{1}(z))$
is strictly decreasing by log-concavity of the normal distribution,
there is some $\delta>0$ so that $\mathbb{E}_{X\sim g_{1}}\left[\frac{g_{1}^{'}(X_{1}+\epsilon^{'})}{g_{1}(X_{1}+\epsilon^{'})}\right]=-\delta.$
Furthermore, for any $\mu_{1}\le\mu_{1}^{\bullet}-\epsilon^{'},$
then for any $x_{1}\in\mathbb{R},$ $\frac{g_{1}^{'}(x_{1}-\mu_{1}+\mu_{1}^{\bullet})}{g_{1}(x_{1}-\mu_{1}+\mu_{1}^{\bullet})}\le\frac{g_{1}^{'}(x_{1}+\epsilon^{'})}{g_{1}(x_{1}-\epsilon^{'})}.$
Along any $\omega$ where $\frac{1}{t}\sum_{s=1}^{t}\frac{g_{1}^{'}(X_{1,s}-(\mu_{1}^{\bullet}-\epsilon^{'})+\mu_{1}^{\bullet})}{g_{1}(X_{1,s}-(\mu_{1}^{\bullet}-\epsilon^{'})+\mu_{1}^{\bullet})}\to-\delta$,
we therefore also have 
\[
\limsup_{t\to\infty}\sup_{\mu_{1}\ge\mu_{1}^{\bullet}-\epsilon^{'}}\frac{1}{t}\sum_{s=1}^{t}\frac{g_{1}^{'}(X_{1,s}-\mu_{1}+\mu_{1}^{\bullet})}{g_{1}(X_{1,s}-\mu_{1}+\mu_{1}^{\bullet})}\le-\delta.
\]

Therefore almost surely 
\[
\liminf_{t\to\infty}\inf_{(\mu_{1},\mu_{2})\in\APLbox_{L}}\left(\nabla_{v}\frac{1}{t}\ell_{t}(\mu_{1},\mu_{2})\right)\ge\frac{\delta}{\sigma^{2}\sqrt{1+\gamma^{2}}}.
\]

\begin{center}\includegraphics[scale=0.6]{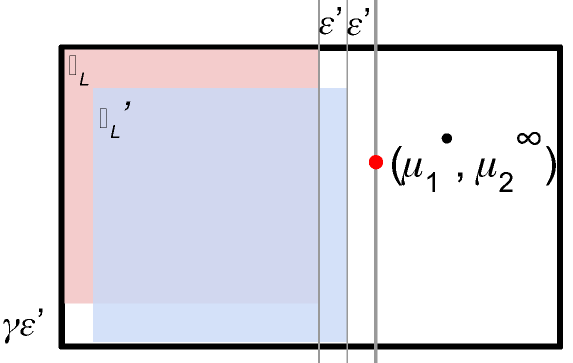}\end{center}

Let $\APLbox_{L}^{'}$ be $\APLbox_{L}$ shifted by the vector $(\epsilon^{'},-\gamma\epsilon^{'})$,
so it remains in $\APLbox$ and at least $\epsilon^{'}$ to the left
of $\mu_{1}^{\bullet}.$ That is, $\APLbox_{L}^{'}:=[\underline{\mu}_{1}+\epsilon^{'},\mu_{1}^{\bullet}-\epsilon^{'}]\times[\underline{\mu}_{2},\bar{\mu}_{2}-\gamma\epsilon^{'}]$
if $\gamma>0$ (illustrated above) and $\APLbox_{L}^{'}:=[\underline{\mu}_{1}+\epsilon^{'},\mu_{1}^{\bullet}-\epsilon^{'}]\times[\underline{\mu}_{2}-\gamma\epsilon^{'},\bar{\mu}_{2}]$
if $\gamma\le0$. I will show that $\lim_{t\to\infty}\tilde{M}_{t}(\APLbox_{L})=0$
almost surely. The idea is we can map every point in $\APLbox_{L}$
to another point in $\APLbox_{L}^{'}$ in the direction of $v$. For
every point, its image under the map will have much higher posterior
probability, since we have a uniform, strictly positive lowerbound
on the directional derivative of log-likelihood $\ell_{t}$ in the
direction of $v$. 
\begin{align*}
\tilde{M}_{t}(\APLbox_{L}) & =\int_{\APLbox_{L}}\tilde{m}_{t}(\mu_{1},\mu_{2})d\mu\\
 & =\int_{\APLbox_{L}^{'}}\tilde{m}_{t}(\mu_{1},\mu_{2})\cdot\frac{\tilde{m}_{t}(\mu_{1}-\epsilon^{'},\mu_{2}-\gamma\epsilon^{'})}{\tilde{m}_{t}(\mu_{1},\mu_{2})}d\mu\\
 & =\int_{\APLbox_{L}^{'}}\tilde{m}_{t}(\mu_{1},\mu_{2})\exp(\ell_{t}(\mu_{1}-\epsilon^{'},\mu_{2}-\gamma\epsilon^{'})-\ell_{t}(\mu_{1},\mu_{2}))d\mu\\
 & =\int_{\APLbox_{L}^{'}}\tilde{m}_{t}(\mu_{1},\mu_{2})\exp(-\int_{0}^{\epsilon}\nabla_{v}\ell_{t}(\mu_{1}-\epsilon^{'}+z,\mu_{2}-\gamma\epsilon^{'}+\gamma z)dz)d\mu
\end{align*}
Almost surely, 
\[
\liminf_{t\to\infty}\inf_{(\mu_{1},\mu_{2})\in\APLbox_{L}^{'},z\in[0,\epsilon^{'}]}\left(\nabla_{v}\ell_{t}(\mu_{1}-\epsilon^{'}+z,\mu_{2}-\gamma\epsilon^{'}+\gamma z)\right)\ge\frac{t\delta}{\sigma^{2}\sqrt{1+\gamma^{2}}},
\]
so almost surely 
\[
\limsup_{t\to\infty}\tilde{M}_{t}(\APLbox_{L})\le\limsup_{t\to\infty}\exp(-\frac{\epsilon^{'}t\delta}{\sigma^{2}\sqrt{1+\gamma^{2}}})\cdot\int_{\APLbox_{L}^{'}}\tilde{m}_{t}(\mu_{1},\mu_{2})d\mu.
\]
But for every $\omega$ and $t$, the RHS is bounded above by $\exp(-\frac{\epsilon^{'}t\delta}{\sigma^{2}\sqrt{1+\gamma^{2}}})$,
which tends to 0 as $t\to\infty$ since $\epsilon^{'},\delta>0$.
So in fact $\tilde{M}_{t}(\APLbox_{L})\to0$ almost surely.

Since the choice of $\epsilon^{'}>0$ was arbitrary, this shows for
every $\epsilon>0,$ almost surely $\lim_{t\to\infty}\tilde{M}_{t}([\underline{\mu}_{1},\mu_{1}^{\bullet}-\epsilon]\times[\underline{\mu}_{2}+\epsilon,\bar{\mu}_{2}])=0$
when $\gamma>0$ and $\lim_{t\to\infty}\tilde{M}_{t}([\underline{\mu}_{1},\mu_{1}^{\bullet}-\epsilon]\times[\underline{\mu}_{2},\bar{\mu}_{2}-\epsilon])=0$
when $\gamma\le0$. And by a symmetric argument, $\lim_{t\to\infty}\tilde{M}_{t}([\mu_{1}^{\bullet}+\epsilon,\bar{\mu}_{1}]\times[\underline{\mu}_{2},\bar{\mu}_{2}-\epsilon])=0$
when $\gamma>0$ and $\lim_{t\to\infty}\tilde{M}_{t}([\mu_{1}^{\bullet}+\epsilon,\bar{\mu}_{1}]\times[\underline{\mu}_{2}+\epsilon,\bar{\mu}_{2}])=0$
when $\gamma\le0.$ Taking the complement of these sets that get assigned
probability 0 in the limit establishes the result.
\end{proof}

\subsubsection{Decomposing Partial Derivative of Log-Likelihood With Respect to
$\mu_{2}$}

I record a decomposition of $\frac{\partial\ell}{\partial\mu_{2}}(\mu_{1},\mu_{2})$,
the partial derivative of the log-likelihood process with respect
to its second argument.

Define two stochastic processes: 
\begin{align*}
\varphi_{s}(\mu_{1},\mu_{2}) & :=-\lambda(X_{2,s}-\mu_{2}+\mu_{2}^{\bullet}+\gamma(X_{1,s}-\mu_{1}))\cdot1\{X_{1,s}\le\tilde{C}_{s}\}\\
\bar{\varphi}_{s}(\mu_{1},\mu_{2}) & :=\frac{\partial}{\partial\mu_{2}}\bar{L}(\mu_{2}+\gamma(\mu_{1}-\mu_{1}^{\bullet})\mid\tilde{C}_{s}),
\end{align*}
where $\bar{L}(\mu_{2}\mid c):=\int_{-\infty}^{c}g_{1}(x_{1})\cdot\int_{-\infty}^{\infty}g_{2}(x_{2}\mid x_{1})\cdot\ln(\phi(x_{2}\mid\mu_{2}-\gamma(x_{1}-\mu_{1}^{\bullet})))dx_{2}dx_{1}$.
Note that $\bar{\varphi}_{s}(\mu_{1},\mu_{2})$ is measurable with
respect to $\mathcal{F}_{s-1},$ since $(\tilde{C}_{t})$ is a predictable
process. Write $\xi_{s}(\mu_{1},\mu_{2}):=\varphi_{s}(\mu_{1},\mu_{2})-\bar{\varphi}_{s}(\mu_{1},\mu_{2})$
and $y_{t}(\mu_{1},\mu_{2}):=\sum_{s=1}^{t}\xi_{s}(\mu_{1},\mu_{2})$.
Write $z_{t}(\mu_{1},\mu_{2}):=\sum_{s=1}^{t}\bar{\varphi}_{s}(\mu_{1},\mu_{2})$.
\begin{lem}
\label{lem:decomposition-general} $\frac{\partial\ell_{t}}{\partial\mu_{2}}(\mu_{1},\mu_{2})=\frac{D_{2}m_{0}(\mu_{1},\mu_{2})}{m_{0}(\mu_{1},\mu_{2})}+y_{t}(\mu_{1},\mu_{2})+z_{t}(\mu_{1},\mu_{2})$
\end{lem}
\begin{proof}
This comes from expanding $\ell_{t}(\mu_{1},\mu_{2})$ and taking
its derivative as in the proof of Lemma \ref{lem:learning_mu1}.
\end{proof}
Now I derive a result about the $\xi_{t}(\mu_{1},\mu_{2})$ processes
for different pairs $(\mu_{1},\mu_{2}).$
\begin{lem}
\label{lem:quadratic_variation-general} There exists $\kappa_{\xi}<\infty$
so that for every $(\mu_{1},\mu_{2})\in\APLbox$ and for every $t\ge1,$
$\omega\in\Omega$, $\mathbb{E}[\xi_{t}^{2}(\mu_{1},\mu_{2})|\mathcal{F}_{t-1}](\omega)\le\kappa_{\xi}$.
\end{lem}
\begin{proof}
Note that $\bar{\varphi}_{t}(\mu_{1},\mu_{2})$ is measurable with
respect to $\mathcal{F}_{t-1}.$ Also, $\varphi_{t}(\mu_{1},\mu_{2})|\mathcal{F}_{t-1}=\varphi_{t}(\mu_{1},\mu_{2})|\tilde{C}_{t}$,
because by independence of $X_{t}$ from $(X_{s})_{s=1}^{t-1},$ the
only information that $\mathcal{F}_{t-1}$ contains about $\varphi_{t}(\mu_{1},\mu_{2})$
is in determining the cutoff threshold $\tilde{C}_{t}$.

At a sample path $\omega$ so that $\tilde{C}_{t}(\omega)=c\in\mathbb{R},$
\begin{align*}
 & \mathbb{E}[\varphi_{s}(\mu_{1},\mu_{2})|\mathcal{F}_{t-1}](\omega)\\
= & \mathbb{E}[-\lambda(X_{2,s}-\mu_{2}+\mu_{2}^{\bullet}+\gamma(X_{1,s}-\mu_{1}))\cdot\boldsymbol{1}\{X_{1}\le c\}]\\
= & \frac{\partial}{\partial\mu_{2}}\int_{-\infty}^{c}g_{1}(x_{1})\cdot\int_{-\infty}^{\infty}g_{2}(x_{2}\mid x_{1})\cdot\ln(\phi(x_{2}\mid\mu_{2}-\gamma(x_{1}-\mu_{1})))dx_{2}dx_{1}\\
= & \frac{\partial}{\partial\mu_{2}}\int_{-\infty}^{c}g_{1}(x_{1})\cdot\int_{-\infty}^{\infty}g_{2}(x_{2}\mid x_{1})\cdot\ln(\phi(x_{2}\mid[\mu_{2}+\gamma(\mu_{1}-\mu_{1}^{\bullet})]-\gamma(x_{1}-\mu_{1}^{\bullet})))dx_{2}dx_{1}\\
= & \frac{\partial}{\partial\mu_{2}}\bar{L}(\mu_{2}+\gamma(\mu_{1}-\mu_{1}^{\bullet})\mid c).
\end{align*}
 This shows that $\mathbb{E}[\varphi_{s}(\mu_{1},\mu_{2})|\mathcal{F}_{t-1}](\omega)=\bar{\varphi}_{s}(\mu_{1},\mu_{2})(\omega)$.
Since this holds regardless of $c$, we get that $\mathbb{E}[\varphi_{s}(\mu_{1},\mu_{2})|\mathcal{F}_{t-1}]=\bar{\varphi}_{t}(\mu_{1},\mu_{2})$
for all $\omega,$ that is to say {\small{}
\[
\mathbb{E}[\xi_{t}^{2}(\mu_{1},\mu_{2})|\mathcal{F}_{t-1}]=\text{Var}[\varphi_{t}(\mu_{1},\mu_{2})|\mathcal{F}_{t-1}]\le\mathbb{E}[\varphi_{t}^{2}(\mu_{1},\mu_{2})|\mathcal{F}_{t-1}]\le\mathbb{E}[(\lambda(X_{2,s}-\mu_{2}+\mu_{2}^{\bullet}+\gamma(X_{1,s}-\mu_{1})))^{2}].
\]
} It suffices to show $\mathbb{E}\left[\left(\lambda(X_{2}-\mu_{2}+\mu_{2}^{\bullet}+\gamma(X_{1}-\mu_{1}))\right)^{2}\right]$
exists for all $\mu_{1},\mu_{2}\in\mathbb{R}$ and is continuous.
The (finite) maximum value this expectation takes on the compact set
$\APLbox$ can be taken as $\kappa_{\xi}$.

Since the second derivative of the log of the normal density is uniformly
bounded, there exists some $\kappa_{f_{2}}<\infty$ so that for all
$z\in\mathbb{R},$ $-\kappa_{f_{2}}<\lambda^{'}(z)<0$. So, $\lambda(z)$
is Lipschitz continuous with constant $\kappa_{f_{2}}$. Let $b_{0}:=\lambda(-\mu_{2}+\mu_{2}^{\bullet}-\gamma\mu_{1})$.

For any $x_{1},x_{2}\in\mathbb{R},$ 
\begin{align*}
\left(\lambda(x_{2}-\mu_{2}+\mu_{2}^{\bullet}+\gamma(x_{1}-\mu_{1}))\right)^{2}= & b_{0}^{2}+\left(\lambda(x_{2}-\mu_{2}+\mu_{2}^{\bullet}+\gamma(x_{1}-\mu_{1}))\right)^{2}-\left(\lambda(-\mu_{2}+\mu_{2}^{\bullet}-\gamma\mu_{1})\right)^{2}\\
\le & b_{0}^{2}+\left|\lambda(x_{2}-\mu_{2}+\mu_{2}^{\bullet}+\gamma(x_{1}-\mu_{1}))-\lambda(-\mu_{2}+\mu_{2}^{\bullet}-\gamma\mu_{1})\right|\cdot\\
 & \times\left|\lambda(x_{2}-\mu_{2}+\gamma(x_{1}+\mu_{2}^{\bullet}-\mu_{1}))+\lambda(-\mu_{2}+\mu_{2}^{\bullet}-\gamma\mu_{1})\right|\\
\le & b_{0}^{2}+(\kappa_{f_{2}}\cdot(|x_{2}|+\gamma|x_{1}|))\cdot(2b_{0}+(\kappa_{f_{2}}\cdot(|x_{2}|+\gamma|x_{1}|))).
\end{align*}
 Note the bound is a second-order polynomial in $|x_{1}|$ and $|x_{2}|$.
We have{\small{}
\[
\mathbb{E}\left[\left(\lambda(X_{2}-\mu_{2}+\mu_{2}^{\bullet}+\gamma(X_{1}-\mu_{1}))\right)^{2}\right]\le\mathbb{E}\left[b_{0}^{2}+(\kappa_{f_{2}}\cdot(|X_{2}|+\gamma|X_{1}|))\cdot(2b_{0}+(\kappa_{f_{2}}\cdot(|X_{2}|+\gamma|X_{1}|)))\right]<\infty,
\]
} where the last inequality is due to the fact that $X_{1},X_{2}$
have finite second moments.
\end{proof}

\subsubsection{A Law of Large Numbers for Martingale Increments}

I use a statistical result from \citet*{heidhues2018unrealistic}
to show that the $y_{t}/t$ term in the decomposition of $\frac{1}{t}\frac{\partial\ell_{t}}{\partial\mu_{2}}$
almost surely converges to 0 in the long run, and furthermore this
convergence is uniform on $\APLbox.$ This lets me focus on terms
of the form $\bar{\varphi}_{s}(\mu_{1},\mu_{2})$, which can be interpreted
as the \emph{expected} contribution to the log likelihood derivative
from round $s$ data. This lends tractability to the problem as $\bar{\varphi}_{s}(\mu_{1},\mu_{2})$
only depends on $\tilde{C}_{s},$ but not on $X_{1,s}$ or $X_{2,s}$.
\begin{lem}
\label{lem:LLN_martingale_general} For every $(\mu_{1},\mu_{2})\in\APLbox$,
$\lim_{t\to\infty}|\frac{y_{t}(\mu_{1},\mu_{2})}{t}|=0$ almost surely.
\end{lem}
\begin{proof}
\citet*{heidhues2018unrealistic}'s Proposition 10 shows that if $(y_{t})$
is a martingale such that there exists some constant $v\ge0$ satisfying
$[y]_{t}\le vt$ almost surely, where $[y]_{t}$ is the quadratic
variation of $(y_{t}),$ then almost surely $\lim_{t\to\infty}\frac{y_{t}}{t}=0$.

Consider the process $y_{t}(\mu_{1},\mu_{2})$ for a fixed $(\mu_{1},\mu_{2})\in\APLbox$.
By definition $y_{t}=\sum_{s=1}^{t}\varphi_{s}(\mu_{1},\mu_{2})-\bar{\varphi}_{s}(\mu_{1},\mu_{2})$.
As established in the proof of Lemma \ref{lem:quadratic_variation-general},
for every $s,$ $\bar{\varphi}_{s}(\mu_{1},\mu_{2})=\mathbb{E}[\varphi_{s}(\mu_{1},\mu_{2})|\mathcal{F}_{s-1}]$.
So for $t^{'}<t,$ 
\begin{align*}
\mathbb{E}[y_{t}(\mu_{1},\mu_{2})|\mathcal{F}_{t^{'}}] & =\sum_{s=1}^{t^{'}}\varphi_{s}(\mu_{1},\mu_{2})-\bar{\varphi}_{s}(\mu_{1},\mu_{2})+\mathbb{E}\left[\sum_{s=t^{'}+1}^{t}\varphi_{s}(\mu_{1},\mu_{2})-\bar{\varphi}_{s}(\mu_{1},\mu_{2})|\mathcal{F}_{t^{'}}\right]\\
 & =\sum_{s=1}^{t^{'}}\varphi_{s}(\mu_{1},\mu_{2})-\bar{\varphi}_{s}(\mu_{1},\mu_{2})+\sum_{s=t^{'}+1}^{t}\mathbb{E}[\mathbb{E}[\varphi_{s}(\mu_{1},\mu_{2})-\bar{\varphi}_{s}(\mu_{1},\mu_{2})|\mathcal{F}_{s-1}]\mid\mathcal{F}_{t^{'}}]\\
 & =\sum_{s=1}^{t^{'}}\varphi_{s}(\mu_{1},\mu_{2})-\bar{\varphi}_{s}(\mu_{1},\mu_{2})+0=y_{t^{'}}(\mu_{1},\mu_{2}).
\end{align*}
This shows $(y_{t}(\mu_{1},\mu_{2}))_{t}$ is a martingale. Also,
\[
[y(\mu_{1},\mu_{2})]_{t}=\sum_{s=1}^{t-1}\mathbb{E}[(y_{s}(\mu_{1},\mu_{2})-y_{s-1}(\mu_{1},\mu_{2}))^{2}|\mathcal{F}_{s-1}]=\sum_{s=1}^{t-1}\mathbb{E}[\xi_{s}^{2}(\mu_{1},\mu_{2})|\mathcal{F}_{s-1}]\le\kappa_{\xi}\cdot t
\]
by Lemma \ref{lem:quadratic_variation-general}. Therefore \citet*{heidhues2018unrealistic}
Proposition 10 applies.
\end{proof}
\begin{lem}
\label{lem:uniform_LLN_genereal} $\lim_{t\to\infty}\sup_{(\mu_{1},\mu_{2})\in\APLbox}|\frac{y_{t}(\mu_{1},\mu_{2})}{t}|=0$
almost surely.
\end{lem}
\begin{proof}
This argument is similar to Lemma 11 in \citet*{heidhues2018unrealistic}.
I apply Lemma 2 of \citet{andrews1992generic}, which says to prove
this result I just need to check conditions BD, P-SSLN, and S-LIP
from \citet{andrews1992generic}. BD holds because $\APLbox$ is a
bounded subset of $\mathbb{R}^{2}.$ P-SLLN holds because by Lemma
\ref{lem:LLN_martingale_general}, which shows for all $(\mu_{1},\mu_{2})\in\APLbox$,
$\lim_{t\to\infty}|\frac{y_{t}(\mu_{1},\mu_{2})}{t}|=0$ almost surely.

Condition S-LIP is essentially a Lipschitz continuity condition. It
requires finding sequence of random variables $B_{t}$ such that $|\xi_{t}(\mu_{1},\mu_{2})-\xi_{t}(\mu_{1}^{'},\mu_{2}^{'})|\le B_{t}\cdot(|\mu_{1}-\mu_{1}^{'}|+|\mu_{2}-\mu_{2}^{'}|)$
almost surely, such that these random variables satisfy $\sup_{t\ge1}\frac{1}{t}\sum_{s=1}^{t}\mathbb{E}[B_{s}]<\infty$,
and $\lim_{t\to\infty}\frac{1}{t}\sum_{s=1}^{t}(B_{s}-\mathbb{E}[B_{s}])=0$
almost surely.

But for every $\omega,$ $\varphi_{s}(\mu_{1},\mu_{2}):=-\lambda(X_{2,s}-\mu_{2}+\mu_{2}^{\bullet}+\gamma(X_{1,s}-\mu_{1}))\cdot1\{X_{1,s}\le\tilde{C}_{s}\}$
\begin{align*}
|\varphi_{s}(\mu_{1},\mu_{2})-\varphi_{s}(\mu_{1}^{'},\mu_{2}^{'})|\le & |\lambda(X_{2,s}-\mu_{2}+\mu_{2}^{\bullet}+\gamma(X_{1,s}-\mu_{1}))-\lambda(X_{2,s}-\mu_{2}^{'}+\mu_{2}^{\bullet}+\gamma(X_{1,s}-\mu_{1}^{'}))|.
\end{align*}
 As $\ln(f_{2}(\cdot))$ has a bounded second derivative,  RHS is
bounded by $\kappa_{f_{2}}\cdot\left(|\mu_{2}-\mu_{2}^{'}|+\gamma\cdot|\mu_{1}-\mu_{1}^{'}|\right)$.

Now that we know $|\varphi_{s}(\mu_{1},\mu_{2})-\varphi_{s}(\mu_{1}^{'},\mu_{2}^{'})|(\omega)\le\kappa_{f_{2}}\cdot\left(|\mu_{2}-\mu_{2}^{'}|+\gamma\cdot|\mu_{1}-\mu_{1}^{'}|\right)$
for all $\omega,$ we must also have $|\bar{\varphi}_{s}(\mu_{1},\mu_{2})-\bar{\varphi}_{s}(\mu_{1}^{'},\mu_{2}^{'})|(\omega)\le\kappa_{f_{2}}\cdot\left(|\mu_{2}-\mu_{2}^{'}|+\gamma\cdot|\mu_{1}-\mu_{1}^{'}|\right)$
for all $\omega$ since $\bar{\varphi}_{s}(\mu_{1},\mu_{2})=\mathbb{E}[\varphi_{s}(\mu_{1},\mu_{2})\mid\mathcal{F}_{s-1}]$.

Setting $B_{s}$ as the constant $2\kappa_{f_{2}}$ for every $s$
satisfies S-LIP.
\end{proof}

\subsubsection{Bounds on Asymptotic Beliefs and Asymptotic Cutoffs}

Recall that Lemma \ref{lem:C_equivalence} implies that for any $\mu_{2},$
all pairs of fundamentals on the line $li(\mu_{2})$ have the same
optimal cutoff threshold. Then against any feasible model $\Psi(\mu_{1},\mu_{2};\gamma)$
with $(\mu_{1},\mu_{2})\in\APLbox$, the best cutoff strategy is between
$C(\mu_{1}^{\bullet},\underline{\mu}_{2}^{\circ};\gamma)$ and $C(\mu_{1}^{\bullet},\bar{\mu}_{2}^{\circ};\gamma)$.
Define these cutoffs as $\underline{c}^{\circ}$ and $\bar{c}^{\circ}$
respectively.
\begin{lem}
\label{lem:belief_bound-general} Let $\underline{c}^{\circ}\le c\le\bar{c}^{\circ}$.
If $r-\gamma<0$, then $\liminf_{t\to\infty}\tilde{C}_{t}\ge c$ almost
surely implies $\lim_{t\to\infty}\tilde{M}_{t}(\ \lozenge[\underline{\mu}_{2}^{\circ},\mu_{2}^{*}(c))\ )=0$
almost surely and $\limsup_{t\to\infty}\tilde{C}_{t}\le c$ almost
surely implies $\lim_{t\to\infty}\tilde{M}_{t}(\ \lozenge(\mu_{2}^{*}(c),\bar{\mu}_{2}^{\circ}]\ )=0$
almost surely. If $r-\gamma>0$, then $\liminf_{t\to\infty}\tilde{C}_{t}\ge c$
almost surely implies $\lim_{t\to\infty}\tilde{M}_{t}(\ \lozenge(\mu_{2}^{*}(c),\bar{\mu}_{2}^{\circ}]\ )=0$
almost surely and $\limsup_{t\to\infty}\tilde{C}_{t}\le c$ almost
surely implies $\lim_{t\to\infty}\tilde{M}_{t}(\ \lozenge[\underline{\mu}_{2}^{\circ},\mu_{2}^{*}(c))\ )=0$
almost surely.
\end{lem}
\begin{proof}
We prove the ``liminf'' statement for the case of $r-\gamma<0$
and briefly discuss the argument for the ``limsup'' statement for
the case of $r-\gamma>0$ — the arguments for the other two statements
are very similar.

Consider the first statement when $r-\gamma<0$, fixing some $\underline{c}$
with $\underline{c}^{\circ}\le\underline{c}\le\bar{c}^{\circ}$. We
show that for all $\epsilon>0,$ there exists $\delta>0$ such that
almost surely,
\[
\liminf_{t\to\infty}\inf_{(\mu_{1},\mu_{2})\in\APLbox\cap\lozenge[\underline{\mu}_{2}^{\circ},\mu_{2}^{*}(\underline{c})-\epsilon]}\frac{1}{t}\frac{\partial\ell_{t}}{\partial\mu_{2}}(\mu_{1},\mu_{2})\ge\delta.
\]
From Lemma \ref{lem:decomposition-general}, we may rewrite LHS as
\begin{align*}
\liminf_{t\to\infty}\inf_{(\mu_{1},\mu_{2})\in\APLbox\cap\lozenge[\underline{\mu}_{2}^{\circ},\mu_{2}^{*}(\underline{c})-\epsilon]}\left[\frac{1}{t}\frac{D_{2}m_{0}(\mu_{1},\mu_{2})}{m_{0}(\mu_{1},\mu_{2})}+\frac{y_{t}(\mu_{1},\mu_{2})}{t}+\frac{z_{t}(\mu_{1},\mu_{2})}{t}\right],
\end{align*}
which is no smaller than taking the inf separately across the three
terms in the bracket, 
\begin{align*}
 & \liminf_{t\to\infty}\inf_{(\mu_{1},\mu_{2})\in\APLbox\cap\lozenge[\underline{\mu}_{2}^{\circ},\mu_{2}^{*}(\underline{c})-\epsilon]}\frac{1}{t}\frac{D_{2}m_{0}(\mu_{1},\mu_{2})}{m_{0}(\mu_{1},\mu_{2})}+\liminf_{t\to\infty}\inf_{(\mu_{1},\mu_{2})\in\lozenge[\underline{\mu}_{2}^{\circ},\mu_{2}^{*}(\underline{c})-\epsilon]}\frac{y_{t}(\mu_{1},\mu_{2})}{t}\\
 & +\liminf_{t\to\infty}\inf_{(\mu_{1},\mu_{2})\in\APLbox\cap\lozenge[\underline{\mu}_{2}^{\circ},\mu_{2}^{*}(\underline{c})-\epsilon]}\frac{z_{t}(\mu_{1},\mu_{2})}{t}.
\end{align*}

Since $D_{2}m_{0}/m_{0}$ is bounded on $\APLbox$ as $D_{2}m_{0}$
is continuous and $m_{0}$ is continuous and strictly positive on
the compact set $\APLbox$, the first term is 0 for every $\omega$.
To deal with the second term, 
\begin{align*}
\liminf_{t\to\infty}\inf_{(\mu_{1},\mu_{2})\in\APLbox\cap\lozenge[\underline{\mu}_{2}^{\circ},\mu_{2}^{*}(\underline{c})-\epsilon]}\frac{y_{t}(\mu_{1},\mu_{2})}{t} & \ge\liminf_{t\to\infty}\inf_{(\mu_{1},\mu_{2})\in\APLbox}-|\frac{y_{t}(\mu_{1},\mu_{2})}{t}|\\
 & =\liminf_{t\to\infty}\left\{ -1\cdot\sup_{(\mu_{1},\mu_{2})\in\APLbox}|\frac{y_{t}(\mu_{1},\mu_{2})}{t}|\right\} .
\end{align*}
Lemma \ref{lem:uniform_LLN_genereal} gives $\lim_{t\to\infty}\sup_{(\mu_{1},\mu_{2})\in\APLbox}|\frac{y_{t}(\mu_{1},\mu_{2})}{t}|=0$
almost surely. Hence, we conclude that, almost surely, 
\[
\liminf_{t\to\infty}\inf_{(\mu_{1},\mu_{2})\in\APLbox\cap\lozenge[\underline{\mu}_{2}^{\circ},\mu_{2}^{*}(\underline{c})-\epsilon]}\frac{y_{t}(\mu_{1},\mu_{2})}{t}\ge0.
\]

It suffices then to find $\delta>0$ and show $\liminf_{t\to\infty}\inf_{(\mu_{1},\mu_{2})\in\APLbox\cap\lozenge[\underline{\mu}_{2}^{\circ},\mu_{2}^{*}(\underline{c})-\epsilon]}\frac{z_{t}(\mu_{1},\mu_{2})}{t}\ge\delta$
almost surely. To do this, I first show $\bar{\varphi}_{s}(\mu_{1},\mu_{2})(\omega)\ge\delta$
whenever $\tilde{C}_{s}(\omega)\ge\underline{c}$ and $\mu_{2}\le\mu_{2}^{*}(\underline{c})-\epsilon$.
At every $\underline{c}^{\circ}\le c'\le\bar{c}^{\circ}$, we get
\[
\frac{\partial}{\partial\mu_{2}}\bar{L}(\mu_{2}\mid c')=\int_{-\infty}^{c'}g_{1}(x_{1})\cdot\int_{-\infty}^{\infty}(-1)\cdot g_{2}(x_{2}\mid x_{1})\cdot\lambda(x_{2}-\mu_{2}+\mu_{2}^{\bullet}+\gamma(x_{1}-\mu_{1}^{\bullet}))dx_{2}dx_{1}.
\]
First-order condition implies that $\frac{\partial}{\partial\mu_{2}}\bar{L}(\mu_{2}^{*}(c')\mid c')=0$.
Since $\lambda$ is strictly decreasing, we also get $\frac{\partial}{\partial\mu_{2}}\bar{L}(\mu_{2}\mid c')>0$
for any $\mu_{2}<\mu_{2}^{*}(c').$ Since we have $r-\gamma<0,$ $\mu_{2}^{*}(\cdot)$
is strictly increasing, which means $\frac{\partial}{\partial\mu_{2}}\bar{L}(\mu_{2}^{*}(\underline{c})-\epsilon\mid c')>0$
for any $\underline{c}\le c'\le\bar{c}^{\circ}.$ Let $\delta>0$
satisfy $\min_{c'\in[\underline{c},\bar{c}^{\circ}]}\frac{\partial}{\partial\mu_{2}}\bar{L}(\mu_{2}^{*}(\underline{c})-\epsilon\mid c')>\delta$,
which exists because $c'\mapsto\frac{\partial}{\partial\mu_{2}}\bar{L}(\mu_{2}^{*}(\underline{c})-\epsilon\mid c')$
is continuous on the compact domain $[\underline{c},\bar{c}^{\circ}].$
When $\tilde{C}_{s}(\omega)=c'\in[\underline{c},\bar{c}^{\circ}]$
and for any $(\mu_{1},\mu_{2})\in\APLbox\cap\lozenge[\underline{\mu}_{2}^{\circ},\mu_{2}^{*}(\underline{c})-\epsilon]$,
we have $\bar{\varphi}_{s}(\mu_{1},\mu_{2})(\omega)=\frac{\partial}{\partial\mu_{2}}\bar{L}(\mu_{2}\mid c')\ge\frac{\partial}{\partial\mu_{2}}\bar{L}(\mu_{2}^{*}(\underline{c})-\epsilon\mid c')>\delta.$

Along any $\omega$ where $\liminf_{t\to\infty}\tilde{C}_{t}\ge\underline{c}$,
we therefore have 
\[
\liminf_{s\to\infty}\inf_{(\mu_{1},\mu_{2})\in\APLbox\cap\lozenge[\underline{\mu}_{2}^{\circ},\mu_{2}^{*}(\underline{c})-\epsilon]}\bar{\varphi}_{s}(\mu_{1},\mu_{2})\ge\delta
\]
 and thus 
\[
\liminf_{t\to\infty}\inf_{(\mu_{1},\mu_{2})\in\APLbox\cap\lozenge[\underline{\mu}_{2}^{\circ},\mu_{2}^{*}(\underline{c})-\epsilon]}\frac{z_{t}(\mu_{1},\mu_{2})}{t}=\liminf_{t\to\infty}\inf_{(\mu_{1},\mu_{2})\in\APLbox\cap\lozenge[\underline{\mu}_{2}^{\circ},\mu_{2}^{*}(\underline{c})-\epsilon]}\frac{1}{t}\left[\sum_{s=1}^{t}\bar{\varphi}_{s}(\mu_{1},\mu_{2})\right]\ge\delta.
\]
Let $R:=[\underline{\mu}_{1},\bar{\mu}_{1}]\times[\underline{\mu}_{2},\bar{\mu}_{2}-\epsilon]\cap\lozenge[\underline{\mu}_{2}^{\circ},\mu_{2}^{*}(\underline{c})-2\epsilon]$,
and let $R^{'}:=R+(0,\epsilon)'$ be $R$ shifted upwards by $\epsilon$.
We have both $R,R^{'}\subseteq\APLbox\cap\lozenge[\underline{\mu}_{2}^{\circ},\mu_{2}^{*}(\underline{c})-\epsilon]$.
The illustration is for the case of $\gamma>0$.

\begin{center}\includegraphics[scale=0.17]{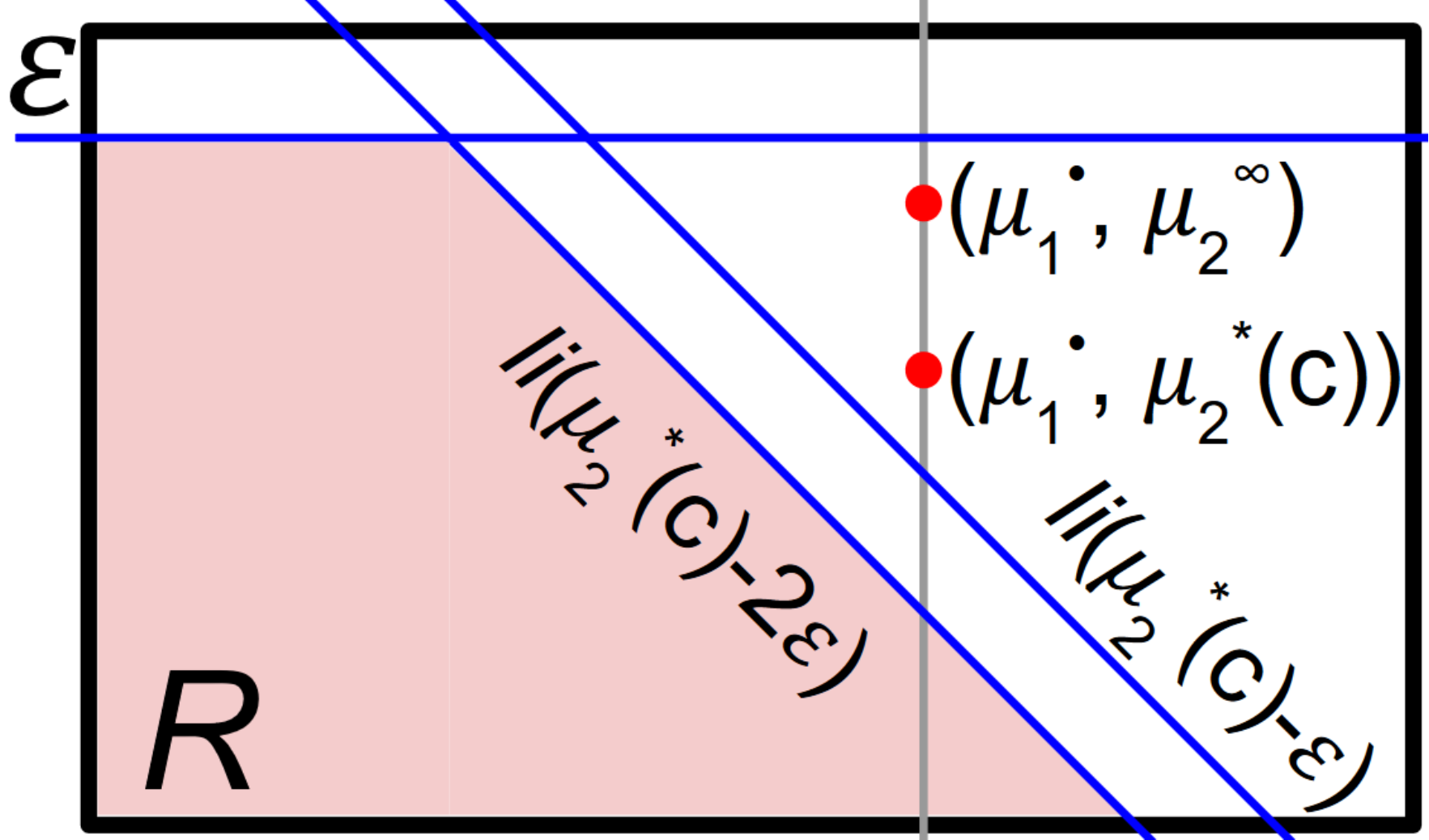}\end{center}

So using the same argument as in the proof of Lemma \ref{lem:learning_mu1},
\begin{align*}
\tilde{M}_{t}(R) & =\int_{R^{'}}\tilde{m}_{t}(\mu_{1},\mu_{2})\cdot\frac{\tilde{m}_{t}(\mu_{1},\mu_{2}-\epsilon)}{\tilde{m}_{t}(\mu_{1},\mu_{2})}d\mu\\
 & =\int_{R^{'}}\tilde{m}_{t}(\mu_{1},\mu_{2})\exp(\ell_{t}(\mu_{1},\mu_{2}-\epsilon)-\ell_{t}(\mu_{1},\mu_{2}))d\mu\\
 & =\int_{R^{'}}\tilde{m}_{t}(\mu_{1},\mu_{2})\exp(-\int_{0}^{\epsilon}\frac{\partial\ell_{t}}{\partial\mu_{2}}(\mu_{1},\mu_{2}-\epsilon+z)dz)d\mu
\end{align*}
Almost surely, 
\[
\liminf_{t\to\infty}\inf_{(\mu_{1},\mu_{2})\in R^{'},z\in[0,\epsilon]}\left(\frac{\partial\ell_{t}}{\partial\mu_{2}}(\mu_{1},\mu_{2}-\epsilon+z)\right)\ge t\delta,
\]
so almost surely 
\[
\limsup_{t\to\infty}\tilde{M}_{t}(R)\le\limsup_{t\to\infty}\exp(-t\epsilon\delta)\cdot\int_{R^{'}}\tilde{m}_{t}(\mu_{1},\mu_{2})d\mu=0.
\]
Letting $\epsilon\to0$ and noting that $li(\mu_{2}^{*}(\underline{c}))$
crosses the top edge of $\APLbox$ to the left of $\mu_{1}^{\bullet}$
when $\gamma>0$, we get $\lim_{t\to\infty}\tilde{M}_{t}(\ \lozenge[\mu_{2}^{*}(\underline{c}),\bar{\mu}_{2}^{\circ}]\cup[\underline{\mu}_{1},\mu_{1}^{\bullet}]\times\{\bar{\mu}_{2}\}\ )=1$
almost surely. But from Lemma \ref{lem:learning_mu1}, the set $[\underline{\mu}_{1},\mu_{1}^{\bullet}]\times\{\bar{\mu}_{2}\}$
must receive no weight in the limit, hence $\lim_{t\to\infty}\tilde{M}_{t}(\ \lozenge[\underline{\mu}_{2}^{\circ},\mu_{2}^{*}(\underline{c}))\ )=0$
almost surely as desired. (The case of $\gamma<0$ is analogous.)

Now consider any $\underline{c}^{\circ}\le\bar{c}\le\bar{c}^{\circ}.$
I briefly discuss why $\limsup_{t\to\infty}\tilde{C}_{t}\le\bar{c}$
almost surely implies $\lim_{t\to\infty}\tilde{M}_{t}(\ \lozenge[\underline{\mu}_{2}^{\circ},\mu_{2}^{*}(\bar{c}))\ )=0$
almost surely when $r-\gamma>0$. As in the argument before, the key
is to find some $\delta>0$ such that $\frac{\partial}{\partial\mu_{2}}\bar{L}(\mu_{2}\mid c')>\delta$
whenever $c'\in[\underline{c}^{\circ},\bar{c}]$ and $\mu_{2}\le\mu_{2}^{*}(\bar{c}).$
For each $c\in[\underline{c}^{\circ},\bar{c}^{\circ}],$ FOC implies
$\frac{\partial}{\partial\mu_{2}}\bar{L}(\mu_{2}^{*}(c')\mid c')=0.$
Since $\lambda$ is strictly decreasing, we also get $\frac{\partial}{\partial\mu_{2}}\bar{L}(\mu_{2}\mid c')>0$
for any $\mu_{2}<\mu_{2}^{*}(c').$ Since we now consider $r-\gamma>0$,
$\mu_{2}^{*}(c)$ is strictly decreasing in $c$, and this shows $\frac{\partial}{\partial\mu_{2}}\bar{L}(\mu_{2}^{*}(\bar{c})-\epsilon\mid c')>0$
for any $\underline{c}^{\circ}\le c'\le\bar{c}$. We can find $\delta>0$
such that $\frac{\partial}{\partial\mu_{2}}\bar{L}(\mu_{2}^{*}(\bar{c})-\epsilon\mid c')>\delta$
for every $\underline{c}^{\circ}\le c'\le\bar{c}$ by continuity,
so we also get $\frac{\partial}{\partial\mu_{2}}\bar{L}(\mu_{2}\mid c')>\delta$
for any $\mu_{2}\le\mu_{2}^{*}(\bar{c})-\epsilon$.
\end{proof}
Now, I use a bound on agents' asymptotic beliefs about $\mu_{2}$
to deduce asymptotic restrictions on their cutoffs.
\begin{lem}
\label{lem:cutoff_bound-general} Suppose that there are $\underline{\mu}_{2}^{\circ}\le\mu_{2}^{l}<\mu_{2}^{h}\le\bar{\mu}_{2}^{\circ}$
such that $\lim_{t\to\infty}\tilde{M}_{t}(\lozenge[\mu_{2}^{l},\mu_{2}^{h}])=1$
almost surely. Then $\liminf_{t\to\infty}\tilde{C}_{t}\ge C(\mu_{1}^{\bullet},\mu_{2}^{l};\gamma)$
and $\limsup_{t\to\infty}\tilde{C}_{t}\le C(\mu_{1}^{\bullet},\mu_{2}^{h};\gamma)$
almost surely.
\end{lem}
\begin{proof}
I show $\liminf_{t\to\infty}\tilde{C}_{t}\ge C(\mu_{1}^{\bullet},\mu_{2}^{l};\gamma)$
almost surely. The argument establishing $\limsup_{t\to\infty}\tilde{C}_{t}\le C(\mu_{1}^{\bullet},\mu_{2}^{h};\gamma)$
is symmetric.

Let $c^{l}=C(\mu_{1}^{\bullet},\mu_{2}^{l};\gamma)$, and recall before
we defined $\underline{c}^{\circ}:=C(\mu_{1}^{\bullet},\underline{\mu}_{2}^{\circ};\gamma)$
and $\bar{c}^{\circ}:=C(\mu_{1}^{\bullet},\bar{\mu}_{2}^{\circ};\gamma)$.

Let $U(c;\mu_{1},\mu_{2})$ be the expected payoff of using the stopping
strategy $S_{c}$ when $(X_{1},X_{2})\sim\Psi(\mu_{1},\mu_{2};\gamma).$
\foreignlanguage{american}{I first show $c\mapsto U(c;\mu_{1},\mu_{2})$
is single peaked: it is strictly increasing up to $c=c^{*},$ the
subjectively optimal cutoff under $\Psi(\mu_{1},\mu_{2};\gamma)$,
then strictly decreasing afterwards. Recall (from the proof of }Lemma
\ref{lem:cutoff_properties}\foreignlanguage{american}{ when $\gamma\ge-1$)
the cutoff form of the best stopping strategy comes from the fact
that \\ $u_{1}(x_{1})<\mathbb{E}_{\Psi(\mu_{1},\mu_{2};\gamma)}[u_{2}(x_{1},X_{2})|X_{1}=x_{1}]$
for $x_{1}<c^{*},$ but $u_{1}(x_{1})<\mathbb{E}_{\Psi(\mu_{1},\mu_{2};\gamma)}[u_{2}(x_{1},X_{2})|X_{1}=x_{1}]$
for $x_{1}>c^{*}.$ For two cutoffs $c_{1}<c_{2}<c^{*},$ the two
stopping strategies $S_{c_{1}},S_{c_{2}}$ only differ in how they
treat first-period draws in the interval $[c_{1},c_{2}],$ so we can
write the difference in their expected payoffs as $\int_{c_{1}}^{c_{2}}\left(\mathbb{E}_{\Psi(\mu_{1},\mu_{2};\gamma)}[u_{2}(x_{1},X_{2})|X_{1}=x_{1}]-u_{1}(x_{1})\right)\phi(x_{1}\mid\mu_{1})dx_{1}.$
The integrand is strictly positive on $[c_{1},c_{2}],$ therefore
$U(c_{1};\mu_{1},\mu_{2})<U(c_{2};\mu_{1},\mu_{2}).$ This shows $U(\cdot;\mu_{1},\mu_{2})$
is strictly increasing up until $c^{*}$; a symmetric argument shows
it is strictly decreasing after $c^{*}$.}

By Lemma \ref{lem:C_equivalence}, $C(\mu_{1}^{'},\mu_{2}^{'};\gamma)=C(\mu_{1}^{\bullet},\mu_{2};\gamma)$
for all $(\mu_{1}^{'},\mu_{2}^{'})\in li(\mu_{2})$. Since $c\mapsto U(c;\mu_{1},\mu_{2})$
is single peaked for every $(\mu_{1},\mu_{2}),$ and since $c^{l}\le C(\mu_{1}^{\bullet},\mu_{2};\gamma)$
for all $\mu_{2}\in[\mu_{2}^{l},\mu_{2}^{h}],$ we also get $c^{l}\le C(\mu_{1}^{'},\mu_{2}^{'};\gamma)$
for every $(\mu_{1}^{'},\mu_{2}^{'})\in\lozenge[\mu_{2}^{l},\mu_{2}^{h}]$,
since $\lozenge[\mu_{2}^{l},\mu_{2}^{h}]$ is the union of the line
segments, $\lozenge[\mu_{2}^{l},\mu_{2}^{h}]=\cup_{\mu_{2}\in[\mu_{2}^{l},\mu_{2}^{h}]}li(\mu_{2})$.

Fix some $\epsilon>0.$ We get $U(c^{l};\mu_{1},\mu_{2})-U(c^{l}-\epsilon;\mu_{1},\mu_{2})>0$
for every $(\mu_{1},\mu_{2})\in\lozenge[\mu_{2}^{l},\mu_{2}^{h}]$.
As $(\mu_{1},\mu_{2})\mapsto\left(U(c^{l};\mu_{1},\mu_{2})-U(c^{l}-\epsilon;\mu_{1},\mu_{2})\right)$
is continuous, there exists some $\kappa^{*}>0$ so that $U(c^{l};\mu_{1},\mu_{2})-U(c^{l}-\epsilon;\mu_{1},\mu_{2})>\kappa^{*}$
for all $(\mu_{1},\mu_{2})\in\lozenge[\mu_{2}^{l},\mu_{2}^{h}]$.
In particular, if $\nu\in\Delta(\lozenge[\mu_{2}^{l},\mu_{2}^{h}])$
is a belief about fundamentals, then $\int U(c^{l};\mu_{1},\mu_{2})-U(c^{l}-\epsilon;\mu_{1},\mu_{2})d\nu(\mu)>\kappa^{*}.$

Now , let $\bar{\kappa}:=\sup_{c\in[\underline{c}^{\circ},\bar{c}^{\circ}]}\sup_{(\mu_{1},\mu_{2})\in\APLbox}U(c;\mu_{1},\mu_{2}),$
$\underline{\kappa}:=\inf_{c\in[\underline{c}^{\circ},\bar{c}^{\circ}]}\inf_{(\mu_{1},\mu_{2})\in\APLbox}U(c;\mu_{1},\mu_{2}).$
Find $p\in(0,1)$ so that $p\kappa^{*}-(1-p)(\bar{\kappa}-\underline{\kappa})=0.$
At any belief $\hat{\nu}\in\Delta(\APLbox)$ that assigns more than
probability $p$ to the parallelogram $\lozenge[\mu_{2}^{l},\mu_{2}^{h}]$,
the optimal cutoff is larger than $c^{l}-\epsilon$. To see this,
take any $\hat{c}\le c^{l}-\epsilon$ and I will show $\hat{c}$ is
suboptimal. If $\hat{c}<\underline{c},$ then it is suboptimal after
any belief on $\lozenge.$ If $\underline{c}\le\hat{c}\le c^{l}-\epsilon$,
I show that $\int U(c^{l};\mu_{1},\mu_{2})-U(\hat{c};\mu_{1},\mu_{2})d\hat{\nu}(\mu)>0.$
To see this, we may decompose $\hat{\nu}$ as the mixture of a probability
measure $\nu$ on $\lozenge[\mu_{2}^{l},\mu_{2}^{h}]$ and another
probability measure $\nu^{c}$ on $\APLbox\backslash\lozenge[\mu_{2}^{l},\mu_{2}^{h}].$
Let $\hat{p}>p$ be the probability that $\nu$ assigns to $\lozenge[\mu_{2}^{l},\mu_{2}^{h}].$
The above integral is equal to:
\begin{align*}
\hat{p}\int_{\lozenge[\mu_{2}^{l},\mu_{2}^{h}]}U(c^{l};\mu_{1},\mu_{2})-U(\hat{c};\mu_{1},\mu_{2})d\nu(\mu)+(1-\hat{p})\int_{\APLbox\backslash\lozenge[\mu_{2}^{l},\mu_{2}^{h}]}U(c^{l};\mu_{1},\mu_{2})-U(\hat{c};\mu_{1},\mu_{2})d\nu^{c}(\mu)
\end{align*}
Since $c^{l}$ is to the left of the optimal cutoff for all $(\mu_{1},\mu_{2})\in\lozenge[\mu_{2}^{l},\mu_{2}^{h}]$
and $\hat{c}\le c^{l}-\epsilon$, then $U(\hat{c};\mu_{1},\mu_{2})\le U(c^{l}-\epsilon;\mu_{1},\mu_{2})$
for all $(\mu_{1},\mu_{2})\in\lozenge[\mu_{2}^{l},\mu_{2}^{h}]$.
The first summand is no less than $\hat{p}\int_{\lozenge[\mu_{2}^{l},\mu_{2}^{h}]}U(c^{l};\mu_{1},\mu_{2})-U(c^{l}-\epsilon;\mu_{1},\mu_{2})d\nu(\mu)\ge\hat{p}\kappa^{*}.$
Also, the integrand in the second summand is no smaller than $-(\bar{\kappa}-\underline{\kappa}),$
therefore $\int U(c^{l};\mu_{1},\mu_{2})-U(\hat{c};\mu_{1},\mu_{2})d\hat{\nu}(\mu)\ge\hat{p}\kappa^{*}-(1-\hat{p})(\bar{\kappa}-\underline{\kappa}).$
Since $\hat{p}>p$, we get $\hat{p}\kappa^{*}-(1-\hat{p})(\bar{\kappa}-\underline{\kappa})>0$.

Along any sample path $\omega$ where $\lim_{t\to\infty}\tilde{M}_{t}(\lozenge[\mu_{2}^{l},\mu_{2}^{h}])(\omega)=1,$
eventually $\tilde{M}_{t}(\lozenge[\mu_{2}^{l},\mu_{2}^{h}])(\omega)>p$
for all large enough $t,$ meaning $\liminf_{t\to\infty}\tilde{C}_{t}(\omega)\ge c^{l}-\epsilon.$
Since $\lim_{t\to\infty}\tilde{M}_{t}(\lozenge[\mu_{2}^{l},\mu_{2}^{h}])=1$
almost surely, this shows $\liminf_{t\to\infty}\tilde{C}_{t}\ge C(\mu_{1}^{\bullet},\mu_{2}^{l};\gamma)-\epsilon$
almost surely. As the choice of $\epsilon>0$ was arbitrary, we conclude
$\liminf_{t\to\infty}\tilde{C}_{t}\ge C(\mu_{1}^{\bullet},\mu_{2}^{l};\gamma)$
almost surely.
\end{proof}

\subsubsection{The Contraction Map}

I now combine the results established so far to prove the convergence
statement in Proposition \ref{prop:one-by-one}.
\begin{proof}
Let $\mu_{2,[1]}^{A}:=\underline{\mu}_{2}^{\circ}$, $\mu_{2,[1]}^{B}:=\bar{\mu}_{2}^{\circ}$.
For $k=2,3,...$, iteratively define $\mu_{2,[k]}^{A}:=\mathcal{I}(\mu_{2,[k-1]}^{A};\gamma)$
and $\mu_{2,[k]}^{B}:=\mathcal{I}(\mu_{2,[k-1]}^{B};\gamma)$. Let
$\mu_{2,[k]}^{l}:=\min(\mu_{2,[k]}^{A},\mu_{2,[k]}^{B})$ and $\mu_{2,[k]}^{h}:=\max(\mu_{2,[k]}^{A},\mu_{2,[k]}^{B})$.
I show by induction that for every $k$, $\lim_{t\to\infty}\tilde{M}_{t}(\lozenge[\mu_{2,[k]}^{l},\mu_{2,[k]}^{h}])=1$
almost surely. (The base case of $k=1$ holds by the support of the
prior belief.)

\textbf{Inductive step when $r-\gamma<0$. }From Lemma \ref{lem:cutoff_bound-general},
if $\lim_{t\to\infty}\tilde{M}_{t}(\lozenge[\mu_{2,[k]}^{l},\mu_{2,[k]}^{h}])=1$
almost surely, then $\liminf_{t\to\infty}\tilde{C}_{t}\ge C(\mu_{1}^{\bullet},\mu_{2,[k]}^{l};\gamma)$
and $\limsup_{t\to\infty}\tilde{C}_{t}\le C(\mu_{1}^{\bullet},\mu_{2,[k]}^{h};\gamma)$
almost surely. Using these conclusions in Lemma \ref{lem:belief_bound-general},
we deduce that almost surely,
\[
\lim_{t\to\infty}\tilde{M}_{t}(\lozenge[\mu_{2}^{*}(C(\mu_{1}^{\bullet},\mu_{2,[k]}^{l};\gamma)),\mu_{2}^{*}(C(\mu_{1}^{\bullet},\mu_{2,[k]}^{h};\gamma))])=1.
\]

Both $C(\mu_{1}^{\bullet},\cdot;\gamma)$ and $\mu_{2}^{*}(\cdot)$
are strictly increasing, so $\lim_{t\to\infty}\tilde{M}_{t}(\lozenge[\mu_{2,[k+1]}^{l},\mu_{2,[k+1]}^{h}])=1$
almost surely.

\textbf{Inductive step when $r-\gamma>0$}. Now, $C(\mu_{1}^{\bullet},\cdot;\gamma)$
is strictly increasing but $\mu_{2}^{*}(\cdot)$ is strictly decreasing.
From Lemma \ref{lem:cutoff_bound-general}, if $\lim_{t\to\infty}\tilde{M}_{t}(\lozenge[\mu_{2,[k]}^{l},\mu_{2,[k]}^{h}])=1$
almost surely, then $\liminf_{t\to\infty}\tilde{C}_{t}\ge C(\mu_{1}^{\bullet},\mu_{2,[k]}^{l};\gamma)$
and $\limsup_{t\to\infty}\tilde{C}_{t}\le C(\mu_{1}^{\bullet},\mu_{2,[k]}^{h};\gamma)$
almost surely. But using these conclusions in Lemma \ref{lem:belief_bound-general},
for the case of $r-\gamma>0$, we further deduce that 
\[
\lim_{t\to\infty}\tilde{M}_{t}(\lozenge[\mu_{2}^{*}(C(\mu_{1}^{\bullet},\mu_{2,[k]}^{h};\gamma)),\mu_{2}^{*}(C(\mu_{1}^{\bullet},\mu_{2,[k]}^{l};\gamma))])=1.
\]
So now we have $\mu_{2,[k+1]}^{l}=\mu_{2}^{*}(C(\mu_{1}^{\bullet},\mu_{2,[k]}^{h};\gamma))$
and $\mu_{2,[k+1]}^{h}=\mu_{2}^{*}(C(\mu_{1}^{\bullet},\mu_{2,[k]}^{l};\gamma))$,
but still conclude $\lim_{t\to\infty}\tilde{M}_{t}(\lozenge[\mu_{2,[k+1]}^{l},\mu_{2,[k+1]}^{h}])=1$
almost surely.

The iterates $(\mu_{2,[k]}^{A})_{k\ge1}$ and $(\mu_{2,[k]}^{B})_{k\ge1}$
are the iterates of a contraction map, so $\lim_{k\to\infty}\mu_{2,[k]}^{A}=\mu_{2}^{\bullet}=\lim_{k\to\infty}\mu_{2,[k]}^{B}$.
Thus, agent's posterior converges in $L^{1}$ to $li(\mu_{2}^{\infty})$
almost surely (since the support of the prior is bounded). In addition,
the sequences of bounds on asymptotic actions also converge by continuity,
$\lim_{k\to\infty}C(\mu_{1}^{\bullet},\mu_{2,[k]}^{A};\gamma)=c^{\infty}=\lim_{k\to\infty}C(\mu_{1}^{\bullet},\mu_{2,[k]}^{B};\gamma)$.
This implies $\lim_{t\to\infty}\tilde{C}_{t}=c^{\infty}$ almost surely.
Finally, combining the asymptotic belief result with Lemma \ref{lem:learning_mu1},
we see that in fact $\tilde{M}_{t}$ converges in $L^{1}$ to the
point $(\mu_{1}^{\bullet},\mu_{2}^{\infty})$ almost surely.
\end{proof}

\end{document}